\newif\ifsubmission

\ifsubmission
\documentclass{llncs}

\else
\documentclass[11pt]{article}
\usepackage{fullpage}
\fi

\ifsubmission \else
\usepackage{palatino}
\fi
\usepackage[normalem]{ulem}
\usepackage{framed}
\usepackage[utf8]{inputenc}
\usepackage[pdfstartview=FitH,pdfpagemode=UseNone,colorlinks,linkcolor=blue,filecolor=blue,citecolor=Violet,urlcolor=red]{hyperref}
\usepackage{cmap}
\usepackage[T1]{fontenc}
\usepackage{multirow}
\usepackage{float}
\usepackage[section]{placeins} 
\usepackage{mathtools} 
\usepackage[dvipsnames]{xcolor}
\usepackage{amsmath,amsthm,amssymb,algpseudocode,algorithm,cryptocode}
\usepackage{mathrsfs}
\usepackage[capitalize]{cleveref}
\usepackage{dashbox}
\usepackage{braket}
\usepackage{physics}
\usepackage{xspace} 
\usepackage{braket}
\usepackage[normalem]{ulem}
\usepackage{soul,xcolor}
\usepackage{tikz}
\usetikzlibrary{quantikz} 
\usepackage{soul}
\usepackage{enumitem}
\usepackage{breakcites}

\newif\ifnotes
\notestrue

\ifsubmission

\newtheorem{axiom}[theorem]{Axiom}
\newtheorem{physicsaxiom}[theorem]{Physics Axiom}
\newtheorem{importedtheorem}[theorem]{Imported Theorem}
\newtheorem{importedlemma}[theorem]{Imported Lemma}
\newtheorem{informaltheorem}[theorem]{Informal Theorem}
\newtheorem{physicstheorem}[theorem]{Physical Theorem}

\newtheorem{claim}[theorem]{Claim}
\newtheorem{subclaim}[theorem]{SubClaim}
\newtheorem{fact}[theorem]{Fact}
\newtheorem{construction}[theorem]{Construction}
\Crefname{importedtheorem}{Imported Theorem}{Imported Theorems}
\Crefname{importedlemma}{Imported Lemma}{Imported Lemma}
\Crefname{theorem}{Theorem}{Theorems}
\Crefname{proposition}{Proposition}{Propositions}
\Crefname{claim}{Claim}{Claims}
\Crefname{subclaim}{SubClaim}{SubClaims}
\Crefname{subsubclaim}{SubSubClaim}{SubSubClaims}
\Crefname{lemma}{Lemma}{Lemmas}
\Crefname{conjecture}{Conjecture}{Conjectures}
\Crefname{corollary}{Corollary}{Corollaries}
\Crefname{construction}{Construction}{Constructions}
\Crefname{property}{Property}{Properties}

\theoremstyle{definition}

\newtheorem{assumption}[theorem]{Assumption}
\newtheorem{notation}[theorem]{Notation}
\Crefname{definition}{Definition}{Definitions}
\Crefname{assumption}{Assumption}{Assumptions}
\Crefname{notation}{Notation}{Notations}

\theoremstyle{remark}

\newtheorem{comment}[theorem]{Comment}

\Crefname{question}{Question}{Questions}
\Crefname{remark}{Remark}{Remarks}
\Crefname{comment}{Comment}{Comments}
\Crefname{fact}{Fact}{Facts}
\Crefname{step}{Step}{Steps}

\else

\newtheorem{theorem}{Theorem}[section]

\newtheorem{claim}[theorem]{Claim}

\newtheorem{lemma}[theorem]{Lemma}

\newtheorem{corollary}[theorem]{Corollary}

\newtheorem{definition}[theorem]{Definition}
\newtheorem{remark}[theorem]{Remark}
\Crefname{importedtheorem}{Imported Theorem}{Imported Theorems}
\Crefname{importedlemma}{Imported Lemma}{Imported Lemmas}
\Crefname{theorem}{Theorem}{Theorems}
\Crefname{proposition}{Proposition}{Propositions}
\Crefname{claim}{Claim}{Claims}
\Crefname{subclaim}{SubClaim}{SubClaims}
\Crefname{subsubclaim}{SubSubClaim}{SubSubClaims}
\Crefname{lemma}{Lemma}{Lemmas}
\Crefname{conjecture}{Conjecture}{Conjectures}
\Crefname{corollary}{Corollary}{Corollaries}
\Crefname{construction}{Construction}{Constructions}
\Crefname{property}{Property}{Properties}

\theoremstyle{definition}

\Crefname{definition}{Definition}{Definitions}
\Crefname{assumption}{Assumption}{Assumptions}
\Crefname{notation}{Notation}{Notations}

\theoremstyle{remark}

\Crefname{question}{Question}{Questions}
\Crefname{comment}{Comment}{Comments}
\Crefname{fact}{Fact}{Facts}
\Crefname{step}{Step}{Steps}

\fi

\newcommand{\secp}{\lambda}

\def\cA{{\cal A}}
\def\cB{{\cal B}}
\def\cC{{\cal C}}
\def\cD{{\cal D}}

\def\cG{{\cal G}}
\def\cH{{\cal H}}
\def\cI{{\cal I}}

\def\cL{{\cal L}}
\def\cM{{\cal M}}

\def\cO{{\cal O}}

\def\cV{{\cal V}}
\def\cW{{\cal W}}
\def\cX{{\cal X}}
\def\cY{{\cal Y}}
\def\cZ{{\cal Z}}

\def\regV{{\cal V}}
\def\regW{{\cal W}}

\def\bbN{{\mathbb N}}

\def\poly{{\rm poly}}

\def\negl{{\rm negl}}

\newcommand{\pk}{\mathsf{pk}}
\newcommand{\sk}{\mathsf{sk}}

\newcommand{\KeyGen}{\mathsf{KeyGen}}

\newcommand{\Com}{\mathsf{Com}}

\newcommand{\Enc}{\mathsf{Enc}}
\newcommand{\Dec}{\mathsf{Dec}}

\newcommand{\ct}{\mathsf{ct}}

\DeclareMathOperator*{\expectation}{\mathbb{E}}
\newcommand{\E}{\expectation}

\newcommand{\Gen}{\mathsf{Gen}}

\newcommand{\vk}{\mathsf{vk}}

\newenvironment{boxfig}[2]{\begin{figure}[#1]\fbox{
    \begin{minipage}{\linewidth}
    \vspace{0.2em}\makebox[0.025\linewidth]{}    \begin{minipage}{0.95\linewidth}{{#2 }}
    \end{minipage}\vspace{0.2em}\end{minipage}}}{\end{figure}}

\newcommand{\pprotocol}[4]{
\begin{boxfig}{h!}{
\begin{center}
\textbf{#1}
\end{center}
    {\small#4}
\vspace{0.2em} } \caption{\label{#3} #2}
\end{boxfig}
}

\newcommand{\protocol}[4]{
\pprotocol{#1}{#2}{#3}{#4} }

\newcommand{\inp}{\mathsf{inp}}

\newcommand{\out}{\mathsf{out}}
\newcommand{\state}{\mathsf{st}}

\newcommand{\Open}{\mathsf{Open}}
\newcommand{\Rec}{\mathsf{Rec}}

\renewcommand{\partial}{\mathsf{partial}}

\newcommand{\Ext}{\mathsf{Ext}}

\newcommand{\TD}{\mathsf{TD}}

\newcommand{\pr}{\mathsf{pr}}

\newcommand{\Adv}{\mathsf{Adv}}

\newcommand{\maj}{\mathsf{maj}}

\newcommand{\msg}{\mathsf{msg}}

\newcommand{\Prove}{\mathsf{Prove}}
\newcommand{\Ver}{\mathsf{Ver}}

\newcommand{\nonnegl}{\mathsf{non}\text{-}\mathsf{negl}}

\newcommand{\NC}{\mathsf{NC}}

\newcommand{\pparam}{\mathsf{pp}}
\newcommand{\sparam}{\mathsf{sp}}
\newcommand{\Invert}{\mathsf{Invert}}
\newcommand{\Claw}{\mathsf{Claw}}
\newcommand{\CHSH}{\mathsf{CHSH}}
\newcommand{\KLVY}{\mathsf{KLVY}}
\newcommand{\Tel}{\mathsf{Tel}}
\newcommand{\Comm}{\mathsf{Commutation}}
\newcommand{\Solve}{\mathsf{Solve}}

\newcommand{\Obligate}{\mathsf{Obligate}}
\newcommand{\OSP}{\mathsf{OSP}}

\newcommand{\PartialInvert}{\mathsf{PartialInvert}}
\newcommand{\PhaseInvert}{\mathsf{PhaseInvert}}
\newcommand{\td}{\mathsf{td}}
\newcommand{\CSG}{\mathsf{CSG}}
\newcommand{\DPCSG}{\mathsf{DBCSG}}
\newcommand{\CNOT}{\mathsf{CNOT}}
\newcommand{\ECNOT}{\mathsf{E}\text{-}\mathsf{CNOT}}
\newcommand{\Apply}{\mathsf{Apply}}
\newcommand{\Sen}{\mathsf{Sen}}
\newcommand{\Bias}{\mathsf{Bias}}

\begin{document}
\setstcolor{red}

\title{On the Power of Oblivious State Preparation}
\author{James Bartusek\thanks{NYU. Email: \texttt{bartusek.james@gmail.com}} \and Dakshita Khurana\thanks{UIUC and NTT Research. Email: \texttt{dakshita@illinois.edu}}}
\date{}
\maketitle

\begin{abstract}
We put forth Oblivious State Preparation (OSP) as a cryptographic primitive that unifies techniques developed in the context of a quantum server interacting with a classical client. 
OSP allows a {\em classical} polynomial-time sender to input a choice of one out of two public observables, and a quantum polynomial-time receiver to recover an eigenstate of the corresponding observable -- while keeping the sender's choice hidden from any malicious receiver.

We obtain the following results:
\begin{itemize}
\item The existence of (plain) trapdoor claw-free functions implies OSP, and the existence of dual-mode trapdoor claw-free functions implies round-optimal (two-round) OSP.
\item OSP implies the existence of proofs of quantumness, test of a qubit, blind classical delegation of quantum computation, and classical verification of quantum computation.
\item Two-round OSP implies quantum money with classical communication, classically-verifiable position verification, and (additionally assuming classical FHE with log-depth decryption) quantum FHE.
\end{itemize}
Thus, the OSP abstraction helps separate the cryptographic layer from the information-theoretic layer when building cryptosystems across classical and quantum participants. Indeed, several of the aforementioned applications were previously only known via tailored LWE-based constructions, whereas our OSP-based constructions yield new results from a wider variety of assumptions, including hard problems on cryptographic group actions. 

Finally, towards understanding the minimal hardness assumptions required to realize OSP, we prove the following:

\begin{itemize}
    \item OSP implies oblivious transfer between one classical and one quantum party.
    \item Two-round OSP implies public-key encryption with classical keys and ciphertexts.
\end{itemize}
In particular, these results help to ''explain'' the use of public-key cryptography in the known approaches to establishing a ''classical leash'' on a quantum server. For example, combined with a result of Austrin et al. (CRYPTO 22), we conclude that \emph{perfectly-correct} OSP cannot exist unconditionally in the (quantum) random oracle model.
\end{abstract}

\newpage

\tableofcontents
\newpage

\section{Introduction}

One of the central concepts driving research in quantum cryptography over the past decade has been that of a ``classical leash'' on quantum systems \cite{10.1145/2422436.2422473}. In other words, how can we enable a classical device, using just classical communication, to exert some element of \emph{control} over a quantum mechanical system? 

In a major conceptual advance from 2018 \cite{Brakerski2018ACT}, the use of (public-key) cryptography was identified as a useful tool for establishing this desired control over a quantum server. There has since been an explosion of results on classical-client quantum-server protocols, ranging from proofs of quantumness under quantum-hard assumptions \cite{Brakerski2018ACT}, certifiable randomness generation \cite{Brakerski2018ACT}, quantum homomorphic encryption with classical ciphertexts \cite{Mahadev2017ClassicalHE}, classical verification of quantum computation \cite{Mahadev2018ClassicalVO}, self-testing a single quantum device \cite{Metger2021selftestingofsingle}, position verification \cite{liu_et_al:LIPIcs.ITCS.2022.100}, secure quantum computation \cite{10.1007/978-3-030-90459-3_1}, and quantum money \cite{10.1145/3318041.3355462,10.1145/3519935.3519952} with classical communication, and proofs of contextuality \cite{ABCC}, among others.

The resounding success of this line of work begs a deeper understanding of the basic principles underlying the paradigm introduced in \cite{Brakerski2018ACT}. For example, \cite{Brakerski2018ACT} based their results on the existence of a fairly ad-hoc and unwieldy cryptographic primitive: a noisy trapdoor claw-free function (TCF) with an adaptive hardcore bit. This primitive has only been shown to exist from the learning with errors (LWE) assumption,\footnote{\cite{10.1007/978-3-031-22318-1_10} has shown how to obtain a weaker variant of the adaptive hardcore bit property from hard problems on cryptographic group actions.} and several followups, including many of the aforementioned results, inherited the use of this primitive. This raises the following (informal) question.

\begin{quote}
\begin{center}
 \emph{Is there a conceptually-simple and easy-to-instantiate primitive that suffices for building powerful classical-client quantum-server applications?}
\end{center}
\end{quote}

We note that some partial progress has been made towards a more ``generic'' approach to constructing some of these end applications. For example, a recent line of work \cite{KCVY,KLVY,NZ23} has yielded classical verification of quantum computation from the assumption of quantum fully-homomorphic encryption (QFHE), and \cite{10.1007/978-3-031-68382-4_8} has shown that QFHE follows from any classical FHE (with decryption in $\text{NC}_1$) plus an appropriate notion of ``dual-mode'' trapdoor claw-free functions (with no need for an adaptive hardcore bit). However, as we show, these approaches can still be generalized much further.

Another aspect of \cite{Brakerski2018ACT}'s approach that demands further investigation is their use of public-key (i.e.\ trapdoor-based) cryptography. This has been justified informally by observing that, because the classical client is computationally weaker than the quantum server, we need to introduce some mechanism for the client to gain the ``upper hand'' on the server. One way to do this is to introduce asymmetric cryptography, allowing the client to send the server a public key while keeping the secret key to themselves. However, as far as we are aware, nothing more than this informal intuition has been proposed in an attempt to address the following fundamental (and again, informal) question. 

\begin{quote}
\begin{center}
 \emph{Is public-key cryptography necessary to establish a classical leash on a quantum server?}
\end{center}
\end{quote}

As an example meant to further illustrate the importance of this question, we note that progress in this direction may shed more light on the recent breakthrough techniques of Yamakawa-Zhandry \cite{10.1145/3658665} for establishing quantum advantage. Indeed, while they show that proofs of quantumness exist in the random oracle model, which can be heuristically instantiated using a cryptographic hash function (i.e.\ \emph{symmetric} cryptography), their techniques have so far resisted attempts at constructing, say, a test of a qubit, verifiable delegation of quantum computation, or other ``classical leash''-style primitives mentioned above. Establishing the necessity of public-key cryptography for these primitives would explain this gap. While we do not completely close this question, our results do show that the principles underlying current approaches to classical-leash primitives (inspired by~\cite{Brakerski2018ACT,KCVY}, etc.) also yield public-key style primitives. We provide further details on these results later in the introduction.

\subsection{Oblivious state preparation}

Aiming to make progress on these two motivating questions, we put forth the idea of \emph{Oblivious State Preparation} (OSP) as a unifying cryptographic primitive in the realm of classical-client quantum-server protocols. 

OSP is simple to describe. It is a protocol that takes place between a classical sender and a quantum receiver. The classical sender has as input a bit $b \in \{0,1\}$ which specifies a choice of one of two public observables. We usually take one to be $Z$ and the other to be $X$, but in principle they could be arbitrary. At the end of the protocol, the receiver outputs a quantum state. We require two properties.

\begin{itemize}
    \item \textbf{Correctness:} If the receiver is honest, their output is an eigenstate of the observable chosen by the sender, and the sender receives a description of this state. In the usual ``standard'' case, this means that when $b=0$, the receiver outputs either $\ket{0}$ or $\ket{1}$, and when $b=1$, the receiver outputs either $\ket{+}$ or $\ket{-}$.
    \item \textbf{Security:} Any quantum polynomial-time (QPT) malicious receiver has $\negl(\secp)$ advantage in guessing the sender's input bit $b$. 
\end{itemize}

OSP highlights an inherent cryptographic property of quantum information arising from the uncertainty principle. That is, it is not necessarily possible to determine the basis of a given state, even though this basis information is well-defined and fixed by the description of the state. Indeed, given the resource of quantum communication, information-theoretically secure OSP is trivial. In this work, however, OSP always refers to the classical-communication case. At a high level, we ask, (1) what cryptography is necessary to obtain OSP, i.e.\ the ability for a classical client to, roughly speaking, set up an instance of the uncertainty principle on a quantum server, and (2) what are the applications of this ability.

In fact, the idea of OSP for the $Z$ and $X$ observables as described above has been previously proposed by \cite{10.1007/978-3-030-34578-5_22} under the name ``malicious 4-states QFactory with basis-blindness.'' They showed that such a protocol can be obtained from a TCF with a particular type of ``homomorphic, hardcore predicate''. However, in order to derive most applications, they required an additional and conjectured \emph{verifiability} property (see also \cite{Cojocaru_2021}). In this work, we give OSP a more general treatment as a primitive, show that it is possible to relax the assumptions under which it can be built, and vastly expand its set of applications.

\subsection{Results}

\subsubsection{Constructions} As mentioned above, researchers beginning with \cite{Brakerski2018ACT} identified the usefulness of (variants of) trapdoor claw-free functions (TCFs) in realizing applications such as a test of a qubit, quantum fully-homomorphic encryption, and classical verification of quantum computation. While several variants of TCFs have appeared over the years, e.g.\ extended, dual-mode, with adaptive hardcore bit, etc., we show that perhaps the most stripped down notion of a TCF suffices to build OSP.

We define a (plain) TCF as a family of functions $f$ that can be sampled along with a trapdoor $\td$. The guarantee, roughly, is that there is a (QPT preparable) distribution $\cD$ over inputs such that with some inverse polynomial probability over $x \gets \cD$, $x$ has exactly one sibling $x'$ (similarly weighted by $\cD$) such that $f(x) = f(x')$, and moreover, both $x$ and $x'$ can be recovered given $f(x)$ and $\td$. Finally, claw-freeness demands that no QPT adversary can recover any such ``claw'' $x$ and $x'$ given only the description of $f$. Building on techniques from \cite{10.1007/978-3-031-38554-4_6}, we show the following.

\begin{theorem}[Informal]
    (Plain) TCFs imply OSP.
\end{theorem}

Our construction of OSP from plain TCFs requires multiple of rounds of interaction. However, a desirable feature for some applications is limited interaction. The best we can hope for is \emph{two-round} OSP, i.e.\ one message from the sender followed by one from the receiver.\footnote{No one-message OSP can be secure. Indeed, the message from sender to receiver would have to fix the description of the desired state $H^b\ket{s}$, and correctness would imply that the receiver could then generate multiple copies of this state, eventually enough to determine the basis with near certainty.} Adapting techniques from \cite{10.1007/978-3-031-68382-4_8}, we show that TCFs with an additional \emph{dual-mode} property imply such two-round OSP. 

\begin{theorem}[Informal]
    Dual-mode TCFs imply two-round OSP.
\end{theorem}

Briefly, a dual-mode TCF is similar to a plain TCF except that the function may be sampled in an ``injective'' mode, where there are no collisions, and it is computationally difficult to distinguish this from the normal ``lossy'' mode. We note that  dual-mode (and thus plain) TCFs are known from LWE \cite{Brakerski2018ACT} and from the ``extended linear hidden shift'' assumption on cryptographic group actions \cite{10.1007/978-3-031-22318-1_10,10.1007/978-3-031-68382-4_8}.

In what follows, we will highlight the power of OSP by establishing several classical-client quantum-server applications, as well as cryptographic implications. We will build everything from ``standard'' OSP, where the two observables are $Z$ and $X$. However, it is meaningful to consider OSP for \emph{any} pair of non-commuting observables, in particular pairs that may not be maximally anti-commuting. We begin to explore the landscape of OSP as a more general primitive, and in particular establish the following result.

\begin{theorem}[Informal]
    OSP for any pair of two-outcome observables that are at a $1/\poly$ angle implies standard ($Z$ and $X$) OSP.
\end{theorem}

\subsubsection{Applications}

\paragraph{From OSP.} We show that OSP is sufficient to obtain several classical-client quantum-server protocols of interest.

\begin{theorem}[Informal]
    OSP implies proofs of quantumness, a test of a qubit, blind classical delegation of quantum computation, and classical verification of quantum computation.
\end{theorem}

A few remarks on these results are in order. The proof of quantumness from OSP can be seen as a modular instantiation of the template proposed by \cite{KCVY} and later tweaked by \cite{10.1007/978-3-031-38554-4_6} and \cite{ABCC}. In particular, we show that a \emph{single} instance of OSP suffices to implement a ``computational Bell test'' between the client and server based on the CHSH game.

Blind classical delegation of quantum computation allows a classical client to outsource a quantum computation of its choice to a quantum server without leaking anything about the actual description of the computation. In a major breakthrough, \cite{Mahadev2017ClassicalHE} gave the first construction of this primitive by building \emph{quantum fully-homomorphic encryption} (QFHE) from LWE. This approach of course relies on techniques that at least imply classical fully-homomorphic encryption, and, to the best of our knowledge, fully-homomorphic encryption has remained the only solution to blind classical delegation of quantum computation that has been made explicit in the literature.\footnote{\cite{Cojocaru_2021} show how to realize blind classical delegation of quantum computation but only against an \emph{honest-but-curious} server, using their protocol for ``pseudo-secret random qubit generator.''} While perhaps folklore, we formalize the fact that classical FHE is not required for blind classical delegation of quantum computation, showing that OSP suffices.

In another major breakthrough, \cite{Mahadev2018ClassicalVO} designed a protocol that allows a classical client to verifiably delegate a BQP computation to a potentially cheating quantum server. Since then, classical verification of quantum computation (CVQC) has remained a central primitive of study in quantum cryptography \cite{9996633,10.1007/978-3-031-15979-4_7,10.1007/978-3-031-07082-2_25,NZ23,gunn2024classicalcommitmentsquantumstates,metger2024succinctargumentsqmastandard}. However, until now, all known constructions relied on the hardness of LWE. As a corollary of our result, we show that CVQC follows from any (plain) TCF, and thus from hard problems on cryptographic group actions. In a nutshell, we formalize the fact that the recent approach of \cite{KLVY,NZ23} establishing CVQC from QFHE can in fact be instantiated from any (potentially interactive, non-compact) blind classical delegation of quantum computation protocol, and thus, from any OSP.

\paragraph{From two-round OSP.} Our next batch of results makes use of \emph{two-round} OSP.

\begin{theorem}[Informal]
    Two-round OSP implies (privately-verifiable) quantum money with classical communication, position verification with classical communication, and (assuming classical FHE with decryption in $\text{NC}_1$) QFHE.
\end{theorem}

Briefly, the first two results go via the intermediate primitive of a ``1-of-2 puzzle'' \cite{10.1145/3318041.3355462}, which we build from OSP using similar techniques to our CHSH-based proof of quantumness. Then, we appeal to \cite{10.1145/3318041.3355462}, who showed that 1-of-2 puzzles imply privately-verifiable quantum money with classical communication, and \cite{liu_et_al:LIPIcs.ITCS.2022.100}, who showed that 1-of-2 puzzles imply position verification with classical communication. Prior to our work, the only construction of 1-of-2 puzzles, due to \cite{10.1145/3318041.3355462}, relied specifically on LWE (via TCFs with the adaptive hardcore bit property).\footnote{We note that a recent concurrent and independent work has shown how to construct classically-verifiable position verification from certified randomness protocols \cite{cryptoeprint:2024/1726}.}

The third result on QFHE follows by adapting the recent techniques of \cite{10.1007/978-3-031-68382-4_8}, who showed how to construct QFHE from any classical FHE (with decryption in $\NC_1$) and dual-mode TCFs. We observe that, in fact, \emph{any} two-round OSP suffices in place of the dual-mode TCF.

\paragraph{Discussion.} Before moving on, it is worth pointing out that each of these results individually do not require deep new techniques. Rather, they mostly follow by adapting, modularizing, and generalizing existing approaches in the literature. However, in our view, the  identification of OSP as a simple primitive that yields all of these applications is useful, both as a pedagogical tool and to enhance future research.  OSP abstracts away the cryptographic essence of major protocols in the area, allowing for easy, information-theoretic design of protocols that call an underlying OSP functionality. This separates the ``cryptographic layer'' from the ``information-theoretic layer'' in the design of these protocols. In some cases, this modular approach also allows us to broaden the set of assumptions under which these applications are known to exist.

\subsubsection{Implications}

Finally, given the broad reach of OSP, we seek to understand the cryptography necessary in order to realize it. Our results on this are summarized as follows.

\begin{theorem}[Informal]
    OSP implies commitments and oblivious transfer (OT) with classical communication (where one party is completely classical), while two-round OSP implies public-key encryption.
\end{theorem}


So what can we conclude about OSP from these results? As mentioned earlier, one of the original motivations for our work was to ``justify'' the use of public-key cryptography in the recent line of work aimed at establishing a classical leash on quantum systems. Progress towards this goal can be appreciated by noting that the techniques introduced in \cite{Brakerski2018ACT,Mahadev2017ClassicalHE,Mahadev2018ClassicalVO,KCVY} all at the very least provide some way to perform an OSP between the classical client and quantum server, and thus, by our result, also provide a way to build OT with classical communication.

So far, the community does not have any approach for building OT with classical communication from minicrypt assumptions, or even from arbitrary trapdoor functions. Thus, our result helps explain why the known constructions of OSP require TCFs, which are more structured than even injective trapdoor functions. In fact, in the classical setting, we actually have an oracle separation between OT and injective trapdoor functions \cite{10.5555/795666.796557}.

However, it remains an open question to give such strong oracle separations in the quantum setting. Progress came when \cite{10.1007/978-3-031-15979-4_6} showed that \emph{perfectly correct} key agreement between one classical and one quantum party does not exist in the quantum random oracle model. This notion of key agreement is implied by the perfectly correct variant of our notion of OT between one classical and one quantum party, and thus, we obtain the following corollary.

\begin{corollary}[Informal]
    Perfectly correct OSP does not exist in the quantum random oracle model.
\end{corollary}

This leaves a sliver of possibility that \emph{non-perfectly-correct} OSP can yet be constructed without public-key assumptions. However, we have established that if one can build OSP from minicrypt primitives (say, by adapting the techniques of \cite{10.1145/3658665}), or even from arbitrary trapdoor functions, then this would also represent a major breakthrough in cryptography more generally - a construction of classical-communication OT from new assumptions.

\section{Technical Overview}

\subsection{Realizing oblivious state preparation}\label{susbec:overview-construction}

We show that OSP follows from any trapdoor claw-free function (TCF), i.e.\ we don't require an additional adaptive hardcore bit or dual-mode property. Our presentation abstracts out the key role of the TCF, which is to generate a claw-state correlation. That is, we define a ``claw-state generator'' (CSG) as any protocol between a classical sender and quantum receiver that outputs $\frac{1}{\sqrt{2}}(\ket{0,x_0} + \ket{1,x_1})$ to the receiver and $x_0,x_1 \in \{0,1\}^n$ to the sender.\footnote{Technically, we call this a \emph{differentiated-bit} CSG, since the first qubit holds a bit that differentiates the two members of the claw state. It is easy to show that one can generically add the differentiated-bit property to any CSG, while maintaining search security (see \cref{lemma:DPCSG}).} We say that the protocol has \emph{search security} if no QPT receiver can output both $(x_0,x_1)$ except with negligible probability.

Now, we show that CSG with search security implies OSP by relying on a sub-protocol from \cite{10.1007/978-3-031-38554-4_6}.\footnote{Later, we will actually show that OSP implies CSG (with an even stronger security property called \emph{indistinguishability} security), meaning that these notions are in fact equivalent.} After the CSG is performed, the sender chooses two random strings $r_0,r_1 \gets \{0,1\}^n$ and sends them to the receiver. The receiver then maps 

\[\frac{1}{\sqrt{2}}\left(\ket{x_0} + \ket{x_1}\right) \to \frac{1}{\sqrt{2}}\left(\ket{x_0}\ket{x_0 \cdot r_0} + \ket{x_1}\ket{x_1 \cdot r_1}\right),\] and measures all but the last qubit in the Hadamard basis to obtain a string $d$, which it returns to the sender.

It is easy to check that if $x_0 \cdot r_0 = x_1 \cdot r_1$, then the receiver obtain a standard basis eigenstate, while if $x_0 \cdot r_0 \neq x_1 \cdot r_1$, then the receiver obtains a Hadamard basis eigenstate. Moreover, which eigenstate obtained can be computed by the sender, who knows $x_0,x_1$, and $d$. To see why this is secure, note that the bit that determines the basis is equal to $x_0 \cdot r_0 \oplus x_1 \cdot r_1 = (x_0,x_1) 
\cdot (r_0,r_1)$. Thus, by Goldreich-Levin, any adversarial receiver that can predict the basis of their received state can be used to extract an entire claw $(x_0,x_1)$, which breaks the search security of the CSG.

Now, while the protocol above implements {\em random-input} OSP where the sender ends up with a random choice of basis, it is generically possible to reorient this into a chosen-input OSP, which we show in \cref{lemma:random-input}. Finally, we note that in the body, we also construct OSP with the optimal round-complexity of two (i.e.\ one message from the sender followed by one from the receiver), by assuming a slightly stronger variant of TCFs, namely \emph{dual-mode} TCFs. This construction adapts the recent techniques of \cite{10.1007/978-3-031-68382-4_8}, and we refer the reader to \cref{subsec:OSP-dTCF} for details.

\subsection{Applications}

\subsubsection{Proofs of quantumness}

Our first use case for OSP is to instantiate a ``computational Bell test'', first introduced by \cite{KCVY}. The resulting protocol is essentially a generalized presentation of \cite{ABCC}'s recent proof of quantumness protocol, in which they instantiated the OSP using an ``encrypted CNOT'' operation based on a structured type of TCF. 

Consider the server's state at the end of an OSP protocol: If the client's input was $a=0$, they obtain $\ket{x}$ for some bit $x$, and if the client's input was $a=1$, they obtain $H\ket{x}$ for some bit $x$. This is \emph{exactly} the same as Bob's state in the CHSH game once (honest) Alice provides an answer $x$ on input question $a$. Indeed, in the CHSH game, Alice and Bob initially share an EPR pair, and Alice measures her half in the standard basis if $a = 0$ and in the Hadamard basis in $a=1$. This suggests the following proof of quantumness protocol.

\begin{itemize}
    \item The verifier samples $a \gets \{0,1\}$ and performs an OSP with the prover in order to deliver $H^a\ket{x}$.
    \item The parties ``complete'' the CHSH game as follows. The verifier samples $b \gets \{0,1\}$ and sends it to the prover. If $b = 0$, the prover measures their state in the $X+Z$ basis to obtain $y$, and if $b=1$, the prover measures their state in the $X-Z$ basis to obtain $y$, and sends $y$ back to the verifier.\footnote{To be concrete, these bases are defined as $X+Z = \left\{\cos(\pi/8)\ket{0}+\sin(\pi/8)\ket{1}, -\sin(\pi/8)\ket{0} + \cos(\pi/8)\ket{1}\right\}$, and $X-Z = \left\{\cos(-\pi/8)\ket{0}+\sin(-\pi/8)\ket{1}, -\sin(-\pi/8)\ket{0} + \cos(-\pi/8)\ket{1}\right\}.$}
    \item The verifier runs the CHSH verification predicate, accepting if $x \oplus y = a \cdot b$.
\end{itemize}

By a standard analysis of the CHSH game, an honest QPT prover following the strategy outlined above makes the verifier accept with probability $\cos^2(\pi/8) > 0.85$. Now, consider any classical polynomial-time prover. Note that an equivalent way to write the verification predicate is to accept if $y = x \oplus a \cdot b$. Given any classical prover that wins with probability $3/4 + 1/\poly$, we can rewind them to extract answers on both $b=0$ and $b=1$ that simultaneously accept with probability at least $1/2 + 1/\poly$. That is, we can obtain $y_0 = x$ and $y_1 = x \oplus a$ with probability $1/2 + 1/\poly$. However, this contradicts the security of the OSP, since $y_0 \oplus y_1 = a$, and OSP demands that no polynomial-time adversary has noticeable advantage in guessing the basis choice $a$.

\subsubsection{1-of-2 puzzles}

Next, we extend the above ideas to realize more applications, via the intermediate primitive of 1-of-2 puzzles. Introduced by \cite{10.1145/3318041.3355462}, a 1-of-2 puzzle consists of four algorithms defined as follows.

\begin{itemize}
        \item $\KeyGen(1^\secp) \to (\pk,\vk)$. The PPT key generation algorithm takes as input the security parameter $1^\secp$ and outputs a public key $\pk$ and a secret verification key $\vk$.
        \item $\Obligate(\pk) \to (\ket{\psi},y)$: The QPT obligate algorithm takes as input the public key, and outputs a classical obligation string $y$ and a quantum state $\ket{\psi}$.
        \item $\Solve(\ket{\psi},b) \to a$: The QPT solve algorithm takes as input a state $\ket{\psi}$ and a bit $b \in \{0,1\}$ and outputs a string $a$.
        \item $\Ver(\vk,y,b,a) \to \{\top,\bot\}$: The PPT verify algorithm takes as input the verification key $\vk$, a string $y$, a bit $b \in \{0,1\}$, and a string $a$, and either accepts or rejects.
\end{itemize}

Correctness requires that on either challenge $b \in \{0,1\}$, the $\Solve$ algorithm produces an accepting answer $a$, while security stipulates that no QPT adversary can \emph{simultaneously} produce an accepting answer $a_0$ on challenge $b=0$ and an accepting answer $a_1$ on challenge $b=1$. This is clearly an inherently quantum primitive that has been shown by prior work to imply both (privately-verifiable) quantum money with classical communication \cite{10.1145/3318041.3355462}, and position verification with classical communication \cite{liu_et_al:LIPIcs.ITCS.2022.100}.

Our idea to obtain a 1-of-2 puzzle is as follows. The $\KeyGen$ and $\Obligate$ algorithms will run $\secp$ parallel instances of two-round OSP on a uniformly random sender's bit $r \gets \{0,1\}$.\footnote{We note that it is possible to define a more interactive version of 1-of-2 puzzles and write down a candidate from OSP rather than two-round OSP. However, for security we rely on an amplification lemma from \cite{10.1145/3318041.3355462} that crucially uses the two-round setup, so we leave an exploration of this generalization to future work.} This results in a state $(H^r)^{\otimes \secp}\ket{s}$, where $s$ is a $\secp$-bit string. Now, challenge $b=0$ asks for a string that matches $s$ on at least $0.85$ fraction of indices, while challenge $b=1$ asks for a string that matches $s \oplus (r,\dots,r)$ on at least $0.85$ fraction of indices. Roughly, this is a $\secp$-parallel repetition of the above proof of quantumness, all using the same verifier bit $r$. Thus, correctness follows from measuring all states in the $X+Z$ basis when $b=0$, and in the $X-Z$ basis when $b=1$ (and a tail bound). We can also establish a weak form of security. Suppose a (quantum) adversary has a $1/2+1/\poly$ probability of passing both challenges simultaneously. Then, by XORing its answers and taking the majority bit, we have that with $1/2+1/\poly$, this adversary can be used to predict the bit $r$, a contradiction to the security of the two-round OSP. Finally, to obtain a full-fledged 1-of-2 puzzle (i.e.\ with negligible security), we appeal to an amplification lemma of \cite{10.1145/3318041.3355462} that is itself based on the parallel repetition for weakly verifiable puzzles of \cite{TCC:CanHalSte05}.

\subsubsection{Blind delegation}

Next, we show that OSP is sufficient to obtain blind classical delegation of any quantum computation. While there are likely several routes to showing this, our approach essentially instantiates the protocol of \cite{doi:10.1139/cjp-2015-0030} using OSP in place of quantum communication from the client to the server.

In full generality, our definition of blind delegation allows the classical client to delegate the computation of some publicly-known quantum operation $Q$ that takes as input a private classical string $x$ from the client and a quantum state on register $\cV$ from the server. At the end of the protocol, the prover recovers the output $Q(x,\cV)$ up to a quantum one-time pad $X^r Z^s$ with keys $(r,s)$ known to the client. Note that this implies the ability to deliver to the client a classical output, by having the prover measure the output register and deliver the result to the client, which will be correct up to a \emph{classical} one-time pad defined by $r$. 

We show that OSP implies this notion as follows.\footnote{It is also easy to see that this definition implies OSP as well, meaning OSP is both necessary and sufficient.} First, we write the circuit $Q$ as a sequence of alternating Clifford operations and $T^\dagger$ gates $Q = C_{\ell+1}T^\dagger C_\ell \dots C_2 T^\dagger C_1$, where $T^\dagger$ is the $\pi/4$ rotation clockwise around the XY plane. To begin the protocol, the client sends $x \oplus r_\inp$, where $x$ is their input and $r_\inp$ is a classical one-time pad. As is typically the case, Clifford operations are straightforward: the server can apply them directly to the current state of the system, and the client can perform a corresponding update to their one-time pad keys. On the other hand, each time we come to a $T^\dagger$ gate, we will use one instance of OSP (and thus some interaction).

The idea is that $T^\dagger X^r Z^s \ket{\psi} = (P^\dagger)^r X^r Z^s T^\dagger\ket{\psi}$, and so applying the $T^\dagger$ reduces to what we call an \emph{encrypted phase gate}. That is, the client holds a private bit $r$, the server holds a single-qubit quantum state $\ket{\psi}$,\footnote{In general, this could be entangled with the rest of their system, but we suppress this in the overview to avoid clutter.} and we want the server to obtain $P^r\ket{\psi}$ after interaction, potentially up to some Pauli error (concretely, our protocol will result in $Z^mP^r\ket{\psi}$, where $m$ is known to the client). 

To enable this, the client first uses OSP to transmit a fresh state $Z^sP^r\ket{+}$ to the server, where $s$ is known to the client. That is, the client and server engage in an OSP for the $X$ and $Y$ observables, which follows by using ``standard'' OSP for the $Z$ and $X$ observables and then having the server apply the appropriate (public) rotation. Now, a straightforward calculation confirms that if the server applies CNOT from their state $\ket{\psi}$ onto $Z^sP^r\ket{+}$ and then measures the second register in the standard basis, the remaining state on the first register will be exactly $P^r\ket{\psi}$ up to some Pauli $Z$ error.

Security of the delegation protocol is immediate from security of the OSP. One by one, we can switch the client's input to each of the OSP instances to 0. Once this is done, the only relevant information the server obtains from running the protocol is $x \oplus r_\inp$, which perfectly hides the client's private input $x$.

\subsubsection{Verifiable delegation}

A protocol for classical verification of arbitrary BQP computation (CVQC) was first shown by \cite{Mahadev2018ClassicalVO}, albeit from a specific type of TCF known only from LWE. Recent work \cite{KLVY,NZ23,metger2024succinctargumentsqmastandard} has explored a different approach that has given us CVQC from the more generic assumption of QFHE. However, QFHE is itself a strong primitive, and is only known from lattices. In this work, we further generalize the approach, showing that it can be instantiated with any blind classical delegation of quantum computation protocol in place of the QFHE, and thus from any OSP. 


The starting point for this approach is the ``KLVY compiler'', which uses QFHE to compile any classical-verifier game sound against two non-communicating (but potentially entangled) provers into a single prover game, with the hope that the semantic security of the QFHE will give soundness against any QPT prover. Given any two-prover game, the compiled protocol is defined as follows.
\begin{itemize}
    \item The verifier sends a QFHE encryption $\Enc(x)$ of Alice's question to the prover.
    \item The prover runs Alice's strategy under the QFHE to obtain an encryption of her answer $\Enc(a)$, along with some auxiliary quantum state. The prover sends $\Enc(a)$ to the verifier.
    \item The verifier sends Bob's question $y$ in the clear.
    \item The prover uses its auxiliary state and $y$ to obtain Bob's answer $b$, which it sends to the verifier.
    \item The verifier decrypts $\Enc(a)$ to obtain $a$ and applies its verification predicate to $(a,b)$.
\end{itemize}

Intuitively, QFHE is used to enforce the non-communicating assumption \emph{computationally}. That is, the semantic security of QFHE implies that the first prover operation (Alice) cannot transmit any information about her question $a$ to the second prover operation (Bob) that can be efficiently recovered. While we don't have a general theorem establishing optimal soundness of the KLVY compiler\footnote{Here, we mean soundness against \emph{quantum} polynomial-time provers. \cite{KLVY} showed a general theorem establishing soundness against \emph{classical} polynomial-time provers.} for \emph{any} game (see \cite{10.1007/978-3-031-38554-4_6, cui2024computationaltsirelsonstheoremvalue, kulpe2024boundquantumvaluecompiled} for progress in this direction), \cite{NZ23} presented a two-player game for arbitrary BQP computation and showed its soundness under the KLVY compiler, giving CVQC from QFHE as a corollary.

In this work, we begin by defining what we call the \emph{generalized} KLVY compiler, which replaces the first round above with any (potentially interactive, non-compact) blind classical delegation of quantum computation protocol. Then, in order to simplify the task of proving soundness of the generalized KLVY compiler, we define a clean class of strategies for two-prover games, which we call \emph{computationally non-local strategies} (\cref{def:comp-non-local}). While the standard notion of a non-local strategy requires that Alice's operation $\{A^x\}_x$ and Bob's operation $\{B^y\}_y$ (where both are written as sets of strategies parameterized by their question) must be applied to disjoint Hilbert spaces, say $\cH_\cA$ and $\cH_\cB$, our notion relaxes this requirement as follows. It includes any strategy $\{A^x\}_x$ on $\cH_\cA \otimes \cH_\cB$ followed by $\{B^y\}_y$ on $\cH_\cB$ such that no QPT distinguisher given the state on register $\cB$ output by $A^x$ can guess $x$ with noticeable advantage. 

We next prove a theorem (\cref{thm:KLVY}) showing that any upper bound on the value of a two-prover game against computationally non-local strategies is also an upper bound on the soundness of the generalized KLVY-compiled game. While straightforward to show, this theorem is quite useful. It allows one to forget all of the underlying details of the cryptographic component when attempting to prove the soundness of a (generalized) KLVY-compiled protocol. Indeed, prior work (e.g.\ \cite{KLVY,NZ23,metger2024succinctargumentsqmastandard}) carried around clunky notation specific to QFHE including public / secret key pairs, ciphertexts, etc., when upper-bounding the soundness of their compiled non-local games. 

Here, we re-visit \cite{NZ23}'s proof strategy, showing that it in fact establishes an upper bound on \emph{any} computationally non-local strategy (i.e.\ not only ones that arise from the use of QFHE). In particular, there are only a handful of places where QFHE is used in their proof, and each time it is only used to show that Bob cannot distinguish between two different Alice questions with better than negligble advantage. Thus, we conclude that CVQC follows generically from any blind classical delegation of quantum computation protocol, and thus from any OSP.

\subsubsection{Encrypted CNOT and applications}

Next, we re-visit a notion that was informally introduced in the influential work of \cite{Mahadev2017ClassicalHE}, called ``encrypted CNOT''. In \cite{Mahadev2017ClassicalHE}, encrypted CNOT was built assuming TCFs with a particular structural requirement on the claws, and it was incorporated into their construction of quantum FHE (in a non-black-box way). Here, we define encrypted CNOT formally as a special case of blind classical delegation of quantum computation, which in particular implies that it follows from any OSP. However, we go a step further, and provide a simple and direct construction of encrypted CNOT, which in particular shows that if we start with a \emph{two-round} OSP, then we obtain a two-round encrypted CNOT.

To be precise, we define encrypted CNOT as a protocol that takes a private bit $b$ from the client and a two-qubit state from the server, and, if $b=0$, does nothing to the server's state, while if $b=1$, applies a CNOT to the server's state. The output state will only be correct up to a quantum one-time pad that must be known to the client. To implement this from OSP, we operate in two steps. Suppose for simplicity that the server initially holds a two qubit state of the form \[(\alpha_0\ket{0} + \alpha_1\ket{1}) \otimes (\beta_0\ket{0} + \beta_1\ket{1}).\]

\begin{enumerate}
    \item First, depending on the bit $b$, either entangle the first qubit with a fresh register, or not. This can be accomplished using OSP as follows. Execute two instances of OSP, where if $b=0$, the client inputs are $(0,1)$ while if $b=1$, the client inputs are $(1,0)$. Thus, up to a one-time pad, the server's state can now be written as
    \begin{align*}
        &(\alpha_0\ket{0} + \alpha_1\ket{1}) \otimes \ket{0} \otimes \ket{+} \otimes (\beta_0\ket{0} + \beta_1\ket{1})  \ \ \ \ \text{if } b = 0 \\
        &(\alpha_0\ket{0} + \alpha_1\ket{1}) \otimes \ket{+} \otimes \ket{0} \otimes (\beta_0\ket{0} + \beta_1\ket{1})  \ \ \ \ \text{if } b = 1.
    \end{align*}
    The server then applies a CNOT from the 1st to the 3rd qubit and a CNOT from the 2nd to the 3rd qubit. In the first case, where the 3rd qubit is $\ket{+}$, this has no effect, while in the second case, this entangles the 1st and 2nd qubit. Then, after measuring the 3rd qubit in the standard basis, the server's state becomes (up to a one-time pad)
    \begin{align*}
        &(\alpha_0\ket{0} + \alpha_1\ket{1}) \otimes \ket{0} \otimes (\beta_0\ket{0} + \beta_1\ket{1}) \ \ \ \ \text{if } b = 0 \\
        &(\alpha_0\ket{00} + \alpha_1\ket{11}) \otimes (\beta_0\ket{0} + \beta_1\ket{1}) \ \ \ \ \text{if } b = 1.
    \end{align*}
    \item Next, apply CNOT from the 2nd qubit to the 3rd qubit, then ``delete'' the 2nd qubit by measuring it in the Hadamard basis. Clearly, in the $b=0$ case this again has no effect, while in the $b=1$ case, this accomplishes a CNOT between the two input qubits. In particular, it can be confirmed that after this step, the server's state becomes (up to a one-time pad)
     \begin{align*}
        &(\alpha_0\ket{0} + \alpha_1\ket{1}) \otimes (\beta_0\ket{0} + \beta_1\ket{1}) \ \ \ \ \text{if } b = 0 \\
        &(\alpha_0\beta_0\ket{00} + \alpha_0\beta_1\ket{1} + \alpha_1\beta_0\ket{11}+ \alpha_1\beta_1\ket{10}) \ \ \ \ \text{if } b = 1.
    \end{align*}
\end{enumerate}

For details (in particular, how the client recovers the one-time pad keys from the OSP information and the server's measurement results), refer to \cref{subsec:encrypted-CNOT}. Here, we mention the applications we obtain from our encrypted CNOT protocol.

\paragraph{Quantum fully-homomorphic encryption.} The first quantum fully-homomorphic encryption (QFHE) scheme was constructed by \cite{Mahadev2017ClassicalHE}. As alluded to above, the construction combines a particular encrypted CNOT protocol with a particular classical FHE protocol in a non-black-box manner in order to achieve QFHE. A recent work of \cite{10.1007/978-3-031-68382-4_8} pioneered a \emph{generic} approach to QFHE from \emph{any} classical FHE (with log-depth decryption) and \emph{any} dual-mode TCF. In \cref{subsec:encrypted-CNOT}, we observe that their work can in fact be seen as constructing QFHE from any (two-round) \emph{encrypted CNOT} protocol plus classical FHE (with log-depth decryption), and thus we establish that QFHE follows from two-round OSP plus classical FHE (with log-depth decryption).

\paragraph{Claw-state generators with indistinguishability security.} Recall the notion of a claw-state generator (CSG) introduced above, which is a protocol that delivers a state $\frac{1}{\sqrt{2}}(\ket{0,x_0} + (-1)^z\ket{1,x_1})$ to a quantum receiver and strings $(x_0,x_1,z)$ to a classical sender.\footnote{In full generality, we allow the claw-state to have a phase specified by the bit $z$, as long as this bit is known to the sender.} We say that such a protocol has \emph{indistinguishability} security if for all $i \in [n]$, no QPT server can predict the bit $x_{0,i} \oplus x_{1,i}$ with better than $\negl(\secp)$ advantage. 

In \cref{subsec:encrypted-CNOT}, we show that (two-round) encrypted CNOT can be used to obtain a (two-round) CSG with indistinguishability security. The construction is straightforward: the sender begins by sampling a string $\Delta \gets \{0,1\}^n$, and the receiver initializes the state $\ket{+}_\cB \otimes \ket{0}_{\cC_1}\otimes \dots \otimes \ket{0}_{\cC_n}$. Then, the parties engage in $n$ encrypted CNOTs, where the $i$'th protocol takes input $\Delta_i$ from the sender, and applies a CNOT from register $\cB$ to register $\cC_i$. It can be confirmed that the receiver ends up with a state of the form \[\frac{1}{\sqrt{2}}\left(\ket{0,x} + (-1)^z\ket{0,x + \Delta}\right),\] where $x \in \{0,1\}^n$ and $z \in \{0,1\}$ can be recovered by the sender. Crucially, note that XOR of the two members of the claw is equal to $\Delta$, and thus, breaking indistinguishability security of the CSG yields an attack on the encrypted CNOT protocol. 

Combined with our construction of OSP from \cref{susbec:overview-construction}, this shows that OSP and CSG (with either search or indistinguishability security) are \emph{equivalent}. Moreover, the notion of a CSG with indistinguishability security will be central to our implications in the next section establishing cryptographic lower bounds for constructing OSP.

\subsection{Implications}\label{subsec:overview-implications}
Our next goal is to understand what cryptographic hardness is necessary for OSP. Towards addressing this, we show that OSP implies commitments and oblivious transfer (where one participant only requires classical capabilities), and that \emph{two-round} OSP implies public-key encryption.

\subsubsection{Commitments}
Our commitment scheme proceeds as follows: the classical committer acts as the sender in $\secp$ executions of an OSP with chosen input basis $b$ in all executions. The committer thus obtains bits $s_1, \ldots, s_\secp$, while an (honest) receiver ends up with $H^b\ket{s_1}, \ldots, H^b\ket{s_\secp}$.

In the decommit phase, the committer reveals $s_1, \ldots, s_\secp$ along with its comitted bit $b$, and the receiver accepts iff for every $i \in [\secp]$, the projection of its $i^{th}$ qubit onto $H^b\ketbra{s_i}H^b$ accepts. 

The hiding of this commitment against a malicious receiver follows from the fact that OSP hides the sender input $b$ from an arbitrary (malicious) receiver, together with a straightforward hybrid argument.

To see why this satisfies statistical (sum) binding, consider the state $\ket{\psi}$ that an honest receiver ends up with after interacting with an arbitrary malicious committer. 
For any fixing of this state $\ket{\psi}$, the probability that a decommitment to $(0,s_0)$ is accepted is $\pr_{0,s_0} = \| \bra{s_0}\ket{\psi} \|^2$, and a decommitment to $(1,s_1)$ is accepted is $\pr_{1,s_1} = \| \bra{s_1}H^{\otimes \secp}\ket{\psi} \|^2$.
Let $\pr_0$ denote $\max_{s_0} (\pr_{0,s_0})$ and $\pr_1$ denote $\max_{s_1} (\pr_{1,s_1})$.
Since for every $s_0, s_1$, \[\|\bra{s_0}H^{\otimes \secp}\ket{s_1}\|^2 = \frac{1}{2^\secp},\] we conclude that $\pr_0 + \pr_1 \leq 1 + \negl(n)$, as desired.

\subsubsection{Oblivious transfer}

We obtain oblivious transfer (OT) by building on the notion of a claw-state generator (CSG) with indistinguishability security, introduced above.\footnote{Our basic construction achieves a somewhat non-standard definition that we call search security (against a malicious receiver). We also show that, by additionally assuming one-way functions, we can achieve a more standard indistinguishability-based definition. We refer the reader to \cref{subsec:OT-from-OSP} for details.} The basic idea is as follows. The OT receiver will delegate the preparation of a state \[\frac{1}{\sqrt{2}}\left(\ket{0,x_0} + (-1)^z\ket{1,x_1}\right)\] to the OT sender, where $x_0$ and $x_1$ are single bits. We will take $b= x_0 \oplus x_1$ to be the receiver's choice bit, which, by the indistinguishability security of the CSG, is computationally unpredictable to any QPT sender. Then, the sender measures their state in the standard basis to obtain two bits $(c,y)$, and defines their OT bits to be $r_0 = y, r_1 = y \oplus c$. 

Note that if $b = 0$, then $y = x_0 = x_1$ is known to the receiver, while the bit $c$ is uniformly random, meaning $r_1$ is unpredictable. On the other hand, if $b=1$, then $x_0 = 1 \oplus x_1$, so the bit $r_1 = c \oplus y$ is known to the receiver, while the bit $y$ is uniformly random, meaning $r_0$ is unpredictable. However, this analysis relies on the fact that the state measured by the sender is indeed the desired state $\frac{1}{\sqrt{2}}\left(\ket{0,x_0} + (-1)^z\ket{1,x_1}\right)$. Unfortunately, the notion of CSG (and also the underlying notion of OSP) does not guarantee any \emph{verifiability} property, meaning that we have no guarantees on what the OT sender's state might look like if the OT receiver is acting maliciously in the CSG protocol. To remedy this, we repeat the CSG protocol several times and use a cut-and-choose protocol to allow the sender to check that the receiver is behaving (close to) honestly. In the end, after combining the several protocols, we arrive at a (game-based) notion of oblivious transfer between a classical (unbounded) receiver and a quantum (polynomial-time) sender. For full details, please refer to \cref{subsec:OT-from-OSP}.

\subsubsection{Public-key encryption}

Finally, we show that two-round OSP implies (CPA-secure) public-key encryption. Our protocol uses the same building block as the OT protocol from above: a (two-round) CSG for generating the state \[\frac{1}{\sqrt{2}}\left(\ket{0,x_0} + (-1)^z\ket{1,x_1}\right)\] with indistinguishability security. 

The public key in our scheme is the first-round (classical) message $\msg_1$ of the CSG protocol, while the secret key is the secret state of the classical sender who generated this message. To encrypt a bit $m$, use $\msg_1$ to generate the above state along with a second-round message $\msg_2$, measure the state in the standard basis to obtain $(b,x_b)$, and output $(\msg_2,b,m \oplus x_b)$ as the ciphertext. 

Given the sender's state, the key generator can recover the values of $x_0,x_1$ from $\msg_2$ and thus decrypt the message. However, breaking the CPA security of this scheme reduces to being able to predict $x_b$ given $\msg_2$ for a random bit $b$. In turn, this yields an adversary attacking the indistinguishability security of the CSG: Given $\msg_1$, it prepares a state $\frac{1}{\sqrt{2}}\left(\ket{0,x_0} + (-1)^z\ket{1,x_1}\right)$ honestly along with $\msg_2$, samples a random $b$, and runs the PKE adversary to obtain a guess for $x_b$. Then, it measures its state to obtain $(b',x_{b'})$. If $b \neq b'$ (which occurs with probability 1/2), this adversary then knows $x_0 \oplus x_1$, breaking the indistinguishability security of the CSG. Again, we refer to the body, in particular \cref{subsec:PKE-from-OSP}, for the full details.

\section{Preliminaries}

Let $\secp$ denote the security parameter. We write $\negl(\cdot)$ to denote any negligible function, which is a function $f$ such that for every constant $c \in \mathbb{N}$ there exists $N \in \mathbb{N}$ such that for all $n > N$, $f(n) < n^{-c}$. We write $\nonnegl(\cdot)$ to denote any function $f$ that is not negligible, that is, there exists a constant $c$ such that for infinitely many $n$, $f(n) \geq n^{-c}$. Finally, we write $\poly(\cdot)$ to denote any polynomial function $f$, that is, there exist constants $c$ and $N$ such that for all $n > N$, $f(n) < n^c$. 

A probabilistic polynomial-time (PPT) family of circuits $\{C_\secp\}_{\secp \in \bbN}$ is a family of randomized classical circuits with $|C_\secp| \leq \poly(\secp)$, and a quantum polynomial-time (QPT) family of circuits $\{Q_\secp\}_{\secp \in \bbN}$ is a family of quantum circuits with $|Q_\secp| \leq \poly(\secp)$.

Let $\Tr$ denote the trace operator. The \emph{trace distance} between two quantum (mixed) states $\rho_0,\rho_1$, denoted $\TD(\rho_0,\rho_1)$ is defined as \[\frac{1}{2}\| \rho_0 - \rho_1\|_1,\]

where $\| \cdot \|_1$ is the \emph{trace norm}, defined by 

\[\| \rho \|_1 \coloneqq \Tr\sqrt{\rho^\dagger \rho}.\]

The trace distance between two states $\rho_0$ and $\rho_1$ is an upper bound on the probability that any (unbounded) algorithm can distinguish $\rho_0$ and $\rho_1$. 

Given two quantum operations $Q_0,Q_1$ that take as input a state on register $\cA$, their diamond distance is defined as


\[D_\diamond(Q_0,Q_1) \coloneqq \sup_{\cB}\max_{\rho_{\cA,\cB}} \| (Q_0 \otimes \cI_\cB)\rho_{\cA,\cB} - (Q_1 \otimes \cI_\cB)\rho_{\cA,\cB} \|_1,\] where $\cI_\cB$ is the identity matrix on register $\cB$. In words, the diamond distance upper bounds the trace distance between the outputs of $Q_0$ and $Q_1$ on any input (which could be entangled with an arbitrary auxiliary register $\cB$).


We will use the usual convention that $Z$ refers to the basis $\{\ket{0},\ket{1}\}$, $X$ refers to the basis $\{\ket{+},\ket{-}\}$, and $Y$ refers to the basis $\{P\ket{+},P\ket{-}\}$, where $P$ is the phase gate. We will often refer to the $X+Z$ and $X-Z$ bases, defined as 

\[X+Z = \left\{\cos(\pi/8)\ket{0}+\sin(\pi/8)\ket{1}, -\sin(\pi/8)\ket{0} + \cos(\pi/8)\ket{1}\right\},\]

\[X-Z = \left\{\cos(-\pi/8)\ket{0}+\sin(-\pi/8)\ket{1}, -\sin(-\pi/8)\ket{0} + \cos(-\pi/8)\ket{1}\right\}.\]


\begin{lemma}[Gentle measurement \cite{DBLP:journals/tit/Winter99}]\label{lemma:gentle-measurement}
Let $\rho$ be a quantum state and let $(\Pi,\cI-\Pi)$ be a projective measurement such that $\Tr(\Pi\rho) \geq 1-\delta$. Let \[\rho' = \frac{\Pi\rho\Pi}{\Tr(\Pi\rho)}\] be the state after applying $(\Pi,\cI-\Pi)$ to $\rho$ and post-selecting on obtaining the first outcome. Then, $\TD(\rho,\rho') \leq 2\sqrt{\delta}$.
\end{lemma}

\begin{lemma}[Quantum Goldreich-Levin \cite{10.5555/646516.696173}]\label{thm:QGL}
    Suppose there exists a string $s \in \{0,1\}^n$, a state $\ket{\psi}$ on $m$ qubits, a unitary $U$ on $n+m+t$ qubits that is classically controlled on its first $n$ qubits, and an $\epsilon \in (0,1)$ such that for uniformly random $r \gets \{0,1\}^n$, measuring the last qubit of $U\ket{x}\ket{\psi}\ket{0^t}$ yields $r \cdot s$ with probability at least $1/2 + \epsilon$. Then given $\ket{\psi}$, there exists a quantum algorithm that outputs $s$ with probability at least $4\epsilon^2$ using a single invocation of $U$ and $U^\dagger$.

\end{lemma}


\section{Oblivious State Preparation}

In this section, we first define a ``standard'' notion of oblivious state preparation (OSP), and then investigate variants of the definition. The standard notion we propose enables a quantum server communicating with a classical client to prepare a single-qubit state in either the standard or Hadamard basis, without actually learning the basis. This corresponds exactly to the functionality of ``Malicious 4-states QFactory with basis-blindness'' proposed by \cite{10.1007/978-3-030-34578-5_22}.

However, the concept of OSP is not fundamentally tied to the standard and Hadamard bases. Conceptually, it captures the ability for a client to enable the preparation of a state in one of two arbitrary bases on the server's system. Thus, later in the section we define a generalized notion of OSP, which enables the angle between the bases to be arbitrary, and we initiate the study  of this generalized notion.

\subsection{Basic definitions}

\begin{definition}[Oblivious State Preparation]\label{def:OSP}
   Oblivious state preparation (OSP) is a protocol that takes place between a PPT sender $S$ with input $b \in \{0,1\}$ and a QPT receiver $R$:
    \[(s,\ket*{\psi}) \gets \langle S(1^\secp,b),R(1^\secp)\rangle,\] where $s \in \{0,1\}$ is the sender's output and $\ket*{\psi}$ is the receiver's output. It should satisfy the following properties.
    \begin{itemize}
        \item \textbf{Correctness.} For any $b \in \{0,1\}$, let \[\Pi_{\OSP,b} \coloneqq \sum_{s \in \{0,1\}}\ketbra*{s} \otimes H^b\ketbra*{s}H^b.\] Then for any $b \in \{0,1\}$, \[\E\left[\|\Pi_{\OSP,b}\ket*{s}\ket*{\psi}\| : (s,\ket*{\psi}) \gets \langle S(1^\secp,b),R(1^\secp)\rangle\right] = 1-\negl(\secp).\]
        We say that the protocol satisfies \emph{perfect} correctness if the expectation above is equal to 1.
        \item \textbf{Security.} For any QPT adversary $\{\Adv_\secp\}_{\secp \in \bbN}$,
        \begin{align*}\Big|&\Pr\left[b_\Adv = 0 : (s,b_\Adv) \gets \langle S(1^\secp,0),\Adv_\secp\rangle\right]\\ &- \Pr\left[b_\Adv = 0 : (s,b_\Adv) \gets \langle S(1^\secp,1),\Adv_\secp\rangle\right]\Big|  = \negl(\secp).\end{align*}
    \end{itemize}

    We say that the protocol is a \emph{two-round OSP} if it consists of just two messages: one from the sender followed by one from the receiver. In this case, we use the following notation to describe the algorithms of the protocol.
    \begin{itemize}
        \item $\OSP.\Sen(1^\secp,b) \to (\msg_S,\state_S)$. The PPT sender takes as input the security parameter $1^\secp$ and a bit $b$, and outputs a message $\msg_S$ and state $\state_S$.
        \item $\OSP.\Rec(\msg_S) \to (\ket*{\psi},\msg_R)$. The QPT receiver takes as input the sender's message $\msg_S$ and outputs its final state $\ket*{\psi}$ and a message $\msg_R$.
        \item $\OSP.\Dec(\state_S,\msg_R) \to s$. The PPT sender takes as input its state $\state_S$ and the receiver's message $\msg_R$, and produces its output bit $s$.
    \end{itemize}
\end{definition}

Sometimes, we will refer to the above definition as a \emph{chosen-input} OSP, in order to distinguish it from a random-input variant defined below, where the sender does not fix a choice of $b$ at the beginning of the protocol.

\begin{definition}[Random-Input Oblivious State Preparation]
    Random-input OSP is a protocol that takes place between a PPT sender $S$ and a QPT receiver $R$:
    \[((s,b),\ket*{\psi}) \gets \langle S(1^\secp),R(1^\secp)\rangle,\] where $(s,b)$ is the sender's output and $\ket*{\psi}$ is the receiver's output. It should satisfy the following properties.
    \begin{itemize}
        \item \textbf{Correctness.} Let \[\Pi_\OSP \coloneqq \sum_{s,b \in \{0,1\}}\ketbra*{s,b} \otimes H^b\ketbra*{s}H^b.\] Then \[\E\left[\|\Pi_\OSP\ket*{s,b}\ket*{\psi}\| : ((s,b),\ket*{\psi}) \gets \langle S(1^\secp),R(1^\secp)\rangle\right] = 1-\negl(\secp).\]

        \item \textbf{Security.} For any QPT adversary $\{\Adv_\secp\}_{\secp \in \bbN}$,
        \[\Big|\Pr\left[b_\Adv = b : ((s,b),b_\Adv) \gets \langle S(1^\secp),\Adv_\secp\rangle\right] - \frac{1}{2}\Big| = \negl(\secp).\]
    \end{itemize}
\end{definition}

\begin{lemma}\label{lemma:random-input}
    Random-input OSP implies (chosen-input) OSP.
\end{lemma}

\begin{proof}
    To obtain chosen-input OSP with sender's choice bit $b$, the parties begin by running a random-input OSP, which (up to negligible trace distance) delivers output $H^{b'}\ket*{s}$ to the receiver and $(s,b')$ to the sender. The sender then sends the bit $c = b \oplus b'$ to the receiver, and the receiver applies $H^c$ to its state to obtain $H^b\ket*{s}$. Security follows from the security of the random-input OSP, which guarantees that the bit $b'$ is unpredictable and thus that the bit $c = b \oplus b'$ masks the sender's choice of $b$.
\end{proof}

\subsection{Claw-state generators}\label{subsec:CSG}

Next, we define (variants of) a claw-state generation (CSG) protocol, and show that CSG implies OSP. Later, in section \cref{subsec:encrypted-CNOT}, we will show that in fact OSP implies CSG, meaning that these notions are equivalent.


\begin{definition}[Claw-State Generator]\label{def:CSG}
    A claw-state generator (CSG) is a protocol that takes places between a PPT sender $S$ and a QPT receiver $R$:
    \[((x_0,x_1,z),\ket{\psi}) \gets \langle S(1^\secp,n),R(1^\secp,n)\rangle,\] where the sender's output consists of $x_0,x_1 \in \{0,1\}^n$ and $z \in \{0,1\}$, and $\ket{\psi}$ is the receiver's output. It should satisfy the following notion of correctness, and, depending on the setting, it should also satisfy either search security or indistinguishability security.

    \begin{itemize}
        \item \textbf{Correctness.} Let \[\Pi_\CSG \coloneqq \sum_{x_0 \neq x_1 \in \{0,1\}^n,z \in \{0,1\}} \ketbra*{x_0,x_1,z} \otimes \frac{1}{2}\left(\ket{x_0}+(-1)^z\ket{x_1})(\bra{x_0}+(-1)^z\bra{x_1}\right).\] Then \[\E\left[\| \Pi_\CSG \ket{x_0,x_1,z}\ket{\psi}
    \| : ((x_0,x_1,z),\ket{\psi}) \gets \langle S(1^\secp,n),R(1^\secp,n)\rangle\right] = 1-\negl(\secp).\]
    We say that the protocol has \emph{perfect} correctness if the above probability is equal to 1.
        \item \textbf{Search security.} For any QPT adversary $\{\Adv_\secp\}_{\secp \in \bbN}$, \[\Pr\left[x_\Adv = (x_0,x_1) : ((x_0,x_1,z),x_\Adv) \gets \langle S(1^\secp,n),\Adv_\secp \rangle\right] = \negl(\secp).\]
        \item \textbf{Indistinguishability security.} For any QPT adversary $\{\Adv_\secp\}_{\secp \in \bbN}$ and any $i \in [n]$, \[\Big|\Pr\left[b_\Adv = x_{0,i} \oplus x_{1,i} : ((x_0,x_1,z),b_\Adv) \gets \langle S(1^\secp,n),\Adv_\secp \rangle\right] - \frac{1}{2}\Big| = \negl(\secp).\]
    \end{itemize}

     We say that the protocol is a \emph{two-round CSG} if it consists of just two messages: one from the sender followed by one from the receiver. In this case, we use the following notation to describe the algorithms of the protocol.
    \begin{itemize}
        \item $\CSG.\Sen(1^\secp,n) \to (\msg_S,\state_S)$. The PPT sender takes as input the security parameter $1^\secp$ and outputs a message $\msg_S$ and state $\state_S$.
        \item $\CSG.\Rec(\msg_S) \to (\ket*{\psi},\msg_R)$. The QPT receiver takes as input the sender's message $\msg_S$ and outputs its final state $\ket*{\psi}$ and a message $\msg_R$.
        \item $\CSG.\Dec(\state_S,\msg_R) \to (x_0,x_1,z)$. The PPT sender takes as input its state $\state_S$ and the receiver's message $\msg_R$, and produces its output $(x_0,x_1,z)$.
    \end{itemize}
    
\end{definition}

We also define a version of a claw-state generator where the honest receiver obtains $\frac{1}{\sqrt{2}}(\ket{0,x_0}+(-1)^z\ket{1,x_1})$, that is, where the two members of the claw are differentiated by the first bit. 

\begin{definition}[Differentiated-Bit Claw-State Generator]\label{def:DPCSG}
    A differentiated-bit claw-state generator is defined exactly like a claw-state generator except that the correctness property is stated as follows. Let \[\Pi_\DPCSG \coloneqq \sum_{x_0 \neq x_1 \in \{0,1\}^n,z \in \{0,1\}} \ketbra*{x_0,x_1,z} \otimes \frac{1}{2}\left(\ket{0,x_0}+(-1)^z\ket{1,x_1})(\bra{0,x_0}+(-1)^z\bra{1,x_1}\right).\] Then \[\E\left[\| \Pi_\DPCSG \ket{x_0,x_1,z}\ket{\psi}
    \| : ((x_0,x_1,z),\ket{\psi}) \gets \langle S(1^\secp,n),R(1^\secp,n)\rangle\right] = 1-\negl(\secp).\]
\end{definition}

It is straightforward to obtain a differentiated-bit CSG with search security from a (plain) CSG with search security.

\begin{lemma}\label{lemma:DPCSG}
    CSG with search security (\cref{def:CSG}) implies differentiated-bit CSG with search security (\cref{def:DPCSG}).
\end{lemma}

\begin{proof}
    The protocol for differentiated-bit CSG goes as follows. First, run a plain CSG. Then, the sender samples a uniformly random $y$ conditioned on $y \cdot x_0 = 0$ and $y \cdot x_1 = 1$, and sends $y$ to the receiver. Finally, the receiver applies the map
    \[\frac{1}{\sqrt{2}} (\ket{x_0}+(-1)^z\ket{x_1}) \to \frac{1}{\sqrt{2}} (\ket{y \cdot x_0,x_0}+(-1)^z\ket{y \cdot x_1,x_1} = \frac{1}{\sqrt{2}} (\ket{0,x_0}+(-1)^z\ket{1,x_1}).\]
    Correctness is immediate, and security follows by reduction. In particular, the reduction to the security of the CSG will run the adversary for differentiated-bit CSG, sample a truly uniform $y$ to feed to the adversary in the last round, and return the adversary's guess $x_\Adv$. The $y$ will be properly distributed with probability 1/2, and thus the reduction succeeds with probability at least half that of the differentiated-bit CSG adversary.
\end{proof}

Next, we prove that OSP follows from any CSG with search security.

\begin{theorem}\label{thm:OSP-from-CSP}
    CSG with search security implies OSP.
\end{theorem}

\begin{proof}
    We show how to use the differentiated-preimage variant of CSG to build a random-input OSP, and then appeal to \cref{lemma:DPCSG} to obtain differentiated-preimage CSG from CSG, and \cref{lemma:random-input} to obtain (chosen-input) OSP from random-input OSP. The main idea is to use Goldreich-Levin, similar to how it is used in \cite{10.1007/978-3-031-38554-4_6}, in order to use a distinguisher for the OSP basis to obtain a predictor for the claw-state. The protocol is given in \cref{fig:OSP-from-CSG}.

    \protocol{OSP from CSG}{Random-input OSP from any differentiated-bit claw-state generator.}{fig:OSP-from-CSG}{
    \begin{itemize}
        \item The sender $S$ and receiver $R$ begin by running a differentiated-bit CSG, which delivers $(x_0,x_1,z)$ to $S$ and (up to negligible trace distance) $\frac{1}{\sqrt{2}}(\ket{0,x_0} + (-1)^z\ket{1,x_1})$ to $R$, where $x_0,x_1 \in \{0,1\}^n$ and $z \in \{0,1\}$.
        \item Next, $S$ samples $r_0,r_1 \gets \{0,1\}^n$ and sends them to $R$.
        \item Using $r_0,r_1$, the receiver $R$ applies the operation that maps \[\frac{1}{\sqrt{2}}(\ket{0,x_1} + \ket{1,x_1}) \to \frac{1}{\sqrt{2}}(\ket{0,x_0}\ket{r_0 \cdot x_0} + (-1)^z\ket{1,x_1}\ket{r_1 \cdot x_1}),\] and then measures all but the last qubit in the Hadamard basis to obtain a string $d \in \{0,1\}^{n+1}$, which it returns to $S$. $R$ outputs its remaining qubit. 
        \item $S$ sets $b \coloneqq (x_0,x_1) \cdot (r_0,r_1)$. If $b = 0$, then $S$ sets $s \coloneqq x_0 \cdot r_0 = x_1 \cdot r_1$. If $b=1$, then $S$ sets $s \coloneqq z \oplus d \cdot (1, x_0 \oplus x_1)$. $S$ outputs $(s,b)$. 
    \end{itemize}
    }

    First, we argue correctness. If the sender's output is $b = (x_0,x_1) \cdot (r_0,r_1) = 0$ and $s = x_0 \cdot r_0 = x_1 \cdot r_1$, then the receiver's state before their final measurement is (negligibly close to) \[\frac{1}{\sqrt{2}}\left(\ket*{0,x_0} + (-1)^z\ket*{1,x_1}\right) \otimes \ket*{s}.\] So their Hadamard basis measurement has no effect on the last qubit, and their output state will be $\ket*{s} = H^b\ket*{s}$.

    Next, if the sender computes $b = (x_0,x_1) \cdot (r_0,r_1) = 1$, then the receiver's state before their final measurement is either (negligibly close to) \[\frac{1}{\sqrt{2}}\left(\ket*{0,x_0}\ket*{0} + (-1)^z\ket*{1,x_1}\ket*{1}\right) \text{    or    } \frac{1}{\sqrt{2}}\left(\ket*{0,x_0}\ket*{1} + (-1)^z\ket*{1,x_1}\ket*{0}\right).\]

    Either way, a standard calculation shows that if they measure all but their last qubit in the Hadamard basis to obtain $d$, the last qubit becomes $Z^{z \oplus d \cdot (1,x_0 \oplus x_1)}\ket*{+} = H\ket*{s}$, for $s = z \oplus d \cdot (1,x_0 \oplus x_1)$.

    Now, we establish security. Suppose there exists $\Adv = \{\Adv_\secp\}_{\secp \in \bbN}$ that has $\nonnegl(\secp)$ advantage in the OSP security game. Then there must be some $\nonnegl(\secp)$ probability that, after $\Adv$ and $S$ interact in the CSG protocol, $\Adv$ still has $\nonnegl(\secp)$ advantage \emph{conditioned on the interaction so far}. 

    That is, let $\ket*{\psi}$ be the state of $\Adv_\secp$ right after the conclusion of the CSG protocol, and define $B$ to be the routine that takes $r_0,r_1 \in \{0,1\}^n$ as input, runs the remainder of $\Adv$'s strategy using state $\ket*{\psi}$ and strings $(r_0,r_1)$, and outputs $\Adv_\secp$'s guess for $b$. Then we have that

    \[\Pr_{\ket*{\psi}}\left[\E_{r_0,r_1 \gets \{0,1\}^n}[B(\ket*{\psi},(r_0,r_1))= (x_0,x_1) \cdot (r_0,r_1)] = \frac{1}{2} + \nonnegl(\secp)\right] = \nonnegl(\secp).\]

    Now, we appeal to \cref{thm:QGL}, which implies that there exists a $B'$ such that when $B$ has $\nonnegl(\secp)$ advantage given advice state $\ket*{\psi}$, $B'(\ket*{\psi})$ has $\nonnegl(\secp)$ probability of outputting $(x_0,x_1)$. But this yields an adversary that breaks the security of the CSG with $\nonnegl(\secp)$ probability, completing the proof.

\end{proof}

\subsection{OSP with generalized angle}

Next, we consider a generalized notion of OSP, where the protocol is defined by \emph{any} choice of two (not necessarily mutually unbiased) single-qubit bases. By post-processing with an appropriate rotation, we can without loss of generality consider one basis to be $\{\ket*{+},\ket*{-}\}$ and the other to be a rotated basis on the XY plane of the Bloch sphere. For example, by having the receiver apply a Hadamard gate followed by a $\sqrt{X}$ rotation to their state received as output from the protocol described in \cref{def:OSP}, we have that if $b = 0$, the receiver obtains either $\ket*{+}$ or $\ket*{-}$ and if $b=1$, the receiver obtains either $\frac{1}{\sqrt{2}}(\ket*{0} + i\ket*{1})$ or $\frac{1}{\sqrt{2}}(\ket*{0} - i\ket*{1})$. 

While this is an example of OSP with mutually unbiased bases (two bases at a maximum angle), one can consider OSP with arbitrary angle between the chosen bases. For any angle $\theta \in [2\pi]$, we define $\ket*{+_\theta} \coloneqq \frac{1}{\sqrt{2}}(\ket*{0} + e^{i\theta}\ket*{1})$. Then, OSP with mutually unbiased bases corresponds to OSP with basis $\{\ket*{+},Z\ket*{+}\}$ or basis $\{\ket*{+_{\pi/2}},Z\ket*{+_{\pi/2}}\}$. For any $\epsilon \in (0,1]$, we define $\epsilon$-OSP to be an OSP with bases $\{\ket*{+},Z\ket*{+}\}$ and $\{\ket*{+_{\epsilon\pi/2}},Z\ket*{+_{\epsilon\pi/2}}\}$, defined formally as follows.

\begin{definition}[$\epsilon$-OSP]
    An OSP with generalized angle, or $\epsilon$-OSP, is a protocol that takes place between a PPT sender $S$ with input $b \in \{0,1\}$ and a QPT receiver $R$:
    \[(s,\ket*{\psi}) \gets \langle S(1^\secp,b),R(1^\secp)\rangle,\] where $s \in \{0,1\}$ is the sender's output and $\ket*{\psi}$ is the receiver's output. It should satisfy the following properties.
    \begin{itemize}
        \item \textbf{Correctness.} For any $b \in \{0,1\}$, let \[\Pi_{\epsilon\text{-}\OSP,b} \coloneqq \sum_{s \in \{0,1\}}\ketbra*{s} \otimes Z^s\ketbra*{+_{b\epsilon\pi/2}}Z^s.\] Then for any $b \in \{0,1\}$, \[\E\left[\|\Pi_{\epsilon\text{-}\OSP,b}\ket*{s}\ket*{\psi}\| : (s,\ket*{\psi}) \gets \langle S(1^\secp,b),R(1^\secp)\rangle\right] = 1-\negl(\secp).\]
        \item \textbf{Security.} For any QPT adversary $\{\Adv_\secp\}_{\secp \in \bbN}$,
        \begin{align*}\Big|&\Pr\left[b_\Adv = 0 : (s,b_\Adv) \gets \langle S(1^\secp,0),\Adv_\secp\rangle\right]\\ &- \Pr\left[b_\Adv = 0 : (s,b_\Adv) \gets \langle S(1^\secp,1),\Adv_\secp\rangle\right]\Big|  = \negl(\secp).\end{align*}
    \end{itemize}
\end{definition}

In this work, we give constructions of ``standard'' OSP (with $\epsilon=1$), and focus on deriving implications of this notion. However, it would be useful to know if this flavor of OSP is in some sense the ``minimal'' OSP assumption, and more generally, how $\epsilon$-OSP relates for various choices of $\epsilon$. While we do not fully resolve these questions in this work, we do show that $\epsilon$-OSP implies OSP for many choices of $\epsilon$. In particular, we show the following claim.\footnote{Ideally, we would like to show the implication for any rational $\epsilon$, but the current approach does not appear to work for $\epsilon = \text{even} / \text{odd}$.} 

\begin{claim}
    For any rational constant $\epsilon \in (0,1]$ that can be written as $\epsilon = c/d$ where $c$ is an odd integer, it holds that $\epsilon$-OSP implies OSP. In particular, for any constant $n$, $1/n$-OSP implies OSP.\footnote{In fact, our proof shows that $\epsilon$-OSP implies OSP even for any inverse-polynomial $\epsilon = 1/\poly(\secp)$. Thus, OSP with any small enough but still non-trivial angle implies standard OSP. We note that \cite{dunjko2016blindquantumcomputingidentical} has established similar results in the information-theoretic setting.}
\end{claim}


\begin{proof}
    The main idea is to use the fact that given two states $Z^{s_1}\ket*{+_{\phi_1}}$ and $Z^{s_2}\ket*{+_{\phi_2}}$, it is possible to produce the state $Z^{s_1 \oplus s_2}\ket*{+_{\phi_1 + \phi_2}}$ with probability 1/2. That is, we can sum the angles of the states using a procedure that succeeds with probability 1/2 (this procedure is used, for example, in Kuperberg's algorithm \cite{doi:10.1137/S0097539703436345}). This follows by simply applying a CNOT from the first to the second state, and then measuring the second state in the standard basis. If the result is 0, the first state is now $Z^{s_1 \oplus s_2}\ket*{+_{\phi_1 + \phi_2}}$. To confirm this, we have
    \begin{align*}
        \text{CNOT}&\left(Z^{s_1}\ket*{+_{\phi_1}} \otimes Z^{s_1}\ket*{+_{\phi_2}}\right) \\ &= \frac{1}{2}\text{CNOT}\left(Z^{s_1}\otimes Z^{s_2}\right)\left(\ket*{0} + e^{i\phi_1}\ket*{1}\right)\left(\ket*{0} + e^{i\phi_2}\ket*{1}\right) \\
        &= \frac{1}{2}\left(Z^{s_1 \oplus s_2}\otimes Z^{s_2}\right)\left(\ket*{00} + e^{i\phi_2}\ket*{01} + e^{i\phi_1}\ket*{11} + e^{i(\phi_1+\phi_2)}\ket*{10}\right) \\
        &= \frac{1}{2}\left(Z^{s_1 \oplus s_2}\otimes Z^{s_2}\right)\left(\left(\ket*{0}+e^{i(\phi_1 +\phi_2)}\ket*{1}\right)\ket*{0} + e^{i \phi_2}\left(\ket*{0}+e^{i(\phi_1-\phi_2)}\ket*{1}\right)\ket*{1}\right) \\
        &= \frac{1}{\sqrt{2}}Z^{s_1 \oplus s_2}\ket*{+_{\phi_1 +\phi_2}}\ket*{0} + \frac{1}{\sqrt{2}}Z^{s_2}e^{i\phi_1}\ket*{+_{\phi_1-\phi_2}}\ket*{1}.
    \end{align*}

    Now, for simplicity suppose that $d$ is a power of 2 (a similar procedure works for arbitrary $d$), and let $\secp$ be the security parameter. Given input $b \in \{0,1\}$, the parties will run $8^{\log d}\secp = d^3 \secp = \poly(\secp)$ many $\epsilon$-OSP protocols with sender input $b$. The receiver is now in possession of $8^{\log d}\secp$ many states (negligibly close to) $Z^s\ket*{+_{b \epsilon \pi/2}}$, each with potentially different $s$. The receiver now runs a procedure to ``sum'' $d$ of them together. If $b=0$, this results in a state \[Z^{s'}\ket*{+} \in \left\{\ket*{+},\ket*{-}\right\},\] whereas if $b=1$, this results in a state \[Z^{s'}\ket*{+_{d\epsilon \pi/2}} = Z^{s'}\ket*{+_{c\pi/2}} \in \left\{\ket*{+_{\pi/2}},Z\ket*{+_{\pi/2}}\right\},\] where the final inclusion uses the fact that $c$ is an odd integer, and $s'$ is the XOR of all the $s$ corresponding to states involved in the sum.

    To obtain the desired sum, the receiver operates in layers. Given $8k\secp$ pairs of states at angle $\phi$, the receiver splits them into $4k\secp$ pairs of states, and applies the above CNOT-and-measure procedure to each pair. The expected number of ``successes'' is $2k\secp$, and by Chernoff, there will be at least $k\secp$ successes with all but $\negl(\secp)$ probability. Thus, with all but $\negl(\secp)$ probability, the receiver obtains $k\secp$ states at angle $2\phi$. 

    Now, starting with $8^{\log d}\secp$ many states at angle $\phi$ and operating for $\log d$ layers, the receiver will end up with $\secp \geq 1$ state at angle $2^{\log d}\phi=d\phi$. This establishes correctness of the OSP protocol.

    Finally, security of the OSP protocol follows by a standard hybrid argument from the security of the $\epsilon$-OSP protocol, which we have simply repeated $d^3 \secp$ times.
    
\end{proof}

We conclude this section by proving a natural structural lemma about OSP, and more generally $\epsilon$-OSP. It shows that in the honest case, the sender's output $s$ must be (negligibly) close to uniformly random for either choice of $b \in \{0,1\}$. That is, an honest run of $\epsilon$-OSP produces a uniformly random eigenstate of the observable specified by $b$ (though we caution that this is no longer necessarily true when the receiver is adversarial).

\begin{lemma}\label{lemma:random-s}
    For any $\epsilon \in (0,1]$ such that $\epsilon \geq 1/\poly(\secp)$, any secure $\epsilon$-OSP protocol, and any $b \in \{0,1\}$, it holds that 
    \[\Big| \Pr\left[s = 0 : (s,\ket{\psi}) \gets \langle S(1^\secp,b),R(1^\secp)\rangle\right] - \frac{1}{2}\Big| = \negl(\secp).\]
\end{lemma}

\begin{proof}
    Suppose otherwise, and without loss of generality suppose that when $b=1$, the probability that $s = 0$ is equal to $1/2 + \delta$ for some $\delta = \nonnegl(\secp)$ (the other cases are symmetric). We show that this would contradict the security of the $\epsilon$-OSP. Consider an adversary $\Adv$ that participates as an honest receiver, measures their final state in the $Y$ basis $\{\ket*{+_{\pi/2}},Z\ket*{+_{\pi/2}}\}$, and guesses $b=1$ if they obtain outcome $\ket*{+_{\pi/2}}$. First note that, no matter what the distribution on $s$ is when $b = 0$, we have that \[\Big|\Pr[\Adv = 1 : b = 0] - \frac{1}{2} \Big| = \negl(\secp),\] since measuring either the $\ket{+}$ or $\ket{-}$ state in the $Y$ basis yields a uniformly random outcome. Next, let \[\epsilon' \coloneqq \cos^2\left((1-\epsilon)\frac{\pi}{4}\right)\] be the probability of obtaining outcome $\ket*{+_{\pi/2}}$ when measuring the state $\ket*{+_{\epsilon \pi/2}}$ in the $Y$ basis, and note that $\epsilon' \geq 1/2 + 1/\poly(\secp)$ whenever $\epsilon \geq 1/\poly(\secp)$. Then to complete the proof, we have that
    \begin{align*}
        \Pr\left[\Adv = 1 : b=1\right] &\geq \left(\frac{1}{2}+\delta\right)\epsilon' + \left(\frac{1}{2}-\delta\right)(1-\epsilon') - \negl(\secp)  \\&= \frac{1}{2} - \delta + 2\delta\epsilon' - \negl(\secp) \\ &\geq \frac{1}{2} + \frac{2\delta}{\poly(\secp)} - \negl(\secp) \\ & =\frac{1}{2} + \nonnegl(\secp).
    \end{align*}

\end{proof}
\section{Constructions}\label{sec:constructions}

In this section, we provide constructions of OSP from trapdoor claw-free functions (TCFs). 

In \cref{subsec:OSP-TCF}, we show how to construct OSP from any (plain) TCF, meaning we assume no extra properties such as dual-mode or adaptive hardcore bit. Moreover, all that we require from the TCF is that there is some inverse-polynomial probability that an image has exactly two preimages (which can be obtained efficiently using the trapdoor), and otherwise it can have 1 or 3 or more (as long as the trapdoor correctly identifies these as ``bad'' images).


In \cref{subsec:OSP-dTCF}, we show how to obtain a \emph{two-round} OSP by assuming an extra \emph{dual-mode} property of the TCF. A dual-mode TCF (dTCF) can be sampled in either a disjoint mode or lossy mode. Again, we only assume that there is some inverse polynomial probability that an image has two preimages in lossy mode, and we use the amplification lemma for dTCFs recently established by \cite{10.1007/978-3-031-68382-4_8} in order to show that such dTCFs still imply two-round OSP.

Before coming to the constructions, we provide definitions of TCFs.

\begin{definition}[Trapdoor claw-free function]\label{def:TCF}
    A trapdoor claw-free function (TCF) consists of a PPT parameter generation algorithm $\Gen(1^\secp) \to \pparam,\sparam$ and a keyed family of PPT computable functions \[\left\{\left\{F_{\pparam} : \{0,1\}^{n(\secp)} \to \{0,1\}^{m(\secp)}\right\}_{(\pparam,\cdot) \in \Gen(1^\secp)}\right\}_{\secp \in \bbN}.\] There exists a family of distributions \[\left\{\left\{\cD_\pparam\right\}_{(\pparam,\cdot) \in \Gen(1^\secp)}\right\}_{\secp \in \bbN}\] over $\{0,1\}^{n(\secp)}$ and a PPT algorithm $\Invert(\sparam,y)$ such that the following properties are satisfied. 
    \begin{itemize}
        \item \textbf{Efficient state preparation.} There is a QPT algorithm that, given any $(\pparam,\cdot) \in \Gen(1^\secp)$, outputs a state within negligible trace distance of the state 
        \[\ket*{\psi_\pparam} \coloneqq \sum_{x \in \{0,1\}^{n(\secp)}}\sqrt{\cD_\pparam(x)}\ket*{x}.\]
        \item \textbf{Efficient inversion.} For any $(\pparam,\cdot) \in \Gen(1^\secp)$, let $\Claw_\secp \subseteq \{0,1\}^{m(\secp)}$ be the set of $y \in \{0,1\}^{m(\secp)}$ such that there exists exactly two $x_0,x_1 \in \{0,1\}^{n(\secp)}$ such that $F_\pparam(x_0) = F_\pparam(x_1) = y$, and \[\Big| \frac{\cD_\pparam(x_0)}{\cD_\pparam(x_0) + \cD_\pparam(x_1)} - \frac{\cD_\pparam(x_1)}{\cD_\pparam(x_0) + \cD_\pparam(x_1)}\Big| = \negl(\secp).\] Then there exists $\delta(\secp) = 1/\poly(\secp)$ such that
        \[\Pr\left[F_\pparam(x_0) = F_\pparam(x_1) = y : \begin{array}{r}(\pparam,\sparam) \gets \Gen(1^\secp) \\ x \gets \cD_\pparam \\ y \coloneqq F_\pparam(x) \\ \{x_0,x_1\} \gets \Invert(\sparam,y)\end{array}\right] \geq \delta(\secp),\] 
        and for all $y \notin \Claw_\secp$, $\Invert(\sparam,y) = \bot$.

        \item \textbf{Claw-free.} For any QPT adversary $\{\Adv_\secp\}_{\secp \in \bbN}$,

        \[\Pr\left[F_\pparam(x_0) = F_\pparam(x_1) : \begin{array}{r} (\pparam,\sparam) \gets \Gen(1^\secp) \\ \{x_0,x_1\} \gets \Adv_\secp(\pparam)\end{array}\right] = \negl(\secp).\]
    
    \end{itemize}
\end{definition}

Next, we define a \emph{dual-mode} variant of TCFs.

\begin{definition}[Dual-mode trapdoor claw-free function]\label{def:dTCF}
    A dual-mode trapdoor claw-free function (dTCF) consists of a PPT parameter generation algorithm $\Gen(1^\secp,\mu) \to \pparam,\sparam$ that takes as input a ``mode'' bit $\mu \in \{0,1\}$, and a keyed family of PPT computable functions\footnote{Notice that, as compared to plain TCFs, dual-mode TCFs take an extra bit of input. We will require that each claw has one preimage that starts with 0 and one that starts with 1, which is important for the amplification lemma of \cite{10.1007/978-3-031-68382-4_8}.} \[\left\{\left\{F_{\pparam} : \{0,1\} \times \{0,1\}^{n(\secp)} \to \{0,1\}^{m(\secp)}\right\}_{(\pparam,\cdot) \in \Gen(1^\secp,\mu)}\right\}_{\mu \in \{0,1\}, \secp \in \bbN}.\] There exists a family of distributions \[\left\{\left\{\cD_\pparam\right\}_{(\pparam,\cdot) \in \Gen(1^\secp,\mu)}\right\}_{\mu \in \{0,1\}, \secp \in \bbN}\] over $\{0,1\}^{n(\secp)}$  and a PPT algorithm $\Invert(\sparam,b,y)$ such that the following properties are satisfied. 

    \begin{itemize}
        \item \textbf{Efficient state preparation.} There is a QPT algorithm that, given any $(\pparam,\cdot) \in \Gen(1^\secp,\mu)$, outputs a state within negligible trace distance of the state 
        \[\ket*{\psi_\pparam} \coloneqq \sum_{x \in \{0,1\}^{n(\secp)}}\sqrt{\cD_\pparam(x)}\ket*{x}.\]
        \item \textbf{Efficient inversion.} For any $\mu \in \{0,1\}$ and $b \in \{0,1\}$,\footnote{Note that this property implies that with overwhelming probability over $(\pparam,\sparam) \gets \Gen(1^\secp,\mu)$ and $x \gets \cD_\pparam$, $x$ has no siblings $x'$ such that $F_\pparam(b,x) = F_\pparam(b,x')$. That is, $F_\pparam(0,\cdot)$ and $F_\pparam(1,\cdot)$ are effectively injective.}
        \[\Pr\left[\Invert(\sparam,b,y) = x : \begin{array}{r}(\pparam,\sparam) \gets \Gen(1^\secp,\mu) \\ x \gets \cD_\pparam \\ y \coloneqq F_\pparam(b,x)\end{array} \right] = 1-\negl(\secp).\]
        \item \textbf{Dual-mode.} 
        \begin{itemize}
            \item Disjoint mode $(\mu=0)$: For any $b \in \{0,1\}$,
            \[\Pr\left[\exists x' \ \text{s.t.} \ F_\pparam(1-b,x') = y : \begin{array}{r}(\pparam,\sparam) \gets \Gen(1^\secp,0) \\ x \gets \cD_\pparam \\ y \coloneqq F_\pparam(b,x)\end{array}\right] = \negl(\secp).\]
            \item Lossy mode $(\mu=1)$: There exists $\delta(\secp) = 1/\poly(\secp)$ such that for any $b \in \{0,1\}$,
            \[\Pr\left[\begin{array}{l}\exists x' \ \text{s.t.} \ F_\pparam(1-b,x') = y  \end{array}: \begin{array}{r}(\pparam,\sparam) \gets \Gen(1^\secp,1) \\ x \gets \cD_\pparam \\ y \coloneqq F_\pparam(b,x)\end{array}\right] \geq \delta(\secp),\]
            and for any $\nu(\secp) = \nonnegl(\secp)$,
            \[\Pr\left[\begin{array}{l}\exists x' \ \text{s.t.} \ F_\pparam(1-b,x') = y  \\
            \wedge \ \Big|\frac{\cD_\pparam(x)}{\cD_\pparam(x) + \cD_\pparam(x')} - \frac{\cD_\pparam(x')}{\cD_\pparam(x) + \cD_\pparam(x')}\Big| \geq \nu(\secp)\end{array}: \begin{array}{r}(\pparam,\sparam) \gets \Gen(1^\secp,1) \\ x \gets \cD_\pparam \\ y \coloneqq F_\pparam(b,x)\end{array}\right] = \negl(\secp),\] where this last requirement enforces that there are (effectively) no ``unbalanced'' claws.
        \end{itemize}
        \item \textbf{Mode indistinguishability.} For any QPT adversary $\{\Adv_\secp\}_{\secp \in \bbN}$,
        \[\Big|\Pr\left[\Adv_\secp(\pparam) = 1 : (\pparam,\sparam) \gets \Gen(1^\secp,0)\right] - \Pr\left[\Adv_\secp(\pparam) = 1 : (\pparam,\sparam) \gets \Gen(1^\secp,1)\right]\Big| = \negl(\secp).\]
    \end{itemize}

\end{definition}

\begin{remark}
    Dual-mode (and thus plain) TCFs are known from LWE \cite{Brakerski2018ACT} (even with polynomial modulus-to-noise ratio, since we can take $\delta = 1-1/\poly$) and from the ``extended linear hidden shift'' assumption on cryptographic group actions \cite{10.1007/978-3-031-22318-1_10,10.1007/978-3-031-68382-4_8}.
\end{remark}

\subsection{OSP from plain TCFs}\label{subsec:OSP-TCF}

\begin{theorem}
    Any TCF satisfying \cref{def:TCF} implies OSP (\cref{def:OSP}).
\end{theorem}

\begin{proof}
    This follows fairly immediately by using the TCF to construct a claw-state generator (\cref{def:CSG}), and then appealing to \cref{thm:OSP-from-CSP}, which shows how to construct OSP from any CSG. Let $\delta = \delta(\secp)$ be the parameter from the efficient inversion property of the TCF. Then the CSG is constructed as follows.

    \begin{itemize}
        \item $S$ samples $(\sparam,\pparam) \gets \Gen(1^\secp)$ and sends $\pparam$ to $R$. 
        \item Run the following at most $\secp/\delta$ times. If the protocol has not terminated at that point, $R$ outputs $\bot$, and $S$ samples $x_0,x_1 \gets \{0,1\}^n$ and outputs $(x_0,x_1,0)$.
        \begin{itemize}
            \item $R$ prepares a state within negligible trace distance of $\ket*{\psi_\pparam}$, applies $F_\pparam$ in superposition to a fresh register, and measures that register to obtain $y$.
            \item $R$ sends $y$ to $S$. If $\Invert(\sparam,y) = \{x_0,x_1\}$, $S$ outputs $(x_0,x_1,0)$, instructs $R$ to terminate, and $R$ outputs its state. Otherwise, if $\Invert(\sparam,y) = \bot$, $S$ instructs $R$ to repeat.
        \end{itemize}
    \end{itemize}

    Correctness of the CSG follows first by noting that the probability that the parties never continue from the loop is at most $(1-\delta)^{\secp/\delta} \leq e^{-\secp} = \negl(\secp)$. Then, when $\Invert(\sparam,y) = \{x_0,x_1\}$, we know that $R$'s state is negligibly close to $\frac{1}{\sqrt{2}}(\ket{x_0} + \ket{x_1})$, by the efficient state preparation and efficient inversion properties of the TCF. 
    
    To show security of the CSG, first note that any adversary with noticeable probability of guessing $S$'s output $(x_0,x_1)$ must cause $S$ to terminate at some point during the loop, since otherwise  $(x_0,x_1)$ is uniformly random and independent of their view. Then, consider the following reduction to the claw-freeness of the TCF. The reduction receives $\pparam$ from its challenger, samples a random round $i \gets [\secp/\delta]$, runs the adversary and instructs them to terminate on the $i$'th invocation of the loop, and returns the adversary's guess for $(x_0,x_1)$. With probability at least $1/(\secp/\delta) = 1/\poly(\secp)$, this matches the adversary's view in the real protocol, which means that the reduction has a noticeable probability of outputting a claw $(x_0,x_1)$.

\end{proof}

\subsection{Two-round OSP from dual-mode TCFs}\label{subsec:OSP-dTCF}


In \cite{10.1007/978-3-031-68382-4_8}, it is shown that any dTCF satisfying \cref{def:dTCF} (i.e.\ with any $\delta = 1/\poly(\secp)$) implies the following variant of dTCF with a \emph{phase computation} property that succeeds with all but negligible probability. Note that we also relax the inversion property to a \emph{partial inversion} property, which only requires that the first bit $b$ of the preimage be recovered. This allows for the possibility that the functions $F_\pparam(0,\cdot)$ and $F_\pparam(1,\cdot)$ are non-injective.

\begin{definition}[dTCF with efficient phase computation]\label{def:phase-computation}
    A dTCF with \emph{efficient phase computation} consists of a PPT parameter generation algorithm $\Gen(1^\secp,\mu) \to \pparam,\sparam$ that takes as input a ``mode'' bit $\mu \in \{0,1\}$, and a keyed family of PPT computable functions \[\left\{\left\{F_{\pparam} : \{0,1\} \times \{0,1\}^{n(\secp)} \to \{0,1\}^{m(\secp)}\right\}_{(\pparam,\cdot) \in \Gen(1^\secp,\mu)}\right\}_{\mu \in \{0,1\}, \secp \in \bbN}.\] There exists a family of distributions \[\left\{\left\{\cD_\pparam\right\}_{(\pparam,\cdot) \in \Gen(1^\secp,\mu)}\right\}_{\mu \in \{0,1\}, \secp \in \bbN}\] over $\{0,1\}^{n(\secp)}$ and PPT algorithms $\PartialInvert(\sparam,y)$ and $\PhaseInvert(\sparam,y,d)$ such that the following properties are satisfied. 

    \begin{itemize}
        \item \textbf{Efficient state preparation.} There is a QPT algorithm that, given any $(\pparam,\cdot) \in \Gen(1^\secp,\mu)$, outputs a state within negligible trace distance of the state 
        \[\ket*{\psi_\pparam} \coloneqq \sum_{x \in \{0,1\}^{n(\secp)}}\sqrt{\cD_\pparam(x)}\ket*{x}.\]

        \item \textbf{Efficient partial inversion.} For any $\mu \in \{0,1\}$,
        \[\Pr\left[B = \left\{b : \exists x \ \text{s.t.} \ F_\pparam(b,x) = y\right\} : \begin{array}{r}(\pparam,\sparam) \gets \Gen(1^\secp,\mu) \\ x \gets \cD_\pparam \\ y \coloneqq F_\pparam(b,x) \\ B \gets \PartialInvert(\sparam,y)\end{array}\right] = 1-\negl(\secp).\]

        \item \textbf{Dual-mode.}
        \begin{itemize}
            \item Disjoint mode $(\mu=0)$: For any $b \in \{0,1\}$,
            \[\Pr\left[\exists x' \ \text{s.t.} \ F_\pparam(1-b,x') = y : \begin{array}{r}(\pparam,\sparam) \gets \Gen(1^\secp,0) \\ x \gets \cD_\pparam \\ y \coloneqq F_\pparam(b,x)\end{array}\right] = \negl(\secp).\]
            \item Lossy mode $(\mu=1)$: For any $(\pparam,\cdot) \in \Gen(1^\secp,1)$, $y \in \{0,1\}^{m(\secp)}$, and $d \in \{0,1\}^{n(\secp)}$, define
            \[w_{\pparam,y,d,0} \coloneqq \sum_{x: F_\pparam(0,x) = y}(-1)^{d \cdot x}\sqrt{\cD_\pparam(x)}, \ \ \ w_{\pparam,y,d,1} \coloneqq \sum_{x: F_\pparam(1,x) = y}(-1)^{d \cdot x}\sqrt{\cD_\pparam(x)},\] and re-normalize
            \[\widetilde{w}_{\pparam,y,d,0} \coloneqq \frac{w_{\pparam,y,d,0}}{\sqrt{w_{\pparam,y,d,0}^2 + w_{\pparam,y,d,1}^2}}, \ \ \ \widetilde{w}_{\pparam,y,d,1} \coloneqq \frac{w_{\pparam,y,d,1}}{\sqrt{w_{\pparam,y,d,0}^2 + w_{\pparam,y,d,1}^2}}.\]
        
            Then for any $d \in \{0,1\}^{n(\secp)}$, there exists $\nu(\secp) = \negl(\secp)$ such that  
            
            \[\Pr\left[\Big|\widetilde{w}_{\pparam,y,d,0} - (-1)^s \widetilde{w}_{\pparam,y,d,1}\Big| \leq \nu(\secp) : \begin{array}{r}(\pparam,\sparam) \gets \Gen(1^\secp,1) \\ b \gets \{0,1\} \\ x \gets \cD_\pparam \\ y \coloneqq F_\pparam(b,x) \\ s \gets \PhaseInvert(\sparam,y,d)\end{array}\right] = 1-\negl(\secp).\]


        \end{itemize}

        \item \textbf{Mode indistinguishability.} For any QPT adversary $\{\Adv_\secp\}_{\secp \in \bbN}$,
        \[\Big|\Pr\left[\Adv_\secp(\pparam) = 1 : (\pparam,\sparam) \gets \Gen(1^\secp,0)\right] - \Pr\left[\Adv_\secp(\pparam) = 1 : (\pparam,\sparam) \gets \Gen(1^\secp,1)\right]\Big| = \negl(\secp).\]

    \end{itemize}

\end{definition}

\begin{lemma}[\cite{10.1007/978-3-031-68382-4_8}, Amplication for dTCFs]
    Any dTCF satisfying \cref{def:dTCF} implies a dTCF with efficient phase computation (\cref{def:phase-computation}).\footnote{Technically, \cite{10.1007/978-3-031-68382-4_8} assume that the dTCF input to their amplification lemma has claws that are perfectly balanced, as opposed to almost perfectly balanced, but it can be confirmed that their amplification lemma holds even when the claws are $1-\negl(\secp)$ balanced.}
\end{lemma}

Now, we show how to construct OSP from any dTCF satisfying \cref{def:phase-computation}.

\begin{theorem}
    Any dTCF with efficient phase computation (\cref{def:phase-computation}) implies two-round OSP.
\end{theorem}

\begin{proof}
    The construction is given in \cref{fig:two-round-OSP}. Security follows immediately from the mode indistinguishability property of the dTCF, so it remains to argue correctness.

    In the case that $b=0$, the efficient partial inversion and disjoint mode properties of the dTCF imply that with all but negligible probability, (1) all of $y$'s preimages begin with the same bit $s$, and thus the receiver's state on $\cB$ collapses to a standard basis state $\ket*{s}$, and (2) the $\PartialInvert$ algorithm will return this $s$ given $(\sparam,y)$, so the sender will obtain the correct description of the receiver's state.

    In the case that $b=1$, the receiver's state after measuring $y$ is
    \[\propto \sum_{x:F_\pparam(0,x)= y}\sqrt{\cD_\pparam(x)}\ket*{0}_\cB\ket*{x}_\cX + \sum_{x:F_\pparam(1,x)= y}\sqrt{\cD_\pparam(x)}\ket*{1}_\cB\ket*{x}_\cX.\] Thus, after measuring $\cX$ in the Hadamard basis to obtain $d$, the state collapses to a state
    \[\propto \sum_{x:F_\pparam(0,x) = y}(-1)^{d \cdot x}\sqrt{\cD_\pparam(x)}\ket*{0} + \sum_{x:F_\pparam(1,x) = y}(-1)^{d \cdot x}\sqrt{\cD_\pparam(x)}\ket*{1} = w_{\pparam,y,d,0}\ket*{0} + w_{\pparam,y,d,1}\ket*{1},\] which after normalizing, is equal to \[\widetilde{w}_{\pparam,y,d,0}\ket*{0} + \widetilde{w}_{\pparam,y,d,1}\ket*{1}.\] The efficient phase computation property implies that with all but negligible probability, the sender outputs $s$ such that $|\widetilde{w}_{\pparam,y,d,0} = (-1)^s \widetilde{w}_{\pparam,y,d,1}| = \negl(\secp)$, in which case the receiver's state is negligibly close to $Z^s\ket*{+}$. This completes the proof.

    \protocol{Two-round OSP from dTCF}{A consruction of two-round OSP from any dTCF with efficient phase computation (\cref{def:phase-computation}), which is known from any dTCF (\cref{def:dTCF}).}{fig:two-round-OSP}{

    \begin{itemize}
        \item $\OSP.\Sen(1^\secp,b)$: Sample $(\pparam,\sparam) \gets \Gen(1^\secp,b)$, and define $\msg_S \coloneqq \pparam$, and $\state_S \coloneqq \sparam$.
        \item $\OSP.\Rec(\msg_S)$: Prepare a state within negligible trace distance of $\ket*{+}_\cB \ket*{\psi_\pparam}_\cX$, apply $F_\pparam$ in superposition to a fresh register, and measure that register to obtain $y$. Then, measure register $\cX$ in the Hadamard basis to obtain $d \in \{0,1\}^n$. Finally, output the remaining qubit on register $\cB$ and set $\msg_R \coloneqq (y,d)$.
        \item $\OSP.\Dec(\state_S,\msg_R)$:
        \begin{itemize}
            \item If $b = 0$, compute $s \gets \PartialInvert(\sparam,y)$, and output $s$.
            \item If $b=1$, compute $s \gets \PhaseInvert(\sparam,y,d)$, and output $s$.
        \end{itemize}
    \end{itemize}
    
    }
\end{proof}

\section{Applications}

In this section, we move to the applications of OSP, establishing the following results.
\begin{itemize}
    \item \cref{subsec:PoQ}: OSP implies proofs of quantumness (as well as a ``test of a qubit'' and certifiable randomness).
    \item \cref{subsec:1-of-2}: Two-round OSP implies 1-of-2 puzzles, which previous work has shown implies (privately-verifiable) quantum money with classical communication, and position verification with classical communication.
    \item \cref{subsec:blind-delegation}: OSP implies blind classical delegation of quantum computation.
    \item \cref{subsec:verifiable-delegation}: Blind classical delegation generically implies verifiable classical delegation, so we can conclude that OSP implies verifiable classical delegation.
    \item \cref{subsec:encrypted-CNOT}: (Two-round) OSP implies (two-round) encrypted CNOT, which yields claw-state generators with indistinguishability security and (additionally assuming classical FHE with log-depth decryption) quantum FHE.
    
\end{itemize}

\subsection{Proofs of quantumness}\label{subsec:PoQ}

\begin{definition}[Proof of quantumness]\label{def:PoQ}
    A proof of quantumness protocol is an interaction between a QPT prover and a PPT verifier \[\{\top,\bot\} \gets \langle P(1^\secp),V(1^\secp)\rangle,\] where $\{\top,\bot\}$ is the output of the verifier. There exists $\epsilon(\secp), \delta(\secp)$ with $\epsilon(\secp)-\delta(\secp) = 1/\poly(\secp)$ such that the following properties hold.
    \begin{itemize}
        \item \textbf{Completeness.} \[\Pr\left[\top \gets \langle P(1^\secp),V(1^\secp)\rangle\right] \geq \epsilon(\secp).\]
        \item \textbf{Soundness.} For any PPT adversary $\{\Adv_\secp\}_{\secp \in \bbN}$, 
        \[\Pr\left[\top \gets \langle \Adv_\secp,V(1^\secp)\rangle\right] \leq \delta(\secp) + \negl(\secp).\]
    \end{itemize}
\end{definition}

\begin{theorem}\label{thm:PoQ}
    OSP (\cref{def:OSP}) implies a proof of quantumness (\cref{def:PoQ}).
\end{theorem}

\begin{proof}
    We describe the protocol in \cref{fig:PoQ}. The construction and proof follow the presentation in \cite{ABCC}, which is based on ideas originated in \cite{KCVY}.

    \protocol{Proof of quantumness from OSP}{Proof of quantumness from OSP.}{fig:PoQ}{
        \begin{itemize}
            \item The verifier samples $r \gets \{0,1\}$ and acts as the sender in an OSP with the prover. The verifier receives a bit $s$ and the prover receives a state (negligibly close to) $H^r\ket{s}$.
            \item The verifier samples $a \gets \{0,1\}$ and sends $a$ to the prover.
            \item If $a = 0$, the prover measures their state in the $X+Z$ basis, and if $a=1$, the prover measures their state in the $X-Z$ basis, to obtain a bit $b$. The prover sends $b$ to the verifier.
            \item The verifier accepts if $b=s \oplus r \cdot a$.
        \end{itemize}
    }

    First we show that the protocol has completeness $\epsilon > 0.85$. By applying gentle measurement (\cref{lemma:gentle-measurement}), we can take the prover's state at the conclusion of the OSP protocol to be exactly $H^r\ket{s}$, and only lose a $\negl(\secp)$ factor in the final bound. By a standard calculation, measuring $H^r\ket{s}$ in the X+Z basis yields $s$ with probability $\cos^2(\pi/8)$, and  measuring $H^r\ket{s}$ in the X-Z basis yields $s \oplus r$ with probability $\cos^2(\pi/8)$. Thus, the verifier accepts with probabilty at least $\cos^2(\pi/8) - \negl(\secp) > 0.85$.


    Now, we show that the protocol has soundness $\delta = 0.75$. To see this, suppose a PPT prover $\Adv = \{\Adv_\secp\}_{\secp \in \bbN}$ has advantage $0.75 + \nonnegl(\secp)$ in the protocol. Now consider the following procedure $\Adv'$ attacking the security of the OSP protocol.

    \begin{itemize}
        \item Interact as $\Adv$ in the OSP protocol and let $\state_\Adv$ be the state of $\Adv$ at the conclusion of the OSP protocol.
        \item Run $\Adv(\state_\Adv,0) \to b_0$ and $\Adv(\state_\Adv,1) \to b_1$ to obtain a final-round answer on each possible question $a \in \{0,1\}$ from the verifier.
        \item Output $b_0 \oplus b_1$ as the guess for $r$.
    \end{itemize}

    Let $D$ be the distribution over $(\state_\Adv,r,s)$ that results from running the OSP protocol with $\Adv$ on a random input $r$, and defining $s$ to be the sender's output. For any $(\state_\Adv,r,s)$ in the support of $D$ and any $a \in \{0,1\}$, define $p_{\state_\Adv,r,s}[a] \coloneqq \Pr[\Adv(\state_\Adv,a) = s \oplus r \cdot a]$.

    Then we have that 
    \begin{align*}
        \Pr\left[\Adv' = r\right] &\geq \E_{(\state_\Adv,r,s) \gets D}\left[p_{\state_\Adv,r,s}[0] \cdot p_{\state_\Adv,r,s}[1]\right] \\
        &\geq \E_{(\state_\Adv,r,s) \gets D}\left[p_{\state_\Adv,r,s}[0] + p_{\state_\Adv,r,s}[1] - 1\right] \\
        &= 2\E_{(\state_\Adv,r,s) \gets D}\left[\frac{1}{2}\left(p_{\state_\Adv,r,s}[0] + p_{\state_\Adv,r,s}[1]\right)\right] - 1 \\
        &= 2(0.75 + \nonnegl(\secp)) - 1 \\
        &= 0.5 + \nonnegl(\secp),
    \end{align*}

    where the second inequality follows from the fact that $x \cdot y \geq x + y -1$ for any $x,y \in [0,1]$, and the second equality follows from the fact that \[\E_{(\state_\Adv,r,s) \gets D}\left[\frac{1}{2}\left(p_{\state_\Adv,r,s}[0] + p_{\state_\Adv,r,s}[1]\right)\right]\] is exactly $\Adv$'s advantage in the proof of quantumness. This yields a contradiction to the security of OSP, completing the proof.

\end{proof}

To conclude this section, we note that \cref{fig:PoQ} fits the protocol template from \cite[Figure 1]{10.1007/978-3-031-38554-4_6} (unsurprisingly, since it is essentially a more modular presentation of the protocol from \cite[Section 5.3]{10.1007/978-3-031-38554-4_6}), and thus inherits the results established by \cite{10.1007/978-3-031-38554-4_6} about this class of proof of quantumness protocols. In particular it implies (1) a ``test of a qubit'' (\cite[Theorem 4.7]{10.1007/978-3-031-38554-4_6}), meaning that any \emph{quantum} prover with advantage close to $\cos^2(\pi/8)$ must be using (close to) anti-commuting operators in the final round, and (2) the ability to generate certifiable randomness from a quantum prover. We refer the reader to \cite{Brakerski2018ACT,10.1007/978-3-031-38554-4_6,Vid20-course,merkulov2023entropyaccumulationpostquantumcryptographic} for formal definitions and more discussion on the notions of test of a qubit and certifiable randomness. We stress that, in general, proofs of quantumness (e.g.\ Shor's algorithm \cite{10.1137/S0097539795293172} and Yamakawa-Zhandry \cite{10.1145/3658665}) do not always imply a test of a qubit, but our OSP-based proof of quantumness does since it fits the template described in \cite{10.1007/978-3-031-38554-4_6}.

\subsection{1-of-2 puzzles}\label{subsec:1-of-2}

Next, we show that two-round OSP implies a \emph{1-of-2 puzzle}, which was originally defined by \cite{10.1145/3318041.3355462}. In words, a 1-of-2 puzzle is a task with two challenges such that a prover can answer either but not both simultaneously. It is a useful abstraction, as it has been shown to imply both privately-verifiable quantum money \cite{10.1145/3318041.3355462} and position verification with classical communication \cite{liu_et_al:LIPIcs.ITCS.2022.100}. We begin by providing the definition.

\begin{definition}[1-of-2 puzzle \cite{10.1145/3318041.3355462}]\label{def:1-of-2}
    A 1-of-2 puzzle consists of four algorithms with the following syntax.
    \begin{itemize}
        \item $\KeyGen(1^\secp) \to (\pk,\vk)$. The PPT key generation algorithm takes as input the security parameter $1^\secp$ and outputs a public key $\pk$ and a secret verification key $\vk$.
        \item $\Obligate(\pk) \to (\ket{\psi},y)$: The QPT obligate algorithm takes as input the public key, and outputs a classical obligation string $y$ and a quantum state $\ket{\psi}$.
        \item $\Solve(\ket{\psi},b) \to a$: The QPT solve algorithm takes as input a state $\ket{\psi}$ and a bit $b \in \{0,1\}$ and outputs a string $a$.
        \item $\Ver(\vk,y,b,a) \to \{\top,\bot\}$: The PPT verify algorithm takes as input the verification key $\vk$, a string $y$, a bit $b \in \{0,1\}$, and a string $a$, and either accepts or rejects.
    \end{itemize}
    We say the puzzle is an $\epsilon(\secp)$-1-of-2 puzzle if it satisfies the following properties.
    \begin{itemize}
        \item \textbf{Completeness.}  \[\Pr\left[\top \gets \Ver(\vk,y,b,a) : \begin{array}{r} (\pk,\vk) \gets \KeyGen(1^\secp) \\ (\ket{\psi},y) \gets \Obligate(\pk) \\ b \gets \{0,1\} \\ a \gets \Solve(\ket{\psi},b) \end{array}\right] = 1-\negl(\secp).\]
        \item \textbf{Soundness.} For any QPT adversary $\{\Adv_\secp\}_{\secp \in \bbN}$, 
        \[\Pr\left[\begin{array}{l}\top \gets \Ver(\vk,y,0,a_0) \wedge \top \gets \Ver(\vk,y,1,a_1) \end{array} : \begin{array}{r}(\pk,\vk) \gets \KeyGen(1^\secp) \\ (y,a_0,a_1) \gets \Adv_\secp(\pk) \\ \end{array}\right] \leq \epsilon(\secp) + \negl(\secp).\]
    \end{itemize}
    
\end{definition}

We call a 1-of-2 puzzle \emph{strong} if $\epsilon = 0$, and note the following amplification theorem due to \cite{10.1145/3318041.3355462}. We will then construct an $\epsilon$-1-of-2 puzzle for $\epsilon = 0.5$ from OSP, which gives a strong 1-of-2 puzzle as a corollary.

\begin{theorem}[\cite{10.1145/3318041.3355462}]
    For any $\epsilon = 1-1/\poly(\secp)$, an $\epsilon$-1-of-2 puzzle implies a strong 1-of-2 puzzle.
\end{theorem}


\begin{theorem}
    Two-round OSP (\cref{def:OSP}) implies an $\epsilon$-1-of-2 puzzle (\cref{def:1-of-2}) for $\epsilon = 0.5$. 
\end{theorem}

\begin{proof}
    Let $(\OSP.\Sen,\OSP.\Rec,\OSP.\Dec)$ be any two-round OSP protocol (see \cref{def:OSP}). We define the 1-of-2 puzzle as follows.
    \begin{itemize}
        \item $\KeyGen(1^\secp)$: Sample $r \gets \{0,1\}$ and for $i \in [\secp]$, sample $(\msg_{S,i},\state_{S,i}) \gets \OSP.\Sen(1^\secp,r)$. Define $\pk \coloneqq (\msg_{S,1},\dots,\msg_{S,\secp})$ and $\vk \coloneqq (\state_{S,1},\dots,\state_{S,\secp})$.
        \item $\Obligate(\pk)$: For each $i \in [\secp]$, run $(\ket{\psi_i},\msg_{R,i}) \gets \OSP.\Rec(\msg_{S,i})$. Define $\ket{\psi} \coloneqq (\ket{\psi_1},\dots,\ket{\psi_\secp})$ and $y \coloneqq (\msg_{R,1},\dots,\msg_{R,\secp})$.
        \item $\Solve(\ket{\psi},b)$: If $b = 0$, measure each $\ket{\psi_{R,i}}$ in the $X+Z$ basis and if $b = 1$, measure each $\ket{\psi_{R,i}}$ in the $X-Z$ basis. This results in a string $a \in \{0,1\}^\secp$.
        \item $\Ver(\vk,y,b,a)$: For each $i \in [\secp]$, set $s_i \coloneqq \OSP.\Dec(\state_{S_i},\msg_{R,i})$, define $s \coloneqq (s_1,\dots,s_\secp)$, and define $s \oplus r \coloneqq (s_1 \oplus r, \dots, s_\secp \oplus r)$. If $b = 0$, accept iff $\Delta(a,s) \geq 0.85$ and if $b=1$, accept iff $\Delta(a,s \oplus r) \geq 0.85$.
    \end{itemize}

    A standard calculation shows that for each $i \in [\secp]$, $\Pr[a_i = s_i] = \cos^2(\pi/8) > 0.85$ if $b = 0$ and $\Pr[a_i = s_i \oplus r] = \cos^2(\pi/8) > 0.85$ if $b = 1$. Moreover, for each $b \in \{0,1\}$, these $\secp$ events are independent. Thus, by a standard tail bound, for each $b \in \{0,1\}$ we have that $\Pr[\Ver(\vk,y,b,a) = \top] = 1-\negl(\secp)$.

    Now, suppose towards contradiction that there exists a QPT $\Adv = \{\Adv_\secp\}_{\secp \in \bbN}$ such that 
    \[\Pr\left[\begin{array}{l}\top \gets \Ver(\vk,y,0,a_0) \\ \wedge \ \top \gets \Ver(\vk,y,1,a_1)\end{array} : \begin{array}{r}(\pk,\vk) \gets \KeyGen(1^\secp) \\ (y,a_0,a_1) \gets \Adv_\secp(\pk) \\ \end{array}\right] = \frac{1}{2} + \nonnegl(\secp).\]

    This implies that with probability $1/2+\nonnegl(\secp)$, the majority bit in $a_0 \oplus a_1$ is equal to $r$. However, by a standard hybrid argument, the security of the two-round OSP implies that $r$ cannot be predicted with better than $\negl(\secp)$ advantage, a contradiction.
    
\end{proof}

\begin{corollary}
    Two-round OSP implies a strong 1-of-2 puzzle.
\end{corollary}

In turn, we obtain the following results as corollaries from prior work.
\begin{itemize}
    \item \cite{10.1145/3318041.3355462}: Two-round OSP implies privately-verifiable quantum money with classical communication.
    \item \cite{liu_et_al:LIPIcs.ITCS.2022.100}: Two-round OSP implies position verification with classical communication. 
\end{itemize}

\subsection{Blind delegation}\label{subsec:blind-delegation}

We first give a very general definition for blind delegation of quantum computation with a classical client. It allows the client to delegate the computation of an arbitrary publicly-known quantum operation that takes a quantum input (provided by the server, and potentially entangled with an auxiliary system held by the server) and a private classical input (chosen by the client). After interaction, the server is able to obtain the (potentially quantum) output up to a one-time pad with keys known to the client.


\begin{definition}[Blind Classical Delegation of Quantum Computation]\label{def:blind-computation}
Let $\cH_\cV,\cH_\cW$ be Hilbert spaces of arbitrary dimension, and let $Q: \{0,1\}^* \times \cH_\regV \to \cH_\regW$ be a polynomial-size quantum circuit that takes as input a classical string $x$ and a state on register $\regV$, and outputs a state on register $\regW$. A protocol for blind classical delegation of quantum computation consists of an interaction 

\[(\regW,(r,s)) \gets \langle S(1^\secp,Q,\regV),C(1^\secp,Q,x)\rangle\]

between 

\begin{itemize}
    \item a QPT server $S(1^\secp,Q,\regV)$ with input the security parameter $1^\secp$, circuit $Q$, and state on register $\regV$, and
    \item a PPT client $C(1^\secp,Q,x)$ with input the security parameter $1^\secp$, circuit $Q$, and classical string $x$.
\end{itemize}

At the end of the interaction, the server outputs a state on register $\regW$ and the client outputs classical strings $(r,s)$. The protocol must satisfy the following properties.
\begin{itemize}
    \item \textbf{Correctness.} For any circuit $Q$ and input $x$, let $\mathsf{IDEAL}[Q,x]$ be the map from $\regV \to \regW$ defined by $Q(x,\cdot)$, and let $\mathsf{REAL}[Q,x]_\secp$ be the map from $\regV \to \regW$ induced by running the protocol \[(\regW,(r,s)) \gets \langle S(1^\secp,Q,\regV),C(1^\secp,Q,x)\rangle\] and then applying $X^r Z^s$ to $\regW$. Then for any $Q$ and $x$,
    \[D_\diamond\left(\mathsf{REAL}[Q,x]_\secp, \mathsf{IDEAL}[Q,x]\right) = \negl(\secp).\]



    \item \textbf{Security.} For any circuit $Q$, QPT adversary $\{\Adv_\secp\}_{\secp \in \bbN}$, and two strings $x_0,x_1$, it holds that 
    \begin{align*}&\bigg| \Pr\left[b_\Adv = 0 :  (b_\Adv,(r,s)) \gets \langle \Adv_\secp,C(1^\secp,Q,x_0)\rangle\right] \\ &-\Pr\left[b_\Adv = 0 :  (b_\Adv,(r,s)) \gets \langle \Adv_\secp,C(1^\secp,Q,x_1)\rangle\right] \bigg| = \negl(\secp), 
    \end{align*}
    where $b_\Adv$ denotes a single bit output by $\Adv_\secp$ after the interaction (which could result from an arbitrary QPT operation applied to its state after the interaction).
\end{itemize}
\end{definition}

\begin{remark}\label{remark:classical-output}
    Note that the above definition implies that if $Q$ has a \emph{classical} output, then the client can obtain this output with one extra message from the server. That is, suppose $Q : \{0,1\}^* \times \cH_\cV \to \{0,1\}^*$. Then at the conclusion of the protocol defined above, the server has a classical output $y \oplus r$, and the client has the one-time pad key $r$ (note that $s$ is irrelevant once the output has been measured in the standard basis). Then, the server returns $y \oplus r$ to the client, who recovers the output $y$. In this case, we denote the protocol 
    \[y \gets \langle S(1^\secp,Q,\cV),C(1^\secp,Q,x)\rangle,\] where $y$ is the client's output.
\end{remark}

Towards showing that OSP implies blind classical delegation of arbitrary quantum computation, we first show that it implies what we call an ``encypted phase'' protocol. 

\begin{definition}[Encrypted phase]\label{def:encrypted-phase}
    A protocol for encrypted phase is the special case of blind classical delegation of quantum computation where $\cH_\regV = \cH_\regW$ is a single-qubit register, the client's private input is a bit $b \in \{0,1\}$, and $Q(b,\regV)$ is the identity if $b=0$ and applies a phase gate $P$ to $\regV$ if $b = 1$.
\end{definition}

\begin{lemma}
    OSP (\cref{def:OSP}) implies encrypted phase.
\end{lemma}

\begin{proof}
    We first describe the protocol.

    \begin{itemize}
        \item The server and client begin by running an OSP protocol, where the client plays the role of the sender with bit $b$ equal to the client's input bit $b$, and the server plays the role of the receiver. Then, the server applies a Hadamard gate followed by a $\sqrt{X}$ gate to their output state. Up to a negligible trace distance, this results in the state
        \begin{itemize}
            \item $Z^s\ket{+}$ if $b=0$,
            \item $Z^sP\ket{+}$ if $b=1$,
        \end{itemize}
        where $s$ is the sender's output bit. Let $\cM$ be the name of the output state's register.
        \item Next, the server applies a CNOT gate from register $\cV$ to register $\cM$, and then measures $\cM$ in the standard basis to obtain a bit $m$. The server sends $m$ to the client, and outputs $\cV$.
        \item If $b=0$, the client outputs $(0,s)$ and if $b=1$, the client outputs $(0,s \oplus m)$.
    \end{itemize}

    Security follows immediately from the security of OSP, so it remains to check correctness. Let $\alpha_0\ket{0}_\cV+\alpha_1\ket{1}_\cV$ be the input state on register $\cV$ (note that it could be entangled with some auxiliary register, i.e.\ $\alpha_0\ket{0}_\cV\ket{\psi_0}+\alpha_1\ket{1}_\cV\ket{\psi_1}$, but we suppress writing the auxiliary register to avoid clutter).
    
    In the case $b=0$, we have that 
    \begin{align*}              \CNOT&\left(\alpha_0\ket{0}_\cV+\alpha_1\ket{1}_\cV\right) \otimes Z^s\ket{+}_\cM \\ &= (Z^s_\cV \otimes Z^s_\cM)\CNOT\left(\alpha_0\ket{0}_\cV+\alpha_1\ket{1}_\cV\right) \otimes \ket{+}_\cM\\ &= (Z^s_\cV \otimes Z^s_\cM)\left(\alpha_0\ket{0}_\cV+\alpha_1\ket{1}_\cV\right) \otimes \ket{+}_\cM,
    \end{align*}

    so measuring $\cM$ does not affect the state on $\cV$. Thus, the server is left with their original state up to a $Z^s$ error, which is the desired outcome.

    In the case $b=1$, we have that 
    \begin{align*}
    \CNOT&\left(\alpha_0\ket{0}_\cV+\alpha_1\ket{1}_\cV\right) \otimes PZ^s\ket{+}_\cM  \\
    &=\frac{1}{\sqrt{2}}(Z^s_\cV \otimes Z^s_\cM)\CNOT\left(\alpha_0\ket{0}_\cV + \alpha_1\ket{1}_\cV\right) \otimes \left(\ket{0}_\cM + i\ket{1}_\cM\right) \\
    &=\frac{1}{\sqrt{2}}(Z^s_\cV \otimes Z^s_\cM)\left(\alpha_0\ket{0}_\cV\ket{0}_\cM + i \alpha_0\ket{0}_\cV\ket{1}_\cM + \alpha_1\ket{1}_\cV\ket{1}_\cM + i\alpha_1\ket{1}_\cV\ket{0}_\cM\right) \\
    &=\frac{1}{\sqrt{2}}(Z^s_\cV \otimes Z^s_\cM)\left(\left(\alpha_0\ket{0}_\cV+i\alpha_1\ket{1}_\cV\right)\ket{0}_\cM + \left(i\alpha_0\ket{0}_\cV+\alpha_1\ket{1}_\cV\right)\ket{1}_\cM\right).
    \end{align*}

    So if the measured bit $m=0$, the remaining state is
    \[Z^s_\cV\left(\alpha_0\ket{0}_\cV + i \alpha_1\ket{1}_\cV\right),\] which is the desired outcome, and if the measured bit $m=1$, the remaining state is 
    \[Z^s_\cV\left(i\alpha_0\ket{0}_\cV +  \alpha_1\ket{1}_\cV\right) = Z^s_\cV\left(\alpha_0\ket{0}_\cV -  i\alpha_1\ket{1}_\cV\right) = Z^{s \oplus 1}_\cV\left(\alpha_0\ket{0}_\cV +  i\alpha_1\ket{1}_\cV\right),\] which is again the desired outcome.

\end{proof}

Next, we show that the ability to perform an encrypted phase implies blind classical delegation of quantum computation for arbitrary quantum operations.  

\begin{theorem}
    Any protocol for encrypted phase (\cref{def:encrypted-phase}) implies blind classical delegation of quantum computation (\cref{def:blind-computation}).
\end{theorem}

\begin{proof}
    Given any circuit $Q$, write it using Clifford and $T^\dagger$ gates, where $T^\dagger = \sqrt{P^\dagger}$. That is, $Q$ can be written as $C_{\ell+1}T^\dagger C_\ell T^\dagger \dots C_2 T^\dagger C_1$, where each $C_i$ is Clifford. We will use the fact that for any polynomial-size Clifford $C$ and one-time padded state $X^r Z^s \ket{\psi}$, it holds that $CX^r Z^s \ket{\psi} = X^{r'}Z^{s'}C\ket{\psi}$, where $r'$ and $s'$ are efficiently computable from $C, r,$ and $s$. Moreover, for a single qubit state $\ket{\psi}$, it holds that $T^\dagger X^r Z^s \ket{\psi} = (P^\dagger)^r X^r Z^s T^\dagger\ket{\psi}$.

    Given a protocol for encrypted phase (\cref{def:encrypted-phase}), we describe a protocol for blind classical delegation of quantum computation:

    \begin{itemize}
    \item Given inputs $(1^\secp,Q,x)$, the client writes $Q = C_{\ell+1}T^\dagger C_\ell T^\dagger \dots C_2 T^\dagger C_1$, where each $C_i$ is Clifford and the $i$'th $T^\dagger$ gate is applied to qubit $t_i$ of the server's register $\regV$.
    \item The client samples a classical one-time pad $r_{\mathsf{inp}}$ and sends $x \oplus r_{\mathsf{inp}}$ to the server. Throughout the protocol, the client will hold the classical keys for a quantum one-time pad applied to the server's register $\regV$ (which we consider to now include the encrypted input $x \oplus r_{\mathsf{inp}}$). The client initializes these keys to $(r,s) \coloneqq ((r_{\mathsf{inp}},0,\dots,0), (0,\dots,0))$.
    \item For $i \in [\ell]$, perform the following steps.
    \begin{itemize}
        \item The server applies $C_i$ to register $\cV$, and the client updates the quantum one-time pad keys $(r,s)$ according to $C_i$.
        \item The server applies a $T^\dagger$ gate to qubit $t_i$.
        \item Let $b_i$ be the $X$-bit of the one-time pad key on qubit $t_i$. The server and client apply an encrypted phase (\cref{def:encrypted-phase}) to qubit $t_i$ with client input equal to $b_i$.
        \item The client uses their output $(r_i,s_i)$ from the encrypted phase to update the one-time pad key on qubit $t_i$.
    \end{itemize}
    \item The server applies the final Clifford $C_{\ell+1}$ to register $\cV$ and outputs $\cV$, and the client does the final one-time pad update according to $C_{\ell+1}$ and outputs their final one-time pad keys $(r_\out,s_\out)$.

    \end{itemize}

Correctness is straightforward using the properties listed prior to the description of the protocol. Security follows from a standard hybrid argument: Starting from the final $T^\dagger$ gate, we switch the client's input to the encrypted phase to 0. Once this is completed, the protocol no longer depends on the client's initial classical one-time pad $r_{\mathsf{inp}}$, and thus we can switch between client inputs $x_0$ and $x_1$ without affecting the view of the server.
\end{proof}

Thus, we obtain the following corollary.

\begin{corollary}\label{cor:blind}
    OSP implies blind classical delegation of quantum computation (\cref{def:blind-computation}).
\end{corollary}

In fact, we also observe the following, which shows that OSP is \emph{both} necessary and sufficient for blind classical delegation of quantum computation.

\begin{claim}
    Blind classical delegation of quantum computation according to \cref{def:blind-computation} implies OSP.
\end{claim}

\begin{proof}
    Let $Q$ be the circuit that takes as input a bit $b$ (and no quantum input) and outputs the state $H^b\ket{0}$. To perform OSP, the sender inputs their bit $b$ to a classical delegation protocol for $Q$, with the receiver acting as the server. The receiver ends up with a state (negligibly close to) $X^rZ^sH^b\ket{0} = H^b\ket{x}$, where $x = r$ if $b=0$ and $x=s$ if $b=1$, and the sender ends up with $(r,s)$. Thus, the sender outputs $x = r$ if $b=0$ and $x=s$ if $b=1$ to complete the description of the OSP protocol.
\end{proof}

\subsection{Verifiable delegation}\label{subsec:verifiable-delegation}

In this section, we show that \emph{blind} classical delegation of quantum computation (\cref{def:blind-computation}) generically implies \emph{verifiable} classical delegation of quantum computation, which then gives verifiable delegation from OSP as a corollary. Previously, it was shown by \cite{NZ23} that QFHE, which is a special case of blind delegation, implies verifiable delegation. Here, we observe that their approach, which builds on the ``KLVY compiler'' of \cite{KLVY} can be generalized to establish the result from any blind delegation. We begin by defining verifiable delegation, which we call classical verification of quantum computation (CVQC).

\begin{definition}[Classical Verification of Quantum Computation]\label{def:CVQC}
    Classical verification of quantum computation (CVQC) is an interaction between a QPT prover and PPT verifier on input an instance $x$ \[\{\top,\bot\} \gets \langle P(1^\secp,x),V(1^\secp,x)\rangle,\] where $\{\top,\bot\}$ is the verifier's output. For any language $\cL$ in BQP, there exists some $\epsilon(\secp), \delta(\secp)$ with $\epsilon(\secp)-\delta(\secp) = 1/\poly(\secp)$ such that the following properties hold.
    \begin{itemize}
        \item \textbf{Completeness.} For any $x \in \cL$, \[\Pr[\top \gets \langle P(1^\secp,x),V(1^\secp,x)\rangle] \geq \epsilon(\secp).\]
        \item \textbf{Soundness.} For any $x \notin \cL$ and QPT adversary $\{\Adv_\secp\}_{\secp \in \bbN}$, \[\Pr\left[\top \gets \langle \Adv_\secp,V(1^\secp,x)\right] \leq \delta(\secp) + \negl(\secp).\]
    \end{itemize}
\end{definition}

This section is dedicated to proving the following theorem.

\begin{theorem}\label{thm:CVQC}
    Blind classical delegation of quantum computation (\cref{def:blind-computation}) implies classical verification of quantum computation (\cref{def:CVQC}).
\end{theorem}

\begin{corollary}
    OSP implies classical verification of quantum computation.
\end{corollary}

\subsubsection{Nonlocal games and the KLVY compiler}

Towards proving this theorem, we recall the KLVY compiler, which uses QFHE to compile any nonlocal game into a single-prover protocol. In this subsection, we observe that the KLVY compiler can be instantiated with any (potentially interactive, non-compact) blind delegation protocol, which yields what we call the \emph{generalized KLVY compiler}. We also define a set of nonlocal game strategies that we call \emph{computationally nonlocal strategies}. This definition provides a clean way to establish soundness of the generalized KLVY compiler - soundness of any compiled game can be upper-bounded by the value of the best computationally nonlocal strategy for that game.

First, we present standard definitions of (families of) nonlocal games, as well as nonlocal strategies for these games.

\begin{definition}[Nonlocal game]
    A \emph{nonlocal game} $G = (Q,V)$ is specified by a distribution $Q$ over pairs $(x,y) \in \{0,1\}^{n_1} \times \{0,1\}^{n_2}$ and a verification predicate $V(x,y,a,b) \in \{0,1\}$, where $a \in \{0,1\}^{m_1}$ and $b \in \{0,1\}^{m_2}$. A \emph{family} of nonlocal games $\cG = \{\cG_\secp\}_{\secp \in \bbN}$ is a set of games parameterized by $\secp$, where each $\cG_\secp$ is itself a set of games $G \in \cG_\secp$. Each game $G \in \cG$ is defined by a distribution $Q_G$ over pairs $(x,y) \in \{0,1\}^{n_{1,G}} \times \{0,1\}^{n_{2,G}}$ and a verification predicate $V_G(x,y,a,b) \in \{0,1\}$, where $a \in \{0,1\}^{m_{1,G}}$ and $b \in \{0,1\}^{m_{2,G}}$. For any game $G \in \cG$, we define $\secp_G$ to be the choice of $\secp$ such that $G \in \cG_\secp$. We say that the family of games is \emph{efficient} if there exists a polynomial $p(\cdot)$ and a procedure that, for any $G \in \cG$, samples from $Q_G$ and computes $V_G$ in time $p(\secp_G)$.
\end{definition}

\begin{definition}[Nonlocal strategy]
    A \emph{nonlocal strategy} $\mathscr{S}$ for game $G$ consists of the following.
    \begin{itemize}
        \item A state $\ket{\psi} \in \cH_\cA \otimes \cH_\cB$.
        \item For every $x \in \{0,1\}^{n_1}$, a projective measurement $\{A_a^x\}_{a}$ acting on $\cH_\cA$ with outcomes $a \in \{0,1\}^{m_1}$.
        \item For every $y \in \{0,1\}^{n_2}$, a projective measurement $\{B_b^y\}_{b}$ acting on $\cH_\cB$ with outcomes $b \in \{0,1\}^{m_2}$.
    \end{itemize}
    The \emph{value} of this strategy is given by 
    \[\omega(G,\mathscr{S}) \coloneqq \E_{(x,y) \gets Q}\sum_{a,b} V(x,y,a,b) \bra{\psi}A_a^x \otimes B_b^y\ket{\psi}.\]


    A strategy $\mathscr{S}$ for a \emph{family} of games $\cG$ consists of a strategy

    \[\mathscr{S}_G = \left(\ket{\psi_G},\{A^x_{a,G}\}_a,\{B^y_{b,G}\}_b\right)\]

    for each $G \in \cG$. We say that $\mathscr{S}$ is \emph{efficient} if $\ket{\psi}$ is QPT-preparable and $\{A_{a,G}^x\}_a$ and $\{B_{b,G}^y\}_b$ are QPT-implementable.
\end{definition}

Next, we present our new definition of a \emph{computationally nonlocal strategy}.    

\begin{definition}[Computationally nonlocal strategy]\label{def:comp-non-local} A \emph{computationally nonlocal strategy} $\mathscr{C}$ for a family of games $\cG = \{\cG_\secp\}_\secp$ consists of the following for each $G \in \cG$.

    \begin{itemize}
        \item A QPT-preparable state $\ket{\psi_G} \in \cH_{\cA,G} \otimes \cH_{\cB,G}$.
        \item For every $x \in \{0,1\}^{n_{1,G}}$, a QPT-implementable unitary $U^x_G$ acting on $\cH_{\cA,G} \otimes \cH_{\cB,G}$. For each $a \in \{0,1\}^{m_{1,G}}$, define \[A^x_{a,G} \coloneqq \ketbra{a}{a}U^x_G,\] where the projection $\ketbra{a}{a}$ is applied to some specified sub-register of $\cH_{\cA,G}$.
        \item For every $y \in \{0,1\}^{n_{2,G}}$, a QPT-implementable projective measurement $\{B^y_{a,G}\}_{b}$ acting on $\cH_{\cB,G}$ with outcomes $b \in \{0,1\}^{m_{2,G}}$.
    \end{itemize}
    For this to be a valid computationally nonlocal strategy, the ``Alice'' operations must satisfy the following property. There exists a negligible function $\mu(\secp)$ such that for any sequence of games $\{G_\secp\}_{\secp \in \bbN}$ where each $G_\secp \in \cG_\secp$, QPT distinguisher $\{M_\secp,I-M_\secp\}_\secp$ \textbf{acting only on} $\cH_{\cB,G_\secp}$, and sequence of pairs of questions $\{x_{0,\secp},x_{1,\secp} \in \{0,1\}^{n_{1,G_\secp}}\}_{\secp}$,

    \[\bigg| \sum_a \bra{\psi_{G_\secp}}A_{a,G_\secp}^{x_{0,\secp},\dagger}M_\secp A_{a,G_\secp}^{x_{0,\secp}}\ket{\psi_{G_\secp}}- \sum_a \bra{\psi_{G_\secp}}A_{a,G_\secp}^{x_{1,\secp},\dagger}M_\secp A_{a,G_\secp}^{x_{1,\secp}}\ket{\psi_{G_\secp}}\bigg| \leq \mu(\secp).\]

    That is, Alice's questions must be computationally hidden by the state passed from the Alice operation to the Bob operation. The \emph{value} of this strategy is a function $\omega(G,\mathscr{C})$ of the game $G \in \cG$, defined by \[\omega(G,\mathscr{C}) \coloneqq \E_{(x,y) \gets Q_G}\sum_{a,b}V_G(x,y,a,b)\bra{\psi_G}A^{x,\dagger}_{a,G}B^y_{b,G}A^x_{a,G}\ket{\psi_G}.\]
    
    We say that the computationally nonlocal value of $\cG$ is upper-bounded by a function $\omega_C(G)$ if for all computationally nonlocal strategies $\mathscr{C}$, there exists a negligible function $\nu(\secp)$ such that for any sequence of games $\{G_\secp\}_{\secp \in \bbN}$ where each $G_\secp \in \cG_\secp$, it holds that 
    \[\omega(G_\secp,\mathscr{C}) \leq \omega_C(G_\secp) + \nu(\secp).\]
\end{definition}

\begin{remark}
    We will sometimes refer to a computationally nonlocal strategy $\mathscr{C}$ for some \emph{fixed} game $G$ (such as the CHSH game). In this case, we view $G$ as a family of games $\cG = \{\cG_\secp\}_\secp = \{G\}_\secp$. That is, each $\cG_\secp$ just consists of $G$ itself, and $\mathscr{C}$ consists of one strategy for each $\secp \in \bbN$.
\end{remark}

The following remark formalizes a simple claim about computational indistinguishability.

\begin{remark}
    By following the arguments in \cite[Lemma 7, Lemma 8]{NZ23}, it is straightforward to see that the Alice operations for any computationally nonlocal strategy $\mathscr{C}$ must also satisfy the following. There exists a negligible function $\mu(\secp)$ such that for any sequence of games $\{G_\secp\}_{\secp \in \bbN}$ where each $G_\secp \in \cG_\secp$, QPT-implementable POVM $\{\{M_{\gamma,\secp}\}_\gamma\}_\secp$ with outcomes in $\gamma \in [0,1]$ acting only on $\cH_{\cB,G_\secp}$, and sequence of pairs of questions $\{x_{0,\secp},x_{1,\secp} \in \{0,1\}^{n_{1,G_\secp}}\}_{\secp}$,
    \[\bigg| \sum_a \sum_\gamma \gamma \bra{\psi_{G_\secp}} A^{x_{0,\secp},\dagger}_{a,G_\secp} M_{\gamma,\secp} A^{x_{0,\secp}}_{a,G_\secp} \ket{\psi_{G_\secp}} - \sum_a \sum_\gamma \gamma \bra{\psi_{G_\secp}} A^{x_{1,\secp},\dagger}_{a,G_\secp} M_{\gamma,\secp} A^{x_{1,\secp}}_{a,G_\secp} \ket{\psi_{G_\secp}}\bigg| \leq \mu(\secp).\]
\end{remark}

Next, we state the ``generalized'' version of the KLVY compiler, where in place of a QFHE protocol, we use any blind classical delegation of quantum computation protocol (\cref{def:blind-computation}).

\begin{definition}[Generalized KLVY Compiler]
    Let $\cG$ be a family of nonlocal games, let $\Pi = \langle S,C \rangle$ be a blind classical delegation of quantum computation protocol, and let $\mathscr{S}$ be an efficient nonlocal strategy for $\cG$. For each $G \in \cG$, let $A_G : \{0,1\}^{n_{1,G}} \times \cH_{\cA} \to \{0,1\}^{m_{1,G}}$ be the QPT circuit that performs the Alice measurement of $\mathscr{S}_G$ and $B_G : \{0,1\}^{n_{2,G}} \times \cH_{\cB} \to \{0,1\}^{m_{2,G}}$ be the QPT circuit that performs the Bob measurement of $\mathscr{S}_G$. The \emph{KLVY-compiled} protocol $\KLVY[\cG,\Pi,\mathscr{S}]$ is an interaction between a QPT prover $\Prove$ and a PPT verifier $\Ver$ that is parameterized by a game $G \in \cG_\secp \subset \cG$, and operates as follows.
    \begin{enumerate}
        \item The verifier samples $(x,y) \gets Q_G$.
        \item Let $\ket{\psi_G} \in \cH_{\cA} \otimes \cH_{\cB}$ be the initial state used in $\mathscr{S}_G$. The prover and verifier engage in a blind classical delegation of quantum computation protocol (with classical output, see \cref{remark:classical-output}) \[a \gets \langle S(1^\secp,A_G,\ket{\psi_G}), C(1^\secp,A_G,x)\rangle,\] with the prover playing the role of the server $S$ and the verifier playing the role of the client $C$.
        \item The verifier sends $y$ to the prover.
        \item The prover runs $b \gets B_G(y,\cB)$ and sends $b$ to the verifier.
        \item The verifier outputs $V_G(x,y,a,b)$.
    \end{enumerate}
    This interaction is denoted \[b_\Ver \gets \langle \Prove(1^\secp,G),\Ver(1^\secp,G)\rangle,\] where $b_\Ver$ denotes the bit output by $\Ver$.

    The \emph{completeness} of $\KLVY[\cG,\Pi,\mathscr{S}]$ is defined by a function \[c(G) \coloneqq \Pr\left[b_\Ver = 1 : b_\Ver \gets \langle \Prove(1^\secp,G),\Ver(1^\secp,G)\rangle\right],\] and $\KLVY[\cG,\Pi,\mathscr{S}]$ has soundness $s(G)$ if for any QPT adversary $\{\Adv_\secp\}_{\secp \in \bbN}$
    \[\Pr\left[b_\Ver = 1 : b_\Ver \gets \langle \Adv_{\secp_G}(1^\secp,G),\Ver(1^\secp,G)\rangle\right] \leq s(G).\]
    
\end{definition}

Finally, we prove the main theorem of this subsection. Essentially, we show that the computational nonlocal value of any game $G$ upper bounds the soundness of the (generalized) KLVY compiled version of $G$.

\begin{theorem}\label{thm:KLVY}
    Let $\cG$ be a family of nonlocal games, let $\Pi$ be a blind classical delegation of quantum computation protocol, and let $\mathscr{S}$ be an efficient nonlocal strategy for $\cG$. Then the completeness of $\KLVY[\cG,\Pi,\mathscr{S}]$ satisfies \[c(G) \geq \omega(G,\mathscr{S})-\negl(\secp_G),\] and, for any upper bound $\omega_C(G)$ on the computationally nonlocal value of $\cG$, $\KLVY[\cG,\Pi,\mathscr{S}]$ has soundness \[s(G) \leq \omega_C(G) + \negl(\secp_G),\] where $\secp_G$ is defined to be the $\secp$ such that $G \in \cG_\secp$.
\end{theorem}

\begin{proof}
    The completeness claim follows directly from the correctness of the blind delegation protocol $\Pi$. The soundness claim follows by observing that any $\Adv$ interacting in $\KLVY[\cG,\Pi,\mathscr{S}]$ defines a computationally nonlocal strategy for $\cG$. This can be argued as follows, where for convenience we will drop parameterizations by $G$ and $\secp$. Consider purifying the interaction between $\Adv$ and $\Ver$ during the second step of the protocol (the blind delegation step), and then measuring the verifier's output $a$. For each verifier input $x$, this defines a unitary $U^x$ that is applied to initial state $\ket{0}_\cA\ket{\psi}_\cB$, where $\ket{\psi}$ is the initial state of $\Adv$, $\cH_\cB$ includes the working register of $\Adv$ along with the register that holds the transcript of communication between $\Adv$ and $\Ver$, and $\cH_\cA$ is the register required to implement the operations of $\Ver$. Then $A^x_a$ is defined as $\ketbra{a}U^x$, i.e.\ the unitary $U^x$ followed by a standard basis projection onto the verifier's output $a$ on a sub-register of $\cH_\cA$. It follows immediately from the security of $\Pi$ that there exists a negligible function $\mu(\secp)$ such that for any QPT distinguisher $\{M,I-M\}$ acting only on $\cH_\cB$, and $x_0,x_1 \in \{0,1\}^{n_1}$,

     \[\bigg| \sum_a \bra{\psi}\bra{0}A_a^{x_0,\dagger}M A_a^{x_0}\ket{\psi}\ket{0}- \sum_a \bra{\psi}\bra{0}A_a^{x_1,\dagger}M A_a^{x_1}\ket{\psi}\ket{0}\bigg| = \negl(\secp),\] 

     which shows that this is a valid computationally nonlocal strategy.

\end{proof}

\subsubsection{The CHSH game}\label{subsec:CHSH}

Next, we recall the CHSH game, which is an important building block in the CVQC protocol of \cite{NZ23}. 

\begin{definition}[CHSH]
    The CHSH game is defined by question distribution $(x,y)  \gets \{0,1\} \times \{0,1\}$ and predicate $V(x,y,a,b)$ that accepts iff $x \cdot y = a \oplus b$, where answers $a,b \in \{0,1\}$.
\end{definition}

An important lemma from \cite{NZ23} establishes a \emph{rigidity} property of the KLVY-compiled CHSH game. That is, any strategy in the (QFHE-based) KLVY-compiled game that approaches the optimal value of $\cos^2(\pi/8)$ for nonlocal strategies must be such that the Bob operations nearly anti-commute. It is not hard to see that the same holds for \emph{any} computationally nonlocal strategy,\footnote{Note that strategies in the QFHE-based KLVY-compiled game are a special case of computationally nonlocal strategies.} and in fact this claim was already essentially shown by \cite[Theorem 4.7]{10.1007/978-3-031-38554-4_6}.

\begin{theorem}\label{thm:CHSH-anticommutation}
    Let $\omega_\CHSH = \cos^2(\pi/8)$. Fix any computationally nonlocal strategy for the CHSH game with value $\omega_\CHSH - \delta(\secp)$, and, dropping the parameterization by $\secp$,  let $B^0 \coloneqq B^0_0 - B^0_1, B^1 \coloneqq B^1_0 - B^1_1$ be binary observables defined by the ``Bob'' measurements. Then for any $x \in \{0,1\}$, it holds that \[\sum_{a \in \{0,1\}}\bra{\psi}A_a^{x,\dagger}\{B^0,B^1\}^2A_a^x\ket{\psi} \leq O(\delta(\secp)) + \negl(\secp).\] 
\end{theorem}

\begin{proof}
    This follows readily from \cite[Theorem 4.7]{10.1007/978-3-031-38554-4_6}, following their arguments in Section 5.4.\footnote{Our formalization of computationally nonlocal strategies is slightly more general than \cite{10.1007/978-3-031-38554-4_6}'s quantum device formalization (since a computationally nonlocal strategy does not necessarily have to be defined by some prover interacting in a protocol), but it is straightforward to verify that their techniques apply to our more general setting. } It can also be verified that a proper adaptation of Lemma 34 from \cite{NZ23} yields this claim using a different proof technique.
\end{proof}

\subsubsection{The \cite{NZ23} BQP verification game}


Now, we define the \cite{NZ23} family of nonlocal games for verifying arbitrary BQP computation.\footnote{We note that the same set of games should suffice for verifying arbitrary QMA languages as well, by providing sufficient copies of the witness state to the prover, but we follow \cite{NZ23} and stick to BQP for this exposition.} The family $\cG$ consists of a game $G[H]$ for each XX/ZZ local Hamiltonian \[H = \sum_{W \in \{X,Z\},i\neq j \in [\secp]}p_{W,i,j}W(e_i + e_j),\] where $\sum_{W,i,j} p_{W,i,j} = 1$. Before coming to the formal specification of the nonlocal game $G[H]$ associated with $H$, we provide a high-level overview. 

To begin with, two provers Alice and Bob share several halves of EPR pairs, and Alice prepares a minimum eigenvalue state $\ket{\psi}$ for $H$. Alice's question instructs her to either (1) \textbf{Teleport} the state to Bob, (2) participate in an anti-commutation (\textbf{CHSH}) game, or (3) participate in a \textbf{Commutation} game. Bob's question is always a single bit that instructs him to either measures all of his halves of EPR pairs in the standard basis or the Hadamard basis, and return the results. 

In the teleportation case, an honest Alice simply teleports $\ket{\psi}$ to Bob and returns the teleportation errors to the verifier. The verifier now either samples an XX or a ZZ term of the Hamiltonian to match Bob's question, and can then recover the outcome of measuring $\ket{\psi}$ with that term by combining Bob's answers with Alice's teleportation errors. In the case that Alice and Bob are honest, repeating this round multiple times suffices to estimate the minimum eigenvalue of $H$. In fact, as long as Bob's operations are (nearly) isometric to a tensor of $Z$ observables when his question asks for standard basis measurements and a tensor of $X$ observables when his question asks for Hadamard basis measurements, we can conclude that the parties cannot on average convince the verifier that $H$ has a significantly lower eigenvalue that it really has. The purpose of the CHSH and Commutation questions is to enforce this structure on Bob's operations. Combined, they yield a variant of the \emph{Pauli braiding test} \cite{10.1145/3055399.3055468,grilo2020simpleprotocolverifiabledelegation}.  

Now, for the purpose of CVQC, we care about the value of any \emph{computationally} nonlocal strategy for this family of games. First, it is imperative that if the strategy passes the CHSH test, then we can conclude that Bob's operations anti-commute. This indeed follows in the computationally nonlocal setting, as discussed in \cref{subsec:CHSH}. It turns out that the only other crucial property is that Bob cannot change his strategy based on whether Alice received a Teleport, CHSH, or Commutation question. This clearly holds for any computationally nonlocal strategy, not just those that arise from the use of QFHE. While these are the main ideas, we now give a precise description of the \cite{NZ23} family of games that can be run through the generalized KLVY compiler to yield CVQC from any blind classical delegation of quantum computation protocol.

\paragraph{The game $H[G]$.} Define $\cG_\secp$ to be the set of games $G[H]$ where $H$ is a $\secp$-qubit Hamiltonian. Let $\beta = \beta(\secp), \alpha = \alpha(\secp) \in [-1,1]$ be functions of the security parameter, and let $\kappa = \Theta((\beta-\alpha)^2)$ be a parameter to be set later. The game $G[H]$ will allow us to decide whether the minimum eigenvalue of $H$ is smaller than $\alpha$ or larger than $\beta$. 

Given $H$, define $D_X$ to be the renormalized distribution over $X(e_i + e_j)$ operators induced by $H$, and define $D_Z$ analogously. Define $D_Q$ to be $U_\secp \times D_X$, where $U_\secp$ is the uniform distribution over $\{0,1\}^\secp$. Then for any $H$, the game $G[H]$ is defined as follows.

\begin{itemize}
    \item $Q_{G[H]}$: Sample the Alice question $q_A$ and Bob question $q_B$ by choosing one of the following options, where the first two are chosen with probability $(1-\kappa)/2$ and the last one with probability $\kappa$.
    \begin{itemize}
        \item \textbf{CHSH.} Sample $(a,b) \gets U_\secp \times D_X$ conditioned on $a \cdot b = 1$ (later, we will refer to this distribution as $D^1_Q$). Sample $x,y \gets \{0,1\}$, and set $q_A = (\text{CHSH},(a,b,x))$ and $q_B = y$.
        \item \textbf{Commutation.} Sample $(a,b) \gets U_\secp \times D_X$ conditioned on $a \cdot b = 0$  (later, we will refer to this distribution as $D^0_Q$). Sample $y \gets \{0,1\}$, and set $q_A = (\text{Commutation},(a,b))$ and $q_B = y$.
        \item \textbf{Teleport.} Sample $y \gets \{0,1\}$ and set $q_A = \text{Teleport}$ and $q_B = y$.
    \end{itemize}
    \item $V_{G[H]}$: Receive Alice answer $s_A$ (of varying size) and Bob answer $s_B \in \{0,1\}^{\secp}$. Compute the following, determined by the type of Alice question.
    \begin{itemize}
        \item \textbf{CHSH.} Let $z \coloneqq (1-y)(a \cdot s_B) + y (b \cdot s_B)$, and output 1 iff $s_A + z = x \cdot y$.
        \item \textbf{Commutation.} Let $z \coloneqq (1-y)(a \cdot s_B) + y (b \cdot s_B)$, and output 1 iff $(s_A)_y = z$.
        \item \textbf{Teleport.} Sample $w$ such that $w=0$ with probability $\sum_{i,j}p_{X,i,j}$ and $w=1$ with probability $\sum_{i,j}p_{Z,i,j}$. If $w \neq q_B$ then output $1$. Otherwise, sample a term $W(e_i+e_j)$ from the distribution induced by $p_{W,i,j}$, where $W = X$ if $w = 0$ and $W = Z$ if $w = 1$. Parse $s_A = (z,x)$, where $z,x \in \{0,1\}^\secp$, and compute the following.
        \begin{itemize}
            \item If $W = Z$, output 1 iff $(-1)^{(s_B)_i+(s_B)_j + (s_A)_i + (s_A)_j} = -1$.
            \item If $W = X$, output 1 iff $(-1)^{(s_B)_i+(s_B)_j + (s_A)_{\secp+ i} + (s_A)_{\secp+j}} = -1$.
        \end{itemize}
    \end{itemize}
\end{itemize}

Finally, we obtain a classical verification protocol for BQP (\cref{def:CVQC}) from blind classical delegation of quantum computation by combining the following theorem with \cref{thm:KLVY}, thus proving \cref{thm:CVQC}.

\begin{theorem}[Adaptation of Theorem 46 from \cite{NZ23}]\label{thm:main-NZ23}
    Let $\cG[\mathsf{YES}]$ be the family of games $\cG$ restricted to those defined by a Hamiltonian $H$ with lowest eigenvalue at most $\alpha$. There exists an efficient nonlocal strategy $\mathscr{S}$ with value \[\omega(G,\mathscr{S}) \geq \frac{1}{2}(1-\kappa)(1+\omega_\CHSH) + \kappa(1-\frac{1}{4}\alpha)\] for all $G \in \cG[\mathsf{YES}]$. Next, let $\cG[\mathsf{NO}]$ be the family of games $\cG$ restricted to those defined by a Hamiltonian $H$ with lowest eigenvalue at least $\beta$. Then there exists a choice of $\kappa = \Theta((\beta-\alpha)^2)$ such that the computationally nonlocal value of $\cG[\mathsf{NO}]$ is upper-bounded by \[\omega_C(G) = \frac{1}{2}(1-\kappa)(1+\omega_\CHSH) + \kappa(1-\frac{1}{4}\alpha) - \frac{1}{8}\kappa(\beta-\alpha).\] 
\end{theorem}

The proof of this theorem follows readily from arguments made in \cite{NZ23}. For completeness, we give an overview of their main claims written in our notation in \cref{sec:NZ23}.

\subsection{Encrypted CNOT and applications}\label{subsec:encrypted-CNOT}

\cite{Mahadev2017ClassicalHE} informally introduced the notion of a (two-round) ``encrypted CNOT'' protocol as a sub-routine in her construction of quantum fully-homomorphic encryption. Here, we define encrypted CNOT formally as a special case of blind classical delegation of quantum computation. In particular, this means that it follows generically from any OSP, due to the results from \cref{subsec:blind-delegation}. 

However, we also provide a simple and direct protocol for encrypted CNOT which in particular implies that \emph{two-round} encrypted CNOT follows from any two-round OSP. Finally, we discuss applications of encrypted CNOT to quantum fully-homormophic encryption and claw-state generators with \emph{indistinguishability} security (\cref{def:CSG}).

\begin{definition}[Encrypted CNOT]\label{def:encrypted-CNOT}
    An encrypted CNOT protocol is the special case of blind classical delegation of quantum computation (\cref{def:blind-computation}) where $\cH_\regV = \cH_\regW$ is a two-qubit register, the client's private input is a bit $b \in \{0,1\}$, and $Q(b,(\regV_0,\regV_1))$ is the identity if $b=0$ and applies a CNOT gate from $\regV_0$ to $\regV_1$ if $b = 1$. We say that the protocol is \emph{two-round} if it just consists of one message from the client followed by one message from the server:
    \begin{itemize}
        \item $\ECNOT.\Gen(1^\secp,b) \to (\msg_C,\state_C)$: The PPT client takes as input the security parameter $1^\secp$ and a bit $b$ and outputs a message $\msg_C$ and state $\state_C$.
        \item $\ECNOT.\Apply((\cV_0,\cV_1),\msg_C) \to ((\cV_0,\cV_1),\msg_S)$: The QPT server takes as input the client's message $\msg_C$, performs an operation on the registers $\cV_0,\cV_1$, and produces a message $\msg_S$.
        \item $\ECNOT.\Dec(\state_C,\msg_S) \to (r,s)$: The PPT decoding algorithm takes as input the client's state $\state_C$ and the server's message $\msg_S$ and outputs one-time pad keys $(r,s)$.
    \end{itemize}
\end{definition}

\begin{theorem}
    OSP (resp. two-round OSP) implies encrypted CNOT (resp. two-round encrypted CNOT).
\end{theorem}

\begin{proof}
    Encrypted CNOT from OSP follows generically from \cref{cor:blind}, so we focus on the two-round case. We first describe the protocol, where $(\OSP.\Sen,\OSP.\Rec,\OSP.\Dec)$ is any two-round OSP.
    \begin{itemize}
        \item $\ECNOT.\Gen(1^\secp,b)$: Sample $\OSP$ first-round messages \[(\OSP.\msg_{S,0},\OSP.\state_{S,0}) \gets \OSP.\Sen(1^\secp,b), \ \ \ (\OSP.\msg_{S,1},\OSP.\state_{S,1}) \gets \OSP.\Sen(1^\secp,1-b),\] and define \[\msg_C \coloneqq (\OSP.\msg_{S,0},\OSP.\msg_{S,1}), \ \ \ \state_C \coloneqq (\OSP.\state_{S,0},\OSP.\state_{S,1}).\]
        \item $\ECNOT.\Apply((\cV_0,\cV_1),\msg_C)$: Generate $\OSP$ states and second-round messages  \[\left(\ket{\psi_0}_{\cO_0},\OSP.\msg_{R,0}\right) \gets \OSP.\Rec(\OSP.\msg_{S,0}), \ \ \ \ \left(\ket{\psi_1}_{\cO_1},\OSP.\msg_{R,1}\right) \gets \OSP.\Rec(\OSP.\msg_{S,1}).\] Apply $\CNOT_{\cV_0 \to \cO_1}, \CNOT_{\cO_0 \to \cO_1}, \CNOT_{\cO_0 \to \cV_1}$, measure $\cO_0$ in the Hadamard basis to obtain bit $m_0$, and measure $\cO_1$ in the standard basis to obtain bit $m_1$. Define \[\msg_S \coloneqq (\OSP.\msg_{R,0},\OSP.\msg_{R,1},m_0,m_1),\] and output $(\cV_0,\cV_1)$.
        \item $\ECNOT.\Dec(\state_C,\msg_S)$: Decode $\OSP$ output bits \[t_0 \gets \OSP.\Dec(\OSP.\state_{S,0},\OSP.\msg_{R,0}), \ \ \ \ t_1 \gets \OSP.\Dec(\OSP.\state_{S,1},\OSP.\msg_{R,1}),\] and compute the one-time pad keys $r,s$ as follows.
        \begin{itemize}
            \item If $b = 0$: $r = (0,t_0)$, $s = (t_1,0)$.
            \item If $b = 1$: $r = (0,m_1 \oplus t_1)$, $s = (m_0 \oplus t_0,0)$.
        \end{itemize}
    \end{itemize}

    Security is immediate from the security of the OSP. Thus, it remains to check correctness, which we check separately for $b=0$ and $b=1$. To avoid clutter, we will assume that registers $\cV_0,\cV_1$ hold pure states $\alpha_0\ket{0} + \alpha_1\ket{1}$ and $\beta_0\ket{0} + \beta_1\ket{1}$, but note that correctness extends readily to the setting where $\cV_0,\cV_1$ may be entangled with each other and with auxiliary systems.

    First, suppose that $b=0$. We will break the server's actions into two stages: (1) apply $\CNOT_{\cV_0 \to \cO_1},\CNOT_{\cO_0 \to \cO_1}$ and measure $\cO_1$, and (2) apply $\CNOT_{\cO_0 \to \cV_1}$ and measure $\cO_0$. By the correctness of $\OSP$, we have that right before the measurement in the first stage, the state of the system on registers $(\cV_0,\cO_0,\cO_1)$ is (negligibly close to)
    \begin{align*}
        &\CNOT_{\cV_0 \to \cO_1}\CNOT_{\cO_0 \to \cO_1}X^{0,t_0,0}Z^{0,0,t_1}\left(\left(\alpha_0 \ket{0} + \alpha_1 \ket{1}\right)_{\cV_0} \otimes \ket{0}_{\cO_0} \otimes \ket{+}_{\cO_1}\right) \\
        &=X^{0,t_0,t_0}Z^{t_1,t_1,t_1}\CNOT_{\cV_0 \to \cO_1}\CNOT_{\cO_0 \to \cO_1}\left(\left(\alpha_0 \ket{0} + \alpha_1 \ket{1}\right)_{\cV_0} \otimes \ket{0}_{\cO_0} \otimes \ket{+}_{\cO_1}\right) \\
        &=X^{0,t_0,t_0}Z^{t_1,t_1,t_1}\left(\left(\alpha_0 \ket{0} + \alpha_1 \ket{1}\right)_{\cV_0} \otimes \ket{0}_{\cO_0} \otimes \ket{+}_{\cO_1}\right).
    \end{align*}

    Thus, measuring $\cO_1$ in the standard basis has no affect on the remaining system. Next, we write the state of the system on registers $(\cV_0,\cV_1,\cO_0)$ right before the measurement in the second stage, which we imagine implementing by applying a Hadamard gate and then measuring in the standard basis.
    \begin{align*}
        &H_{\cO_0}\CNOT_{\cO_0 \to \cV_1}X^{0,0,t_0}Z^{t_1,0,t_1}\left(\left(\alpha_0\ket{0} + \alpha_1\ket{1}\right)_{\cV_0} \otimes \left(\beta_0\ket{0} + \beta_1\ket{1}\right)_{\cV_1} \otimes \ket{0}_{\cO_0}\right) \\ 
        &= X^{0,t_0,t_1}Z^{t_1,0,t_0}H_{\cO_0}\CNOT_{\cO_0 \to \cV_1}\left(\left(\alpha_0\ket{0} + \alpha_1\ket{1}\right)_{\cV_0} \otimes \left(\beta_0\ket{0} + \beta_1\ket{1}\right)_{\cV_1} \otimes \ket{0}_{\cO_0}\right) \\
        &= X^{0,t_0,t_1}Z^{t_1,0,t_0}\left(\left(\alpha_0\ket{0} + \alpha_1\ket{1}\right)_{\cV_0} \otimes \left(\beta_0\ket{0} + \beta_1\ket{1}\right)_{\cV_1} \otimes \ket{+}_{\cO_0}\right)
    \end{align*}

    Thus, measuring $\cO_0$ in the standard basis has no affect on the system $(\cV_0,\cV_1)$, which is \[X^{0,t_0}Z^{t_1,0}\left(\left(\alpha_0\ket{0} + \alpha_1\ket{1}\right)_{\cV_0} \otimes \left(\beta_0\ket{0} + \beta_1\ket{1}\right)_{\cV_1}\right),\] as desired.

    Next, suppose $b=1$. We break the server's actions into two stages in the same manner. By the correctness of $\OSP$, we have that right before the measurement in the first stage, the state of the system on registers $(\cV_0,\cO_0,\cO_1)$ is (negligibly close to)
    \begin{align*}
        &\CNOT_{\cV_0 \to \cO_1}\CNOT_{\cO_0 \to \cO_1}X^{0,0,t_1}Z^{0,t_0,0}\left(\left(\alpha_0 \ket{0} + \alpha_1 \ket{1}\right)_{\cV_0} \otimes \ket{+}_{\cO_0} \otimes \ket{0}_{\cO_1}\right) \\
        &=X^{0,0,t_1}Z^{0,t_0,0}\CNOT_{\cV_0 \to \cO_1}\CNOT_{\cO_0 \to \cO_1}\left(\left(\alpha_0 \ket{0} + \alpha_1 \ket{1}\right)_{\cV_0} \otimes \ket{+}_{\cO_0} \otimes \ket{0}_{\cO_1}\right) \\
        &=X^{0,0,t_1}Z^{0,t_0,0}\left(\frac{1}{\sqrt{2}}\alpha_0\ket{000} + \frac{1}{\sqrt{2}}\alpha_0\ket{011} + \frac{1}{\sqrt{2}}\alpha_1\ket{101} + \frac{1}{\sqrt{2}}\alpha_1\ket{110}\right) \\
        &=X^{0,0,t_1}Z^{0,t_0,0}\left(\frac{1}{\sqrt{2}}\left(\alpha_0\ket{00} + \alpha_1\ket{11}\right)\ket{0} + \frac{1}{\sqrt{2}}\left(\alpha_0\ket{01} + \alpha_1\ket{10}\right)\ket{1}\right) \\
        &=\sum_{c \in \{0,1\}} \frac{1}{\sqrt{2}}X^{0, c,t_1}Z^{0,t_0,0}\bigg(\left(\alpha_0\ket{00} + \alpha_1\ket{11}\right) \otimes \ket{c}\bigg) \\
        &= \sum_{m_1 \in \{0,1\}}\frac{1}{\sqrt{2}}\bigg(X^{0,m_1 \oplus t_1}Z^{0,t_0}\left(\alpha_0\ket{00} + \alpha_1\ket{11}\right)\bigg) \otimes \ket{m_1},
    \end{align*}
    where the last equality follows by a change of variables $m_1 = c \oplus t_1$. Now, the server measures the last register $(\cO_1)$ in the standard basis to obtain $m_1$, and then applies the second stage. We write the state of the system on registers $(\cV_0,\cV_1,\cO_0)$ right before the measurement in the second stage, which again we implement by applying a Hadamard gate and then measuring in the standard basis. Below, the one-time pad keys are always written in the order $(\cV_0,\cV_1,\cO_0)$.
    \begingroup
\allowdisplaybreaks
    \begin{align*}
        &H_{\cO_0}\CNOT_{\cO_0 \to \cV_1}X^{0,0,m_1 \oplus t_1}Z^{0,0,t_0}\left(\left(\alpha_0\ket{00} + \alpha_1\ket{11}\right)_{\cV_0,\cO_0} \otimes \left(\beta_0\ket{0} + \beta_1\ket{1}\right)_{\cV_1}\right) \\
        &= X^{0,m_1 \oplus t_1,t_0}Z^{0,0,m_1 \oplus t_1}H_{\cO_0}\CNOT_{\cO_0 \to \cV_1}\left(\left(\alpha_0\ket{00} + \alpha_1\ket{11}\right)_{\cV_0,\cO_0} \otimes \left(\beta_0\ket{0} + \beta_1\ket{1}\right)_{\cV_1}\right) \\
        &= X^{0,m_1 \oplus t_1,t_0}Z^{0,0,m_1 \oplus t_1}H_{\cO_0}\left(\alpha_0\beta_0\ket{000} + \alpha_0\beta_1\ket{010} + \alpha_1\beta_0\ket{111} + \alpha_1\beta_1\ket{101}\right)_{\cV_0,\cV_1,\cO_0} \\
        &= X^{0,m_1 \oplus t_1,t_0}Z^{0,0,m_1 \oplus t_1}\frac{1}{\sqrt{2}}\bigg(\big(\alpha_0\beta_0\ket{00} + \alpha_0\beta_1\ket{01} + \alpha_1\beta_0\ket{11} + \alpha_1\beta_1\ket{10}\big)\ket{0} \\ & ~~~~~~~~~~~~~~~~~~~~~~~~~~~~~~~~~~~~~~~~~ + \big(\alpha_0\beta_0\ket{00} + \alpha_0\beta_1\ket{01} - \alpha_1\beta_0\ket{11} - \alpha_1\beta_1\ket{10}\big) \ket{1}\bigg) \\
        &= \sum_{c \in \{0,1\}} \frac{1}{\sqrt{2}}X^{0,m_1 \oplus t_1,t_0}Z^{c,0,m_1 \oplus t_1}\Big(\alpha_0\beta_0\ket{00} + \alpha_0\beta_1\ket{01} + \alpha_1\beta_0\ket{11} + \alpha_1\beta_1\ket{10}\Big) \otimes \ket{c} \\
        &= \sum_{m_0 \in \{0,1\}} \frac{1}{\sqrt{2}}X^{0,m_1 \oplus t_1}Z^{m_0 \oplus t_0,0}\Big(\alpha_0\beta_0\ket{00} + \alpha_0\beta_1\ket{01} + \alpha_1\beta_0\ket{11} + \alpha_1\beta_1\ket{10}\Big)_{\cV_0,\cV_1}\otimes Z^{m_1 \oplus t_1}\ket{m_0}_{\cO_0},
    \end{align*}
    \endgroup
    where the equality follows by a change of variables $m_0 = c \oplus t_0$. Thus, measuring $\cO_0$ in the standard basis to obtain $m_0$ yields the final state \[X^{0,m_1 \oplus t_1}Z^{m_0 \oplus t_0,0}\left(\alpha_0\beta_0\ket{00} + \alpha_0\beta_1\ket{01} + \alpha_1\beta_0\ket{11} + \alpha_1\beta_1\ket{10}\right)_{\cV_0,\cV_1},\] as desired.
   
\end{proof}

\paragraph{Application: QFHE.} \cite{Mahadev2017ClassicalHE} combined a particular classical FHE protocol with a particular two-round encrypted CNOT protocol in a non-black-box way in order to achieve quantum FHE. However, \cite{10.1007/978-3-031-68382-4_8} recently pioneered a \emph{generic} approach to constructing quantum FHE from \emph{any} classical FHE (with log-depth decryption) and \emph{any} dual-mode TCF. Implicit in their work is that, in fact, one can use any two-round encrypted CNOT protocol, and thus (due to our work) any two-round OSP. We confirm this in the following theorem.

\begin{theorem}[Adapted from \cite{10.1007/978-3-031-68382-4_8}]
    Given any FHE with decryption in $\text{NC}_1$ and any two-round encrypted CNOT, there exists QFHE.
\end{theorem}

\begin{corollary}
    Given any FHE with decryption in $\text{NC}_1$ and any two-round OSP, there exists QFHE.
\end{corollary}

\begin{proof}
    In \cite[Section 5]{10.1007/978-3-031-68382-4_8}, it is shown that classical FHE (with decryption in $\text{NC}_1$) plus a procedure for generating a ``hidden Bell pair'' is sufficient to construct quantum FHE. We now describe the requirements for the hidden Bell pair procedure, where $\Enc$ refers to a classical FHE encryption.

    \begin{itemize}
        \item Before the procedure begins, the client holds a bit $\mu$ and the server holds a state 
        \[X^rZ^s\left(\frac{1}{\sqrt{2}}\left(\ket{0}_{\cW_0}\ket{\phi_0^0} + \ket{1}_{\cW_0}\ket{\phi_0^1}\right) \otimes \frac{1}{\sqrt{2}}\left(\ket{0}_{\cW_1}\ket{\phi_1^0} + \ket{1}_{\cW_1}\ket{\phi_1^1}\right)\right),\] where the one-time pad $(r,s)$ is on registers $\cW_0,\cW_1$, along with encryptions $\Enc(s,r)$ of the one-time pad keys.
        \item The client runs a parameter generation algorithm $(\pparam,\sparam) \gets \Gen(1^\secp,\mu)$ and publishes $\pparam, \Enc(\sparam)$. We requires that $\pparam$ is a semantically secure encryption of $\mu$.
        \item The server runs a procedure to obtain a state that is negligibly close in trace distance to \begin{align*}X^{r'}Z^{s'}\left(\frac{1}{\sqrt{2}}\left(\ket{00}_{\cW_\mu,\cW_2}\ket{\phi_\mu^0} + \ket{11}_{\cW_\mu,\cW_2}\ket{\phi_\mu^1}\right) \otimes \frac{1}{\sqrt{2}}\left(\ket{0}_{\cW_{1-\mu}}\ket{\phi_{1-\mu}^0} + \ket{1}_{\cW_{1-\mu}}\ket{\phi_{1-\mu}^1}\right)\right),\end{align*} where the one-time pad $(s',r')$ is applied to registers $\cW_0,\cW_1,\cW_2$. The server also obtains encryptions $\Enc(r',s')$ of the one-time pad keys.
    \end{itemize}

    It is straightforward to implement such a procedure given a two-round encrypted CNOT:

    \begin{itemize}
        \item The client's parameter generation algorithm runs $(\msg_{C,0},\state_{C,0}) \gets \ECNOT.\Gen(1^\secp,\mu)$ and $(\msg_{C,1},\state_{C,1}) \gets \ECNOT.\Gen(1^\secp,1-\mu)$, and sets $\pparam \coloneqq (\msg_{C,0},\msg_{C,1})$ and $\sparam \coloneqq (\state_{C,0},\state_{C,1})$.
        \item The server initializes register $\cW_2$ to $\ket{+}_{\cW_2}$, applies \[((\cW_0,\cW_2),\msg_{S,0}) \gets \ECNOT.\Apply((\cW_0,\cW_2),\msg_{C,0}),\] applies \[((\cW_1,\cW_2),\msg_{S,1}) \gets \ECNOT.\Apply((\cW_1,\cW_2),\msg_{C,1}),\] and finally uses $\Enc(r,s),\Enc(\sparam) = \Enc(\state_{C,0},\state_{C,1})$ and $\msg_{S,0},\msg_{S,1}$ to homomorphically obtain $\Enc(r',s')$ under the FHE.
    \end{itemize}

    Correctness is straightforward, as $\mu$ determines whether the CNOT is applied from $\cW_0$ to $\cW_2$ or from $\cW_1$ to $\cW_2$. Security follows immediately from the security of the encrypted CNOT.

\end{proof}

\paragraph{Application: CSG.} Finally, we return to the notion of a claw-state generator ($\CSG$), as defined in \cref{subsec:CSG}. There, it was shown that any $\CSG$ with search security implies OSP. Here, we show (a strengthening of) the reverse implication: Any $\OSP$ implies $\CSG$ with \emph{indistinguishability} security. This implication goes via the intermediate primitive of encrypted CNOT.

\begin{theorem}
    Encrypted CNOT (resp. two-round encrypted CNOT) implies differentiated-bit $\CSG$ (resp. two-round $\CSG$) with indistinguishability security, for any $n = \poly(\secp)$.
\end{theorem}

\begin{corollary}
    OSP (resp. two-round OSP) implies differentiated-bit $\CSG$ (resp. two-round $\CSG$) with indistinguishability security, for any $n = \poly(\secp)$.
\end{corollary}

\begin{proof}
    We give the protocol in the two-round case, which can also be instantiated with any (not necessarily two-round) encrypted CNOT to yield a (not necessarily two-round) CSG. Let \[(\ECNOT.\Gen,\ECNOT.\Apply,\ECNOT.\Dec)\] be any two-round encrypted CNOT.

    \begin{itemize}
        \item $\CSG.\Sen(1^\secp,n)$: Sample a uniformly random string $\Delta \gets \{0,1\}^n$, and for each $i \in [n]$, sample \[(\ECNOT.\msg_{C,i},\ECNOT.\state_{C,i}) \gets \ECNOT(1^\secp,\Delta_i).\] Define \[\msg_S \coloneqq (\ECNOT.\msg_{C,1},\dots,\ECNOT.\msg_{C,n}), \ \ \ \ \ \state_S \coloneqq (\ECNOT.\state_{C,1},\dots,\ECNOT.\state_{C,n}).\]
        \item $\CSG.\Rec(\msg_S)$: Initialize a register $\cB$ to $\ket{+}_\cB$. For each $i \in [n]$, initialize a register $\cC_i$ to $\ket{0}_{\cC_i}$ and apply \[((\cB,\cC_i),\ECNOT.\msg_{S,i}) \gets \ECNOT.\Apply((\cB,\cC_i),\ECNOT.\msg_{C,i}).\] Output the state on registers $\cB,\cC_1,\dots,\cC_n$, and define \[\msg_R \coloneqq (\ECNOT.\msg_{S,1},\dots,\ECNOT.\msg_{S,n}).\]
        \item $\CSG.\Dec(\state_S,\msg_R)$: For each $i \in [n]$, compute 
        \[((r_{i,0},r_{i,1}),(s_{i,0},s_{i,1})) \gets \ECNOT.\Dec(\ECNOT.\state_{C,i},\ECNOT.\msg_{S,i}),\] define \[r_0 \coloneqq \bigoplus_{i \in [n]} r_{i,0} \in \{0,1\}, \ \  r \coloneqq (r_{1,1},\dots,r_{n,1}) \in \{0,1\}^n, \ \  s \coloneqq \left(\bigoplus_{i \in [n]}s_{i,0},s_{1,1},\dots,s_{n,1}\right) \in \{0,1\}^{n+1}, \] and output \[x_0 = r \oplus r_0 \cdot \Delta, \ \ x_1 = r \oplus (1-r_0) \cdot \Delta, \ \ z = s \cdot (1,\Delta).\]
    \end{itemize}

    Observe that, by the correctness of the encrypted CNOT protocol, the state on registers $(\cB,\cC_1,\dots,\cC_n)$ output by $\CSG.\Rec$ is (negligibly close to)
    \begin{align*} 
        &X^{\bigoplus_{i \in [n]}r_{i,0},r_{1,1},\dots,r_{n,1}}Z^{\bigoplus_{i \in [n]}s_{i,0},s_{1,1},\dots,s_{n,1}}\left(\frac{1}{\sqrt{2}}\ket{0,0^n} + \frac{1}{\sqrt{2}}\ket{1,\Delta}\right) \\
        &= X^{r_0,r}Z^s\left(\frac{1}{\sqrt{2}}\ket{0,0^n} + \frac{1}{\sqrt{2}}\ket{1,\Delta}\right) \\
        &= \frac{1}{\sqrt{2}}\ket{0,r \oplus r_0 \cdot \Delta} + (-1)^{s \cdot (1,\Delta)}\frac{1}{\sqrt{2}}\ket{1,r \oplus (1-r_0) \cdot \Delta},
    \end{align*}

    and thus correctness holds. Indistinguishability security follows immediately from the security of the encrypted CNOT, since for each $i \in [n]$, $x_{0,i} \oplus x_{1,i} = \Delta_i$.

\end{proof}

\section{Implications}

In this section, we show that OSP implies both commitments and oblivious transfer with one classical party, while \emph{two-round} OSP implies public-key encryption with classical public keys and ciphertexts. The purpose of exploring this direction is to place lower bounds on the cryptography necessary to build OSP.

\subsection{Commitments from OSP}\label{subsec:com}

\begin{definition}[Commitment]
    A (statistically-binding, computationally-hiding) commitment between a classical committer and quantum receiver consists of an interaction \[\left(\state_\Com,\ket{\psi}_\Rec\right) \gets \langle \Com(1^\secp,b),\Rec(1^\secp)\rangle,\] where $\Com$ is a PPT machine with input bit $b \in \{0,1\}$ and $\Rec$ is a QPT machine, along with algorithms $(\Open,\Ver)$ with the following syntax.
    \begin{itemize}
        \item $\Open(\state_\Com) \to (b,s)$ is a PPT algorithm that takes as input the committer's state $\state_\Com$ and produces a bit $b$ and opening information $s$.
        \item $\Ver(\ket{\psi}_\Rec,b,s) \to \{\top, \bot\}$ is a QPT algorithm that takes as input the receiver's state $\ket{\psi}_\Rec$, a bit $b$, and opening information $s$, and either accepts or rejects.
    \end{itemize}
    \emph{Correctness} requires that for any $b \in \{0,1\}$,
    \[\Pr\left[\Ver(\ket{\psi}_\Rec,b,s) = \top : \begin{array}{r}(\state_\Com,\ket{\psi}_\Rec) \gets \langle \Com(1^\secp,b),\Rec(1^\secp)\rangle \\ (b,s) \gets \Open(\state_\Com)\end{array}\right] = 1-\negl(\secp).\]
    The commitment satisfies \emph{computational hiding} if for any QPT adversarial receiver $\{\Adv_\secp\}_{\secp \in \bbN}$,
    \begin{align*}
    &\bigg|\Pr\left[b_\Adv = 0 : (\state_\Com,b_\Adv) \gets \langle \Com(1^\secp,0),\Adv_\secp\rangle\right] \\ &- \Pr\left[b_\Adv = 0 : (\state_\Com,b_\Adv) \gets \langle \Com(1^\secp,1),\Adv_\secp\rangle\right]\bigg| = \negl(\secp).
    \end{align*}
    The commitment satisfies \emph{statistical binding} if for any unbounded adversarial committer $\{\Adv_\secp\}_{\secp \in \bbN}$,
    \begin{align*}
        \Pr\left[\Ver(\ket{\psi}_\Rec,b,s) = \top : \begin{array}{r} (\state_\Adv,\ket{\psi}_\Rec) \gets \langle \Adv_\secp,\Rec(1^\secp)\rangle \\ b \gets \{0,1\} \\ s \gets \Adv_\secp(\state_\Adv,b)\end{array}\right] \leq \frac{1}{2} + \negl(\secp).
    \end{align*}
\end{definition}

\begin{theorem}
    OSP implies a commitment between a classical committer and quantum receiver.
\end{theorem}

\begin{proof}

    Given any OSP, consider the following commitment scheme.

    \begin{itemize}
        \item For all $i \in [\secp]$, execute an OSP between the committer playing the role of the sender with input $b$ and the receiver playing the role of the OSP receiver. Each execution results in an output $s_i$ for the sender a state $\ket{\psi_i}$ for the receiver. Define $\state_\Com \coloneqq (b,s_1,\dots,s_\secp)$ and $\ket{\psi}_\Rec \coloneqq \ket{\psi_1} \otimes \dots \otimes \ket{\psi_\secp}$.
        \item $\Open(\state_\Com)$ outputs $b, s = (s_1,\dots,s_\secp)$.
        \item $\Ver(\ket{\psi}_\Rec,b,s)$ measures $\ket{\psi}_\Rec$ in the standard basis if $b=0$ or the Hadamard basis if $b=1$, and accepts if the outcome is $s$.
    \end{itemize}

    Correctness is straightforward from the correctness of OSP, and hiding follows from the security of OSP via a standard hybrid argument. 

    To show binding, consider any state $\ket{\psi}_\Rec$ that the receiver could have after interacting with $\Adv_\secp$ in the commit phase. Note that for any $s_0,s_1$,
    \[\Pr[\Ver(\ket{\psi}_\Rec,0,s_0) = \top] = \| \bra{s_0}\ket{\psi}_\Rec \|^2 \ \ \text{and} \ \ \Pr[\Ver(\ket{\psi}_\Rec,1,s_1) = \top] = \| \bra{s_1}H^{\otimes \secp}\ket{\psi}_\Rec \|^2.\]
    Let 
    \[s_0 \coloneqq \max_s\left\{\Pr[\Ver(\ket{\psi}_\Rec,0,s) = \top]\right\} \ \ \text{and } \ \ s_1 \coloneqq \max_s\left\{\Pr[\Ver(\ket{\psi}_\Rec,1,s) = \top]\right\},\] meaning that in the binding game, the adversary's optimal strategy given $b$ is to output $s_b$. 

    Then since \[\|\bra{s_0}H^{\otimes \secp}\ket{s_1}\|^2 = \frac{1}{2^\secp} = \negl(\secp),\] we conclude that \[\Pr[\Ver(\ket{\psi}_\Rec,0,s_0) = \top] + \Pr[\Ver(\ket{\psi}_\Rec,1,s_1) = \top] \leq 1+\negl(\secp), \] which completes the proof.

\end{proof}

\begin{remark}
    Even though the above commitment satisfies a standard notion of statistical binding (sum binding), it is unclear if the effective committed bit can actually be extracted (even inefficiently) from an adversarial committer. Note that this is always possible in both the fully-classical and fully-quantum setting. Intuitively, the difficulty here arises from the fact that the span of receiver states that can be (perfectly) opened to 0 and the span of receiver states that be (perfectly) opened to 1 are equivalent: the span is the entire Hilbert space. Thus, we cannot define an inefficient projective measurement that successfully distinguishes these spaces. We mention this both out of curiosity, and also because in the next section we will actually use an inefficiently-extractable commitment scheme. Since we are not able to show that our OSP-based commitment is inefficiently-extractable, we rely on a one-way-function based commitment for this purpose.
\end{remark}

\subsection{OT from OSP}\label{subsec:OT-from-OSP}

First, we define a notion of game-based oblivious transfer (OT) with one-sided statistical security, where one party (the receiver) is classical and the other party (the sender) is quantum. For security, we allow the receiver to be an unbounded algorithm, but require the adversarial sender to be QPT. 

\begin{definition}[OT with one-sided statistical security between classical receiver and quantum sender]\label{def:SSP-OT}
    Oblivious transfer (OT) with one-sided statistical security is a protocol that takes place between a PPT receiver $R$ with input $b \in \{0,1\}$ and a QPT sender $S$:

    \[(r,(r_0,r_1),\tau) \gets \langle R(1^\secp,b),S(1^\secp)\rangle,\]

    where $r_0,r_1 \in \{0,1\}^n$ is the sender's output, $r \in \{0,1\}^n$ is the receiver's output, and $\tau$ is the (classical) transcript of communication produced during the protocol. We require the following properties.

    \begin{itemize}
        \item \textbf{Correctness.} For any $b \in \{0,1\}$, 
        \[\Pr\left[r = r_b : (r,(r_0,r_1),\tau) \gets \langle R(1^\secp,b),S(1^\secp)\rangle\right] = 1-\negl(\secp).\]
        We say the protocol satisfies \emph{perfect} correctness if the above probability is equal to 1.
        \item \textbf{Security against a QPT sender.} For any QPT adversarial sender $\{\Adv_\secp\}_{\secp \in \bbN}$,
        \begin{align*}&\Big|\Pr\left[b_\Adv = 0 : (r, b_\Adv, \tau) \gets \langle R(1^\secp,0),\Adv_\secp\rangle\right]\\ &- \Pr\left[b_\Adv = 0 : (r, b_\Adv, \tau) \gets \langle R(1^\secp,1),\Adv_\secp\rangle\right]\Big| = \negl(\secp).\end{align*}
    \end{itemize}

    We consider two different security properties that may hold against an unbounded receiver. 

    \begin{itemize}

        \item \textbf{Search security against an unbounded receiver.} For any unbounded adversarial receiver $\Adv$,
        \[\Pr[r_\Adv = (r_0,r_1) : (r_\Adv,(r_0,r_1),\tau) \gets \langle \Adv,S(1^\secp)\rangle] = \negl(\secp).\]

        \item \textbf{Indistinguishability security against an unbounded receiver.} For this definition, we restrict our attention to the case where $n=1$ (i.e.\ $r_0,r_1$ are single bits). There exists an unbounded extractor $\Ext$ such that for any unbounded adversarial receiver $\Adv$,
        \[\Bigg|\Pr[r_{1-b,\Adv} = r_{1-b} : \begin{array}{r}((r_{0,\Adv},r_{1,\Adv}),(r_0,r_1),\tau) \gets \langle \Adv,S(1^\secp)\rangle \\ b \gets \Ext(\tau)\end{array}] - \frac{1}{2}\Bigg| = \negl(\secp).\]
    \end{itemize}
    
\end{definition}

We first provide a construction of OT with search security against an unbounded receiver, assuming only a differentiated-bit CSG with indistinguishability security (defined in \cref{def:CSG} and constructed from OSP in \cref{subsec:encrypted-CNOT}). 

Since search security is a somewhat non-standard notion of OT security, we also show how to tweak the construction to obtain indistinguishability security by additionally using an \emph{inefficiently-extractable} commitment (defined in \cref{sec:classical-com}). We use a fully-classical (post-quantum) inefficiently-extractable commitment, which is known from any (post-quantum) one-way function \cite{10.1007/BF00196774}. Our protocols are given in \cref{fig:OT-from-OSP} and \cref{fig:OT-from-OSP-OWF}.

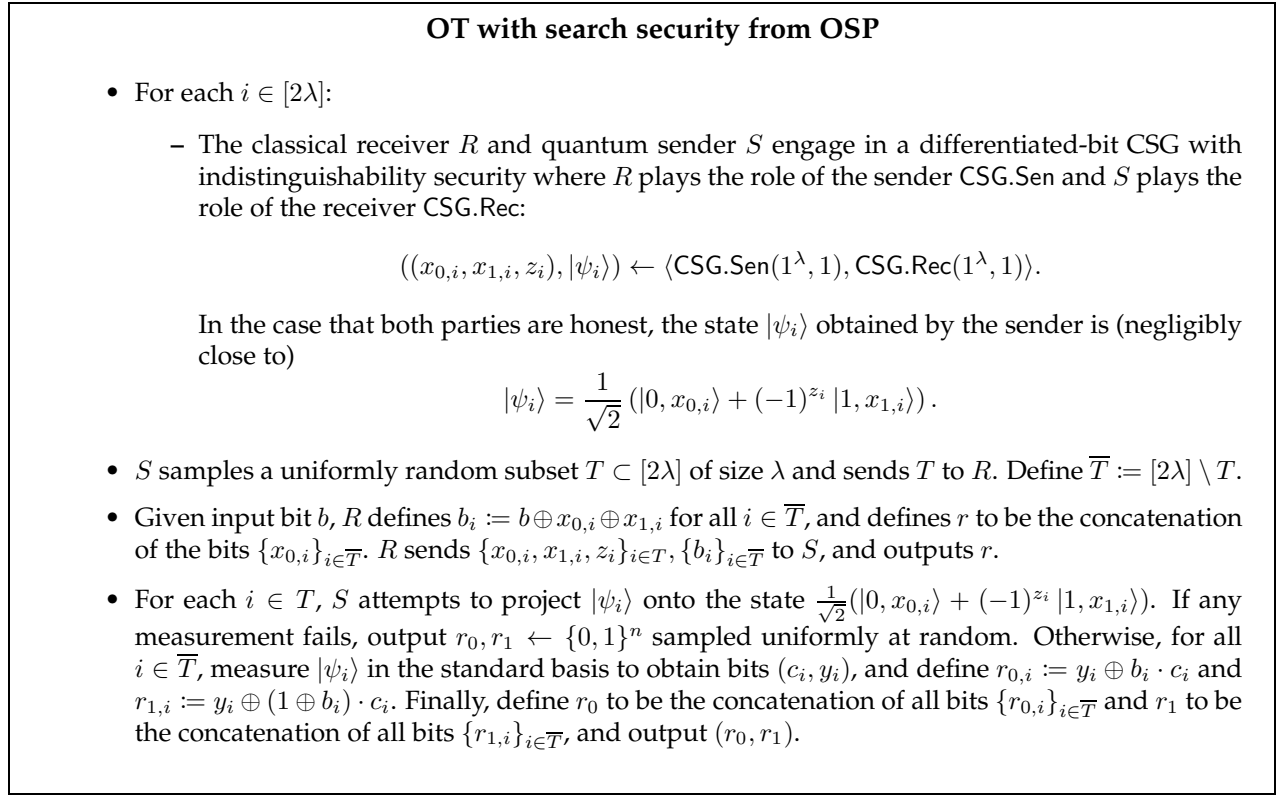
\protocol{OT with search security from OSP}{A protocol for OT with search security against an unbounded receiver, assuming only claw-state generators with indistinguishability security (which follow from OSP).}{fig:OT-from-OSP}{

\begin{itemize}

    \item For each $i \in [2\secp]$:
    \begin{itemize}
        \item The classical receiver $R$ and quantum sender $S$ engage in a differentiated-bit CSG with indistinguishability security where $R$ plays the role of the sender $\CSG.\Sen$ and $S$ plays the role of the receiver $\CSG.\Rec$:
        \[((x_{0,i},x_{1,i},z_i),\ket{\psi_i}) \gets \langle \CSG.\Sen(1^\secp,1),\CSG.\Rec(1^\secp,1)\rangle.\]
        In the case that both parties are honest, the state $\ket{\psi_i}$ obtained by the sender is (negligibly close to)
        \[\ket{\psi_i} = \frac{1}{\sqrt{2}}\left(\ket{0,x_{0,i}} + (-1)^{z_i}\ket{1,x_{1,i}}\right).\]
    
    \end{itemize}

    \item $S$ samples a uniformly random subset $T \subset [2\secp]$ of size $\secp$ and sends $T$ to $R$. Define $\overline{T} \coloneqq [2\secp] \setminus T$.
    \item Given input bit $b$, $R$ defines $b_i \coloneqq b \oplus x_{0,i} \oplus x_{1,i}$ for all $i \in \overline{T}$, and defines $r$ to be the concatenation of the bits $\{x_{0,i}\}_{i \in \overline{T}}$. $R$ sends $\{x_{0,i},x_{1,i},z_i\}_{i \in T}, \{b_i\}_{i \in \overline{T}}$ to $S$, and outputs $r$.
    \item For each $i \in T$, $S$ attempts to project $\ket{\psi_i}$ onto the state $\frac{1}{\sqrt{2}}(\ket{0,x_{0,i}} + (-1)^{z_i}\ket{1,x_{1,i}})$. If any measurement fails, output $r_0,r_1 \gets \{0,1\}^n$ sampled uniformly at random. Otherwise, for all $i \in \overline{T}$, measure $\ket{\psi_i}$ in the standard basis to obtain bits $(c_i,y_i)$, and define $r_{0,i} \coloneqq y_i \oplus b_i \cdot c_i$ and $r_{1,i} \coloneqq y_i \oplus (1 \oplus b_i) \cdot c_i$. Finally, define $r_0$ to be the concatenation of all bits $\{r_{0,i}\}_{i \in \overline{T}}$ and $r_1$ to be the concatenation of all bits $\{r_{1,i}\}_{i \in \overline{T}}$, and output $(r_0,r_1)$.
    
    \end{itemize}
}

\protocol{OT with indistinguishability security from OSP plus OWF}{A protocol for OT with indistinguishability security against an unbounded receiver, assuming CSG plus one-way functions. The differences from the OT protocol in \cref{fig:OT-from-OSP} are highlighted in red.}{fig:OT-from-OSP-OWF}{
\begin{itemize}

    \item For each $i \in [2\secp]$:
    \begin{itemize}
        \item The classical receiver $R$ and quantum sender $S$ engage in a differentiated-bit CSG with indistinguishability security where $R$ plays the role of the sender $\CSG.\Sen$ and $S$ plays the role of the receiver $\CSG.\Rec$:

        \[((x_{0,i},x_{1,i},z_i),\ket{\psi_i}) \gets \langle \CSG.\Sen(1^\secp,1),\CSG.\Rec(1^\secp,1)\rangle.\]
        In the case that both parties are honest, the state $\ket{\psi_i}$ obtained by the sender is (negligibly close to)
        \[\ket{\psi_i} = \frac{1}{\sqrt{2}}\left(\ket{0,x_{0,i}} + (-1)^{z_i}\ket{1,x_{1,i}}\right).\]
        \item \textcolor{red}{$R$ and $S$ engage in an inefficiently-extractable commitment, with $R$ playing the role of $\Com$ and $S$ playing the role of $\Rec$:
        \[(\state_{\Com,i},\state_{\Rec,i},\tau_i) \gets \langle \Com(1^\secp, (x_{0,i},x_{1,i},z_i)),\Rec(1^\secp)\rangle.\]}
    \end{itemize}
    
    \item $S$ samples a uniformly random subset $T \subset [2\secp]$ of size $\secp$ and sends $T$ to $R$. Define $\overline{T} \coloneqq [2\secp] \setminus T$.
    \item \textcolor{red}{For each $i \in T$, $R$ computes the opening to its commitment $((x_{0,i},x_{1,i},z_i),w_i) \gets \Open(\state_{\Com,i})$.} Then, given input bit $b$, $R$ defines $b_i \coloneqq b \oplus x_{0,i} \oplus x_{1,i}$ for all $i \in \overline{T}$, and defines \textcolor{red}{$r \coloneqq \bigoplus_{i \in \overline{T}} x_{0,i}$.}
     Finally, $R$ sends $\{x_{0,i},x_{1,i},z_i,\textcolor{red}{w_i}\}_{i \in T}, \{b_i\}_{i \in \overline{T}}$ to $S$, and outputs $r$.
    \item For each $i \in T$, \textcolor{red}{$S$ checks that $\Ver(\state_\Rec,(x_{0,i},x_{1,i},z_i),w_i) = \top$,} and attempts to project $\ket{\psi_i}$ onto the state $\frac{1}{\sqrt{2}}(\ket{0,x_{0,i}} + (-1)^{z_i}\ket{1,x_{1,i}})$. If any \textcolor{red}{verification check} or measurements fails, output $r_0,r_1 \gets \{0,1\}$ sampled uniformly at random. Otherwise, for all $i \in \overline{T}$, measure $\ket{\psi_i}$ in the standard basis to obtain bits $(c_i,y_i)$, and define $r_{0,i} \coloneqq y_i \oplus b_i \cdot c_i$ and $r_{1,i} \coloneqq y_i \oplus (1 \oplus b_i) \cdot c_i$. Finally, define \textcolor{red}{\[r_0 \coloneqq \bigoplus_{i \in \overline{T}} r_{0,i}, \ \ \ r_1 \coloneqq \bigoplus_{i \in \overline{T}} r_{1,i},\]} and output $(r_0,r_1)$. \end{itemize}

}

\begin{theorem}
    The protocol in \cref{fig:OT-from-OSP}  (resp. \cref{fig:OT-from-OSP-OWF}) satisfies OT with search (resp. indistinguishability) security against an unbounded receiver (\cref{def:SSP-OT}). That is, OSP implies OT with search security, while OT plus one-way functions implies OT with indistinguishability security.
\end{theorem}

\begin{proof}
    We write out the proofs for the protocol in \cref{fig:OT-from-OSP-OWF}, and note that essentially the same proof strategy works to show security of the protocol in \cref{fig:OT-from-OSP}. In what follows, we prove correctness, security against a QPT sender, and indistinguishability security against an unbounded receiver.

    \paragraph{Correctness.} We need to show that for any $b \in \{0,1\}$, the receiver's output $r$ is equal to the sender's output $r_b$ with overwhelming probability. To see this, it suffices to show that for all $i \in \overline{T}$, it holds that (with overwhelming probability) $x_{0,i} = y_i \oplus (b \oplus b_i) \cdot c_i$, where $b_i = b \oplus x_{0,i} \oplus x_{1,i}$, and $(c_i,y_i)$ are obtained by measuring a state that is (negligibly close to)

    \[\frac{1}{\sqrt{2}}\left(\ket{0,x_{0,i}} + \ket{1,x_{1,i}}\right)\] 

    in the standard basis. This is easy to check by plugging in the two choices of $(c_i,y_i) = (0,x_{0,i})$ and $(c_i,y_i) = (1,x_{1,i})$.

    \paragraph{Security against a QPT sender.} This follows directly from the indistinguishability security of the CSG. Indeed, the only information that the sender obtains about $b$ are the bits $b_i = b \oplus x_{0,i} \oplus x_{1,i}$, and the security of CSG implies that each bit $x_{0,i} \oplus x_{1,i}$ is computationally unpredictable to the QPT sender.

    \paragraph{Security against an unbounded receiver.} First, we define the unbounded-time extractor $\Ext$. Given the classical transcript $\tau$ of the OT protocol, $\Ext$ identifies the transcripts $\tau_{\Com,1},\dots,\tau_{\Com,2\secp}$ of the commitment protocols, along with the set $T$ and bits $\{b_i\}_{i \in \overline{T}}$. It runs the unbounded-time extractor for the commitment scheme on each $\{\tau_{\Com,i}\}_{i \in \overline{T}}$ to obtain $\{(x_{0,i},x_{1,i},z_i)\}_{i \in \overline{T}}$, and outputs $b = \maj\{b_i \oplus x_{0,i} \oplus x_{1,i}\}_{i \in \overline{T}}$.

    For any $\{0,1\}$-valued random variable $s$ sampled during the protocol, we define 

    \[\Bias(s) \coloneqq \Big| \Pr[s = 0 \ |\  \state_R] - \frac{1}{2}\Big|,\]

    where $\state_R$ is the final state of the receiver at the conclusion of the protocol. That is, $\Bias(s)$ captures that advantage that any unbounded receiver has in guessing the bit $s$ at the conclusion of the protocol. Our goal is to show that $\Bias(r_{1\oplus b}) = \negl(\secp)$. Towards showing this, we prove the following claim.

    \begin{claim}
        Consider running the commitment scheme extractor on all commitment transcripts $\{\tau_{\Com,i}\}_{i \in [2\secp]}$ to obtain $\{(x_{0,i},x_{1,i},z_i)\}_{i \in [2\secp]}$. Define 

         \[\ket{\psi_i} \coloneqq \frac{1}{\sqrt{2}}\left(\ket{0,x_{0,i}} + \ket{1,x_{1,i}}\right)\] and let $\ket{\psi_i'}$ be the state obtained by the sender after the $i$'th CSG protocol (note that these states are unentangled with each other). Then with probability $1-\negl(\secp)$, either (1) one of the sender's measurements fails, or (2) for at least $7\secp/4$ indices $i \in [2\secp]$, it holds that $|\bra{\psi_i}\ket{\psi_i'}|^2 \geq 1-1/25$.
        
    \end{claim}

    \begin{proof}
        Suppose that condition (2) does not hold. We will show that this implies that with probability $1-\negl(\secp)$, one of the sender's measurements fails, which suffices to show the claim. 
        
        First note that due to the extractability property of the commitment, there is only $\negl(\secp)$ probability that the receiver can open any commitment $i$ to a value other than the value $(x_{0,i},x_{1,i},z_i)$ that was extracted. Thus, with all but $\negl(\secp)$ probability, the sender attempts to project each $\ket{\psi_i'}$ for $i \in T$ onto the state $\ket{\psi_i}$ defined above. Since $T$ is sampled as a uniformly random subset of $[2\secp]$ of size $\secp$, a straightforward tail bound shows that, except with $\negl(\secp)$ probability, there will be $\Omega(\secp)$ many indices $i$ such that $i \in T$ and $|\bra{\psi_i}\ket{\psi_i'}|^2 < 1-1/25$. For each such $i$, the probability that the sender's measurments fails is at least $1/25$, and these probabilities are independent. Thus the overall probability of \emph{not} failing is at most $(1-1/25)^{\Omega(\secp)} \leq e^{-\Omega(\secp)} = \negl(\secp)$.
    \end{proof}

    Now, in the case that one of the sender's measurements fails, they output uniformly random bits $r_0,r_1$, meaning $\Bias(r_{1\oplus b}) = 0$. If not, the above claim shows that it suffices to consider the scenario where condition (2) holds, with only a $\negl(\secp)$ difference in the final bound on $\Bias(r_{1\oplus b})$.
    
    In this case, let $S \subseteq \overline{T}$ be the set of indices such that (1) $|\bra{\psi_i}\ket{\psi_i'}|^2 \geq 1-1/25$ and (2) $b = b_i \oplus x_{0,i} \oplus x_{1,i}$. By the claim above and the definition of $b$, we know that $|S| \geq \secp/4$. For each $i \in S$, we have that $r_{1\oplus b,i} = y_i \oplus (1\oplus b \oplus b_i) \cdot c_i = y_i \oplus (1\oplus x_{0,i} \oplus x_{1,i}) \cdot c_i$, where $(c_i,y_i)$ are obtained by measuring $\ket{\psi'_i}$ in the standard basis. Note that if instead, $(c_i,y_i)$ were obtained by measuring 

    \[\ket{\psi_i} = \frac{1}{\sqrt{2}}\left(\ket{0,x_{0,i}} + \ket{1,x_{1,i}}\right)\]

    in the standard basis, then $r_{1\oplus b,i}$ would be a uniformly random bit (even conditioned on $\state_R$). To see this, note that in the case that $x_{0,i} = x_{1,i}$ we have that $r_{1\oplus b,i} = y_i \oplus c_i$ where $y_i$ is fixed and $c_i$ is a uniformly random bit, while in the case that $x_{0,i} \neq x_{1,i}$ we have that $r_{1\oplus b,i} = y_i$ where $y_i$ is a uniformly random bit. Thus, by Gentle Measurement (\cref{lemma:gentle-measurement}), we have that the total variation distance between $r_{1\oplus b, i}$ and a uniformly random bit is $\leq 2\sqrt{1/25} = 2/5$, which implies that $\Bias(r_{1\oplus b,i}) \leq 4/5$. Finally, noting that each $r_{1\oplus b,i}$ is an independent random variable (since the states $\ket{\psi'_i}$ are all unentangled), we have that \[\Bias(r_{1\oplus b}) \leq \prod_{i \in S}\Bias(r_{1\oplus b,i}) \leq (4/5)^{\secp/4} = \negl(\secp),\] which completes the proof.
    
\end{proof}

Next, we argue that combining this result with  \cite[Theorem 3.1]{10.1007/978-3-031-15979-4_6} allows us to conclude the following.

\begin{corollary}
    Perfectly correct OSP does not exist in the quantum random oracle model.
\end{corollary}

\begin{proof}
    First, it is straightforward to verify that when the protocol in \cref{fig:OT-from-OSP} is instantiated with a perfectly correct claw-state generator (which can be constructed from a perfectly-correct OSP), then the resulting protocol is a perfectly-correct OT with search security against an unbounded receiver. 

    Next, we confirm that our notion of perfectly-correct OT (with a classical receiver) implies perfectly-correct key agreement (with one classical party). The key agreement protocol will have Alice sample a random $b \gets \{0,1\}$ and act as the OT receiver in a protocol with Bob, obtaining string $r_b$. Then Alice sends her bit $b$ to Bob, and their shared key is defined to be $r_b$. Any QPT eavesdropper that can output $r_b$ with noticeable probability can be used to break the search security of the OT. Indeed, by security against a QPT sender, such an eavesdropper will \emph{also} be able to output $r_{1-b}$ with noticeable probability if the final bit of the transcript they are given is flipped to $1-b$. Thus, an adversarial receiver can first run the OT protocol honestly with the sender to obtain $(b,r_b)$, and then run the eavesdropper on $1-b$ to obtain $r_{1-b}$ with noticeable probability.

    Finally, we appeal to \cite[Theorem 3.1]{10.1007/978-3-031-15979-4_6}, which establishes that perfectly-correct key agreement between one classical and one quantum party does not exist in the quantum random oracle model. Thus, we conclude that perfectly-correct OSP does not exist in the quantum random oracle model. 
\end{proof}

\subsection{PKE from two-round OSP}\label{subsec:PKE-from-OSP}

In this section, we show that two-round OSP implies CPA-secure public-key encryption with a classical key generator / decryptor and a quantum encryptor. In particular, the scheme has classical keys and ciphertexts. First, we define CPA-secure public-key encryption.

\begin{definition}[CPA-secure PKE]\label{def:PKE}
    A CPA-secure public-key encryption scheme with classical key generator and quantum encryptor consists of the following algorithms.

    \begin{itemize}
        \item $\KeyGen(1^\secp) \to (\pk,\sk)$: The PPT key generation algorithm takes as input the security parameter and outputs a public key $\pk$ and secret key $\sk$.
        \item $\Enc(\pk,m) \to \ct$: The QPT encryption algorithm takes as input the public key and a plaintext bit $m \in \{0,1\}$, and outputs a (classical) ciphertext $\ct$.
        \item $\Dec(\sk,\ct) \to m$: The PPT decryption algorithm takes as input the secret key and a ciphertext, and outputs a plaintext $m$.  
    \end{itemize}

    It should satisfy the following properties.

    \begin{itemize}
        \item \textbf{Correctness}: For any $m \in \{0,1\}$,
        \[\Pr\left[\Dec(\sk,\ct) = m : (\pk,\sk) \gets \KeyGen(1^\secp), \ct \gets \Enc(\pk,m)\right] = 1-\negl(\secp).\]
        \item \textbf{Security}: For any QPT adversary $\{\Adv_\secp\}_{\secp \in \bbN}$, 
        \begin{align*}
            &\Big| \Pr\left[\Adv_\secp(\pk,\ct) = 0 : (\pk,\sk) \gets \KeyGen(1^\secp), \ct \gets \Enc(\pk,0)\right]\\ &- \Pr\left[\Adv_\secp(\pk,\ct) = 0 : (\pk,\sk) \gets \KeyGen(1^\secp), \ct \gets \Enc(\pk,1)\right]\Big| = \negl(\secp).
        \end{align*}
        
    \end{itemize}
\end{definition}

We build our scheme from any two-round differentiated-bit CSG with indistinguishability security (defined in \cref{def:CSG} and constructed from two-round OSP in \cref{subsec:encrypted-CNOT}). Letting $(\CSG.\Sen,\CSG.\Rec,\CSG.\Dec)$ be the CSG algorithms, the scheme is constructed as follows.

\begin{itemize}
    \item $\KeyGen(1^\secp)$: Sample $(\msg_S,\state_S) \gets \CSG.\Sen(1^\secp,1)$ and output $\pk = \msg_S$ and $\sk = \state_S$.
    \item $\Enc(\pk,m)$: Given a message bit $m \in \{0,1\}$, sample $(\ket{\psi},\msg_R) \gets \CSG.\Rec(\msg_S)$ and measure $\ket{\psi}$ in the standard basis to obtain bits $(b,x_b)$. Output $\ct \coloneqq (\msg_R,b,m \oplus x_b)$.
    \item $\Dec(\sk,\ct)$: Parse $\ct = (\msg_R,b,m')$, run $(x_0,x_1,z) \gets \CSG.\Dec(\state_S,\msg_R)$, and output $m = m' \oplus x_b$.
\end{itemize}

\begin{theorem}
    The scheme described above satisfies \cref{def:PKE}.
\end{theorem}

\begin{proof}
    First, we show correctness. By the correctness of $\CSG$, the state $\ket{\psi}$ sampled by $\Enc$ is (negligibly close to) $\frac{1}{\sqrt{2}}(\ket{0,x_0} + (-1)^z\ket{1,x_1})$, where the bits $(x_0,x_1,z)$ can be recovered by running $(x_0,x_1,z) \gets \CSG.\Dec(\state_S,\msg_R)$. Thus, given $b$, the decryptor can determine $x_b$, and unmask $m \oplus x_b$ to recover the message.

    Next, we show security. It suffices to show that $x_b$ is unpredictable, that is, for any QPT adversary $\{\Adv_\secp\}_{\secp \in \bbN}$, it holds that 

    \[\Bigg| \Pr\left[\Adv_\secp(\msg_S,\msg_R,b) = x_b : \begin{array}{r} (\msg_S,\state_S) \gets \CSG.\Sen(1^\secp,1) \\ (\ket{\psi},\msg_R) \gets \CSG.\Rec(\msg_S) \\ (x_0,x_1,z) \gets \CSG.\Dec(\state_S,\msg_R) \\ b \gets \{0,1\}\end{array}\right] - \frac{1}{2}\Bigg| = \negl(\secp).\]

    Indeed, note that in the real scheme, $b$ is obtained by measuring $\ket{\psi}$ in the standard basis, but since $\ket{\psi}$ is (negligibly close to) a uniform superposition over $(0,x_0)$ and $(1,x_1)$, we can imagine just sampling a random bit $b \gets \{0,1\}$, and $\Adv_\secp$ will have negligibly close to the same advantage.

    Now, suppose there exists $\Adv_\secp$ that has noticeable advantage in the above game. We use such an adversary to break the indistinguishability security of the $\CSG$. The $\CSG$ adversary will do the following:
    \begin{itemize}
        \item Receive $\msg_S$ from its challenger.
        \item Run $(\ket{\psi},\msg_R) \gets \CSG.\Rec(\msg_S)$.
        \item Sample $b \gets \{0,1\}$ and run $\Adv_\secp(\msg_S,\msg_R,b)$ to obtain a guess for $x_b$.
        \item Measure $\ket{\psi}$ in the standard basis to obtain $(b',x_{b'})$.
        \item If $b = b'$, output a random bit, and otherwise, if $b \neq b'$, output $x_b \oplus x_{b'}$.
    \end{itemize}

    Note that with probability 1/2, the $\CSG$ adversary makes a uniformly random guess, and otherwise, the $\CSG$ makes a guess for $x_0 \oplus x_1$ with noticeable advantage, due to the guarantee on $\Adv_\secp$. Thus the $\CSG$ adversary has a noticeable advantage in breaking the indistinguishability security of the $\CSG$.
     
\end{proof}

\ifsubmission
\bibliographystyle{plain}
\else
\bibliographystyle{alpha}
\fi

\bibliography{abbrev3,crypto,main}

\appendix
\section{Adaptation of the proof from \cite{NZ23}}\label{sec:NZ23}

In this section, we overview how the proof of \cref{thm:main-NZ23} goes, following arguments made in \cite{NZ23}. Fix some computationally nonlocal strategy $\mathscr{C}$ that associates to every $G \in \cG[\mathsf{NO}]$ a set \[\left(\ket{\psi_G},\left\{A^{q_A}_{s_A,G}\right\}_{s_A}, \left\{B^{q_B}_{s_B,G}\right\}_{s_B}\right),\]
and fix any sequence $\{G[H]_\secp\}_\secp$ where $G[H]_\secp \in \cG[\mathsf{NO}]_\secp$ for each $\secp \in \bbN$. From now on, we will drop the parameterization by $\secp$ and refer to a single game $G[H]$ parameterized by a Hamiltonian $H$, as well as a fixed strategy $(\ket{\psi},\{A^{q_A}_{s_A}\}_{s_A},\{B^{q_B}_{s_B}\}_{s_B})$, when we really mean an infinite sequence of games, Hamiltonians and strategies. To prove the theorem, we must show that

\[\E_{(q_A,q_B) \gets Q}\sum_{s_A,s_B}V(q_A,q_B,s_A,s_B)\bra{\psi}A^{q_A,\dagger}_{s_A}B^{q_B}_{s_B}A^{q_A}_{s_A}\ket{\psi} \leq \omega_C + \negl(\secp).\]

Following \cite{NZ23}, we now introduce some notation. First, for any observable $O$, we define $\langle O \rangle \coloneqq \bra{\psi}O\ket{\psi}$, where $\ket{\psi}$ is the state defined by the strategy we fixed above. Next, since $q_B \in \{0,1\}$ is just a single bit, and $s_B \in \{0,1\}^\secp$ is a $\secp$-bit string, we will define \[\{Z_\gamma\}_{\gamma \in \{0,1\}^\secp} \coloneqq \{B^0_{s_B}\}_{s_B},~~~,\{X_\gamma\}_{\gamma \in \{0,1\}^\secp} \coloneqq \{B^1_{s_B}\}_{s_B},\] and assume wlog that there exist unitaries $U_Z,U_X$ such that \[\{Z_\gamma\}_\gamma = \{U_Z^\dagger (\ketbra{\gamma}{\gamma} \otimes I)U_Z\}_\gamma, ~~~ \{X_\gamma\}_\gamma = \{U_X^\dagger (\ketbra{\gamma}{\gamma} \otimes I)U_X\}_\gamma.\] Finally, we define sets of binary observables $\{Z(a)\}_{a \in \{0,1\}^\secp}, \{X(a)\}_{a \in \{0,1\}^\secp}$ as follows:

\begin{align*}
    &Z(a) \coloneqq \sum_{\gamma} (-1)^{a \cdot \gamma} U_Z^\dagger (\ketbra{\gamma}{\gamma} \otimes I)U_Z, \\
    &X(a) \coloneqq \sum_{\gamma} (-1)^{a \cdot \gamma} U_X^\dagger (\ketbra{\gamma}{\gamma} \otimes I)U_X.
\end{align*}

Now, we adapt several lemmas from \cite{NZ23}.

\begin{lemma}[Adaptation of Lemma 36 from \cite{NZ23}]\label{lemma:BQP-anticommutation}
   Suppose the strategy succeeds in the CHSH subtest with probability at least $\omega_\CHSH - \epsilon$. Then
   \[\E_{\substack{(a,b) \gets D^1_Q, \\ q_A \coloneqq (\CHSH,(a,b,0))}}\sum_{s_A}\langle A^{q_A,\dagger}_{s_A} \cdot |\{Z(a),X(b)\}|^2 \cdot A^{q_A}_{s_A} \rangle \leq O(\epsilon).\]
\end{lemma}

\begin{proof}
    Following \cite{NZ23}, for any fixed $a,b$, this can be seen as an instance of a computationally nonlocal strategy applied to the CHSH game. Thus, the claim follows from \cref{thm:CHSH-anticommutation}.  
\end{proof}

\begin{lemma}[Adaptation of Lemma 37 from \cite{NZ23}]\label{lemma:BQP-commutation}
    Suppose the strategy succeeds in the commutation subtest with probability at least $1-\epsilon$. Then
    \[\E_{\substack{(a,b) \gets D^0_Q, \\ q_A \coloneqq (\Comm,(a,b))}}\sum_{s_A}\langle A^{q_A,\dagger}_{s_A} \cdot |[Z(a),X(b)]|^2 \cdot A^{q_A}_{s_A} \rangle \leq O(\epsilon).\]
\end{lemma}

\begin{proof}
    Again, for any fixed $a,b$, this can be seen as an instance of a computationally nonlocal strategy applied to the commutation game described in \cite[Section 3]{NZ23}. Thus, this follows from \cite[Lemma 23]{NZ23}, which analyzes the commutation game. Since this analysis does not use the blindness of QFHE at all (which is not required because the commutation game has no Alice question), there is no change to the proof in our setting.
\end{proof}

\begin{lemma}[Adaptation of Lemma 38 from \cite{NZ23}]\label{lemma:subtests}
    Suppose the strategy succeeds in the CHSH subtest with probability at least $\omega_\CHSH = \epsilon$, and in the commutation subtest with probability at least $1-\epsilon$. Then
    \[\E_{\substack{(a,b) \gets D_Q, \\ q_A \coloneqq \Tel}}\sum_{s_A}\langle A^{q_A,\dagger}_{s_A} \cdot |(-1)^{a \cdot b}Z(a)X(b)-X(b)Z(a)|^2 \cdot A^{q_A}_{s_A} \rangle \leq O(\epsilon) + \negl(\secp).\]
\end{lemma}

\begin{proof}
    Here, we crucially use the fact that the strategy is computationally nonlocal to switch the Alice questions in the above lemmas to $\Tel$. Indeed, this lemma is implied by \cref{lemma:BQP-anticommutation} and \cref{lemma:BQP-commutation} by following the proof of \cite[Lemma 38]{NZ23}, where equations (192) and (195) follow in our setting from the fact that the strategy is computationally nonlocal.
\end{proof}

\begin{lemma}[Adaptation of Lemma 39 from \cite{NZ23}]\label{lemma:conditioned}
    For any $u_1,u_2 \in \{0,1\}$, it holds that
    \[\E_{\substack{(a,b = (e_i + e_j)) \gets D_Q, \\ q_A \coloneqq \Tel}}\sum_{\substack{s_A: (s_A)_i = u_1,\\ (s_A)_j = u_2}}\langle A^{q_A,\dagger}_{s_A} \cdot |(-1)^{a \cdot b}Z(a)X(b)Z(a)-X(b)| \cdot A^{q_A}_{s_A} \rangle \leq O(\epsilon^{1/2}) + \negl(\secp).\]
\end{lemma}

\begin{proof}
    This is implied by \cref{lemma:subtests} by following the proof of \cite[Lemma 39]{NZ23} with no changes.
\end{proof}

Next, we import the following definitions.

\begin{itemize}
    \item \textbf{SWAP isometry.} Let $\cH_\cY$ and $\cH_\cZ$ be two copies of $(\mathbb{C}^2)^{\otimes \secp}$. The $\secp$-qubit SWAP isometry $V : \cH_\cB \to \cH_\cB \otimes \cH_\cY \otimes \cH_\cZ$ is defined by the following expression:
    \[V\ket{\phi} = \left(\frac{1}{2^{\secp}}\sum_{u,v \in \{0,1\}^\secp}Z(u)X(v) \otimes I \otimes \sigma_Z(u)\sigma_X(v)\right)\ket{\phi}\ket{+}^{\otimes \secp}.\]
    \item \textbf{Expected verifier outcomes.} Let $H_X$ be $H$ restricted to $XX$ terms. Let $\hat{\mathbb{E}}[H_X]$ be the expected value of the meausurement outcome computed by the verifier in a teleport round, conditioned on 1) $w = q_B$, so the verifier perform an energy check instead of accepting automatically, and 2) the verifier choosing an XX term to check. Then

    \[\hat{\mathbb{E}}[H_X] = \sum_{u_1,u_2 \in \{0,1\}}(-1)^{u_1 + u_2}\E_{\substack{(b = e_i + e_j) \gets D_X, \\ q_A \coloneqq \Tel}}\sum_{\substack{s_A : (s_A)_i = u_1, \\ (s_A)_j = u_2}}\langle A^{q_A,\dagger}_{s_A}X(b)A^{q_A}_{s_A}\rangle.\]

    Define $\hat{\mathbb{E}}[H_Z]$ analogously, which gives

    \[\hat{\mathbb{E}}[H_Z] = \sum_{v_1,v_2 \in \{0,1\}}(-1)^{v_1 + v_2}\E_{\substack{(a = e_i + e_j) \gets D_Z, \\ q_A \coloneqq \Tel}}\sum_{\substack{s_A : (s_A)_{\secp+i} = v_1, \\ (s_A)_{\secp+j} = v_2}}\langle A^{q_A,\dagger}_{s_A}Z(a)A^{q_A}_{s_A}\rangle.\]

\end{itemize}

\begin{lemma}[Adaptation of Lemmas 43 and 44 from \cite{NZ23}]\label{lemma:extracted-state}
    Define \[\rho_{s_A} \coloneqq \Tr_{\cB,\cZ}[VA^\Tel_{s_A}\ketbra{\psi}{\psi}A^{\Tel,\dagger}_{s_A}V^\dagger].\]
    Then, assuming the strategy succeeds with probability $\omega_\CHSH-\epsilon$ in the CHSH subtest and with probability $1-\epsilon$ in the commutation subtest,
    \[\sum_{u_1,u_2}(-1)^{u_1 + u_2}\sum_{\substack{s_A: (s_A)_i = u_1,\\ (s_A)_j = u_2}}\E_{b \gets D_X}\Tr[\sigma_X(b)\rho_{s_A}] \approx_{O(\epsilon^{1/2}) + \negl(\secp)} \hat{\mathbb{E}}[H_X].\]
    Moreover, 
    \[\sum_{v_1,v_2}(-1)^{v_1 + v_2}\sum_{\substack{s_A: (s_A)_{\secp+i} = v_1,\\ (s_A)_{\secp+j} = v_2}}\E_{a \gets D_Z}\Tr[\sigma_Z(a)\rho_{s_A}] = \hat{\mathbb{E}}[H_Z].\]
\end{lemma}

\begin{proof}
    This is implied by \cref{lemma:conditioned} by following the proofs of \cite[Lemma 43]{NZ23} and \cite[Lemma 44]{NZ23} with no changes.
\end{proof}

\begin{lemma}[Adaptation of Lemma 45 from \cite{NZ23}]\label{lemma:Hamiltonian}
    Assuming the strategy succeeds with probability $\omega_\CHSH-\epsilon$ in the CHSH subtest and with probability $1-\epsilon$ in the commutation subtest, there exists a state $\rho$ such that 
    \begin{align*}
        &\E_{a \gets D_Z} \Tr[\rho_Z(a)\rho] = \hat{\mathbb{E}}[H_Z], \\
        &\E_{b \gets D_X} \Tr[\rho_X(b)\rho] \approx_{O(\epsilon^{1/2})+\negl(\secp)} \hat{\mathbb{E}}[H_X].
    \end{align*}
\end{lemma}

\begin{proof}
    This is implied by \cref{lemma:extracted-state} by following the proof of \cite[Lemma 45]{NZ23} with no changes.
\end{proof}

Now, following analysis in the proof of \cite[Theorem 46]{NZ23} and assuming for contradiction that the strategy succeeds with probability greater than 
\begin{align*}
&\frac{1}{2}(1-\kappa)(1+\omega_{\CHSH}) + \kappa(1-\frac{1}{4}\alpha)-\frac{1}{8}\kappa(\beta-\alpha)\\ = &\frac{1}{2}(1-\kappa)(1+\omega_{\CHSH}) + \kappa(1-\frac{1}{4}\beta)+\frac{1}{8}\kappa(\beta-\alpha),
\end{align*}

we can conclude that for an appropriate choice of $\kappa = \Theta((\beta-\alpha)^2)$, \cref{lemma:Hamiltonian} implies that there exists a state $\rho$ such that $\Tr[H\rho] > \beta$, which gives a contradiction. 
\section{Inefficiently extractable commitments}\label{sec:classical-com}

In this section, we define (post-quantum) classical inefficiently-extractable commitments, which can be constructed from any post-quantum one-way function \cite{10.1007/BF00196774}. 


\begin{definition}[Inefficiently-extractable commitment]
    An inefficiently-extractable commitment between a classical committer and classical receiver consists of an interaction \[\left(\state_\Com,\state_\Rec,\tau\right) \gets \langle \Com(1^\secp,b),\Rec(1^\secp)\rangle,\] where $\tau$ is the (classical) transcript of interaction produced by the protocol, along with algorithms $(\Open,\Ver)$ with the following syntax.
    \begin{itemize}
        \item $\Open(\state_\Com) \to (b,w)$ is a PPT algorithm that takes as input the committer's state $\state_\Com$ and produces a bit $b$ and opening information $w$.
        \item $\Ver(\state_\Rec,b,w) \to \{\top, \bot\}$ is a PPT algorithm that takes as input the receiver's state $\state_\Rec$, a bit $b$, and opening information $w$, and either accepts or rejects.
    \end{itemize}
    \emph{Correctness} requires that for any $b \in \{0,1\}$,
    \[\Pr\left[\Ver(\state_\Rec,b,w) = \top : \begin{array}{r}(\state_\Com,\state_\Rec,\tau) \gets \langle \Com(1^\secp,b),\Rec(1^\secp)\rangle \\ (b,w) \gets \Open(\state_\Com)\end{array}\right] = 1-\negl(\secp).\]
    The commitment satisfies (post-quantum) \emph{computational hiding} if for any QPT adversarial receiver $\{\Adv_\secp\}_{\secp \in \bbN}$,
    \begin{align*}
    &\bigg|\Pr\left[b_\Adv = 0 : (\state_\Com,b_\Adv) \gets \langle \Com(1^\secp,0),\Adv_\secp\rangle\right] \\ &- \Pr\left[b_\Adv = 0 : (\state_\Com,b_\Adv) \gets \langle \Com(1^\secp,1),\Adv_\secp\rangle\right]\bigg| = \negl(\secp).
    \end{align*}
    The commitment is \emph{inefficiently extractable} if there exists an unbounded extractor $\Ext$ such that for any unbounded adversarial committer $\Adv$,
    \begin{align*}
        \Pr\left[\Ver(\state_\Rec,1-b,w) = \top : \begin{array}{r} (\state_\Adv,\state_\Rec,\tau) \gets \langle \Adv,\Rec(1^\secp)\rangle \\ b \gets \Ext(\tau) \\ w \gets \Adv(\state_\Adv,1-b)\end{array}\right] = \negl(\secp).
    \end{align*}
\end{definition}

\end{document}